\documentclass[a4paper, USenglish, numberwithinsect, cleveref, autoref, thm-restate]{jsat}

\pdfoutput=1 %uncomment to ensure pdflatex processing (mandatatory e.g. to submit to arXiv)
\hideJSAT  %uncomment to remove references to LIPIcs series (logo, DOI, ...), e.g. when preparing a pre-final version to be uploaded to arXiv or another public repository

\title{Extending CDCL to disjunctions of parity equations} %TODO Please add

\author{Paul Beame}{University of Washington, USA \and \url{https://homes.cs.washington.edu/~beame/site/} }{beame@cs.washington.edu}{https://orcid.org/0000-0002-2666-3545}{}%TODO mandatory, please use full name; only 1 author per \author macro; first two parameters are mandatory, other parameters can be empty. Please provide at least the name of the affiliation and the country. The full address is optional. Use additional curly braces to indicate the correct name splitting when the last name consists of multiple name parts.

\author{Glenn Sun}{University of Washington, USA \and \url{https://www.glennsun.com}}{glennsun@cs.washington.edu}{https://orcid.org/0000-0001-6510-2508}{}

\authorrunning{P. Beame and G. Sun} %TODO mandatory. First: Use abbreviated first/middle names. Second (only in severe cases): Use first author plus 'et al.'

\Copyright{Paul Beame and Glenn Sun} %TODO mandatory, please use full first names. LIPIcs license is "CC-BY";  http://creativecommons.org/licenses/by/3.0/

 %TODO mandatory; please add comma-separated list of keywords

\category{} %optional, e.g. invited paper

\relatedversiondetails{Conference Version}{https://doi.org/10.4230/LIPIcs.SAT.2026.5} 
\relatedversiondetails{Extended Version}{https://arxiv.org/abs/2605.15002v1} %linktext and cite are optional

\supplementdetails[subcategory={Source Code}]{Software}{https://github.com/glenn-sun/xorcle} %linktext, cite, and subcategory are optional

\funding{This work was supported by NSF grant CCF-2422205.} %optional, to capture a funding statement, which applies to all authors. Please enter author specific funding statements as fifth argument of the \author macro.

\nolinenumbers %uncomment to disable line numbering

\EventNoEds{2}
\EventLongTitle{This article is submitted to JSAT}
\EventShortTitle{JSAT YYYY}
\EventAcronym{JSAT}
\EventYear{YYYY}
\EventDate{}
\EventLocation{Little Whinging, United Kingdom}
\EventLogo{}
\SeriesVolume{VV(I))} % Need to match!
\ArticleNo{0(0)} % Need to match!
\Submitted{MMMM YYYY}
\Revised{MMMM YYYY}
\usepackage{bussproofs} % for prooftree environment
\usepackage{mathtools} % for \vdotswithin

\usepackage[ruled]{algorithm}
\usepackage[noEnd=true, indLines=false, rightComments=false, beginComment={\qquad$\triangleright$~}]{algpseudocodex}
\algrenewcommand\algorithmicrequire{\textbf{Input:}}
\algrenewcommand\algorithmicensure{\textbf{Output:}}

\usepackage{longtable}

\DeclareMathOperator{\Res}{Res}
\DeclareMathOperator{\CDCL}{CDCL}
\DeclareMathOperator{\LRUP}{LRUP}
\DeclareMathOperator{\linspan}{span}
\DeclareMathOperator{\width}{width}

\DeclareMathOperator{\red}{red}
\DeclareMathOperator{\pos}{pos}
\DeclareMathOperator{\Sol}{Sol}
\DeclareMathOperator{\Conseq}{Conseq}

\newcommand{\eq}[1]{\mbox{\ensuremath{[#1]}}}

\newcommand{\linneg}[1]{\mathrm{neg}(#1)}
\newcommand{\highlight}[1]{\colorbox{black!5}{$#1$}}
\newcommand{\solverName}{{\sf{Xorcle}}}

\newcommand{\reason}[1]{\RightLabel{\small\textnormal{\sffamily :\ #1}}}

\usepackage{crossreftools}
\usepackage{tabularray}

\crefname{claim}{Claim}{Claims}
\theoremstyle{plain}
\newtheorem{fact}[theorem]{Fact}
\crefname{fact}{Fact}{Facts}
\theoremstyle{definition}
\newtheorem{convention}[theorem]{Convention}
\crefname{convention}{Convention}{Convention}
\crefformat{assumption}{#2#1#3}
\begin{document}

\maketitle

\begin{abstract}
Because CDCL produces proofs in the Resolution proof system, problems provably hard for Resolution are also provably hard for CDCL. Exponentially shorter proofs can sometimes be found using stronger proof systems such as $\Res(\oplus)$, a generalization of Resolution to \emph{XNF formulas}, whose constraints are disjunctions of parity equations, called \emph{linear clauses}, such as $\lnot (x \oplus y) \lor (y \oplus z)$. While some modern solvers like CryptoMiniSat reason on Boolean clauses with separate parity equations in a system called CNF-XOR, reasoning about more general linear clauses is less explored.

We present $\CDCL(\oplus)$, a generalization of CDCL to XNF formulas that includes branching on arbitrary parity equations, linear-algebraic reasoning during unit propagation, and learning linear clauses through conflict analysis. Our key technical tool is the \emph{affine resolution} rule, a generalization of the classical resolution rule that allows us to easily translate resolution-based subroutines, such as 1-UIP clause learning and variable elimination. We further develop a theory of UIPs to show not just an algorithmic similarity, but the same theoretical guarantee that we learn an asserting clause with the largest possible backjump.  Mirroring the classical relationship between CDCL and Resolution, we prove a bidirectional connection between $\CDCL(\oplus)$ and $\Res(\oplus)$: $\CDCL(\oplus)$ not only produces $\Res(\oplus)$ proofs, but also polynomially simulates $\Res(\oplus)$ given nondeterministic decisions and restarts.

We provide an implementation of $\CDCL(\oplus)$ called \solverName{}, which cleanly adapts over a dozen major classical CDCL heuristics. We introduce the CSIDS decision heuristic, which selects decisions from activity-ranked clauses, and the $\LRUP(\oplus)$ format for proof logging, which is supported by all of our heuristics. On a suite of cryptographic benchmarks, \solverName{} is generally competitive with or outperforms existing solvers such as Kissat, CaDiCaL, and CryptoMiniSat, and is significantly faster than the concurrently developed XNF solver Xorricane. Additionally, \solverName{}'s reasoning and search alone (i.e., without detecting XORs or symmetry) solves CNF versions of Tseitin formulas, perfect matching formulas, and bijective pigeonhole formulas with $n+1$ pigeons into $n$ holes using apparently polynomial running time, indicating that the theoretical advantages of $\CDCL(\oplus)$ can indeed be realized empirically. 

\end{abstract}

\section{Introduction}
Although conflict-driven clause learning (CDCL) is the dominant method today for SAT solving, it also has known limitations. Because CDCL produces proofs of unsatisfiability in the Resolution proof system, all problems with exponential lower bounds for Resolution proof size must also be hard for CDCL~\cite{DBLP:journals/jair/BeameKS04}.
These limitations have led some solvers to incorporate ideas and fragments of stronger proof systems. One example is CryptoMiniSat~\cite{DBLP:conf/sat/SoosNC09, DBLP:conf/sat/Soos10, DBLP:conf/aaai/SoosM19, DBLP:conf/cav/SoosGM20}, which couples CDCL with Gaussian elimination over $\mathbb{F}_2 = \{0, 1\}$ to support constraints that involve XOR, called parity equations, such as $\lnot (x \oplus y)$. Detection of parity equations implicit in CNF formulas allows CryptoMiniSat to efficiently solve Tseitin formulas, which are known to be hard for Resolution and CDCL \cite{DBLP:journals/jacm/Urquhart87}. Because CryptoMiniSat considers the conjunction of a CNF formula and parity equations, it is often called a \emph{CNF-XOR solver}.

Despite CryptoMiniSat's empirical success in dealing with parity equations arising from many application domains, such as cryptanalysis and model counting~\cite{DBLP:conf/sat/SoosNC09, DBLP:conf/sat/Soos10, DBLP:conf/aaai/SoosM19, DBLP:conf/cav/SoosGM20}, the theory of its reasoning ability is less understood. In proof complexity, the preferred way to extend Resolution to include parity reasoning has been the system $\Res(\oplus)$ (``Res-parity'') \cite{DBLP:conf/mfcs/ItsyksonS14, DBLP:journals/apal/ItsyksonS20}, which not only allows standard Boolean clauses in conjunction with parity equations, but generalizes both into \emph{linear clauses}: disjunctions of parity equations, such as $\lnot(x \oplus y) \lor (y \oplus z)$. In particular, Boolean clauses are just linear clauses where each parity equation has one variable, and parity equations are just linear clauses of width 1. A formula that is a conjunction of linear clauses is said to be in \emph{XOR-OR-AND normal form} (XNF)~\cite{DBLP:journals/mics/AndraschkoDK24}. 

The CDCL and Gaussian elimination core of CryptoMiniSat can be encoded in the $\Res(\oplus)$ proof system, just like how pure CDCL solvers' reasoning can be encoded in Resolution. However, the relationship between CDCL and Resolution also goes the other way, in that \emph{any} Resolution proof of unsatisfiability can be realized by a real CDCL execution for an appropriate choice of decisions and restarts, with a polynomial amount of overhead~\cite{DBLP:journals/ai/PipatsrisawatD11}. This is known as $p$-simulation~\cite{DBLP:journals/jsyml/CookR79}, and gives theoretical justification to the performance of CDCL. No such result is known for CryptoMiniSat and $\Res(\oplus)$, and it is hard to imagine what that might even look like, given that CNF-XOR formulas used by CryptoMiniSat are only a fragment of the XNF formulas in $\Res(\oplus)$.

In this work, we introduce $\CDCL(\oplus)$ (``CDCL-parity''), a CDCL-style algorithm that works natively on XNF formulas instead of CNF, \emph{without} the need for special handling of parity equations outside of core search. We show that $\CDCL(\oplus)$ is theoretically powerful, $p$-simulating $\Res(\oplus)$ in the same way that CDCL $p$-simulates Resolution, and also that $\CDCL(\oplus)$ is empirically effective, with our implementation \solverName{} (XOR Clause LEarning) being generally competitive with or outperforming CryptoMiniSat on a range of cryptographic benchmarks. \solverName{} not only solves Tseitin formulas in apparently polynomial time, but also other hard problems for Resolution such as perfect matching formulas and bijective pigeonhole principle formulas \cite{DBLP:journals/tcs/Haken85, DBLP:journals/jcss/Razborov04, DBLP:journals/fuin/ItsyksonOSS16, DBLP:journals/theoretics/AustrinR22}, which even CryptoMiniSat continues to struggle on. All these problems are theoretically easy in $\Res(\oplus)$ \cite{DBLP:conf/mfcs/ItsyksonS14, DBLP:journals/apal/ItsyksonS20}, so \solverName{}'s performance here indicates that at least some of $\CDCL(\oplus)$'s theoretical advantages can indeed be realized empirically.

\subsection{Our contributions}

Our algorithm $\CDCL(\oplus)$ generalizes CDCL to XNF formulas by branching on arbitrary parity equations, performing linear-algebraic reasoning during unit propagation (\cref{ssec:up}), and learning linear clauses through conflict analysis (\cref{ssec:learning}). The clause learning algorithm follows naturally from a generalization of the resolution rule to XNF formulas that we call \emph{affine resolution} (\cref{ssec:aff-res}). We also describe how to use affine resolution to generalize variable elimination (\cref{ssec:elim}). 

A key technical hurdle is that various aspects of classical CDCL are often described using \emph{implication graphs}, which draw edges from antecedent literals to propagated literals, but it is unclear how to capture XNF reasoning with a graph structure. Instead, our definitions and algorithms emphasize sharing the same semantics and optimality as classical CDCL. This allows us to claim, in some sense, that our definitions and algorithms are canonical translations of classical CDCL concepts. For reference, we review semantic characterizations of classical CDCL in \cref{sec:classic}. Then, following our description of the main $\CDCL(\oplus)$ algorithm in \cref{sec:cdcl-par}, the discussion of theoretical properties in \cref{sec:theo} and the discussion of heuristics and experimental evaluation in  \cref{sec:empr} can be read independently. 

In \cref{sec:theo}, we first prove that the conflict analysis in $\CDCL(\oplus)$ finds the asserting clause that gives the largest possible backjump among all asserting clauses (\cref{ssec:backjump}), mirroring~\cite{DBLP:conf/sat/AudemardBHJS08}. This involves developing a semantic theory of unique implication points (UIPs) and extending classical definitions from~\cite{DBLP:journals/tc/Marques-SilvaS99}. Second, we show that the $\Res(\oplus)$ proof system from~\cite{DBLP:conf/mfcs/ItsyksonS14, DBLP:journals/apal/ItsyksonS20} is polynomially equivalent to $\CDCL(\oplus)$ with nondeterministic decisions and restarts (\cref{ssec:sim-res}). This relies on first observing that our affine resolution rule is equivalent to the inference rules of the original formulation of $\Res(\oplus)$. Then, in one direction, our unit propagation is equivalent to a fragment of $\Res(\oplus)$ that we call \emph{input $\Res(\oplus)$}, analogous to~\cite{DBLP:journals/jair/BeameKS04}. In the other direction, we prove that $\CDCL(\oplus)$ with nondeterministic decisions and restarts $p$-simulates $\Res(\oplus)$ by using a small modification to the argument for classical CDCL and Resolution in \cite{DBLP:journals/ai/PipatsrisawatD11}. When our modification is backported to the classical case, it matches the polynomial simulation factor in \cite{DBLP:journals/lmcs/BeyersdorffB23} with a simpler argument, which is a factor $n$ improvement over the original argument in \cite{DBLP:journals/ai/PipatsrisawatD11}.

In \cref{sec:empr}, we describe our implementation 
\solverName{}, which not only implements $\CDCL(\oplus)$ but also translates over a dozen heuristics, demonstrating that many classical CDCL techniques can be readily adapted to the XNF setting (\cref{ssec:heur}). These include watched equations, phase saving, LBD-based clause deletion, local clause minimization, restarts, generalized bounded variable elimination, and more. We make special note of a ``Clause State Independent Decaying Sum'' (CSIDS) decision heuristic that chooses \emph{clauses} by activity, inspired by VSIDS \cite{DBLP:conf/dac/MoskewiczMZZM01} but more directly generalizing the older clause-first decision heuristics BerkMin \cite{DBLP:conf/date/GoldbergN02} and Clause-Move-To-Front \cite{DBLP:conf/hvc/GershmanS05}. By deciding a combination of equations from a chosen clause, \solverName{} makes complicated decisions even on plain CNF input, which is crucial for efficiently solving Tseitin and other formulas without XOR detection. The solver implements proof logging using a proof format we call $\LRUP(\oplus)$ (``LRUP-parity''), and these proofs are compatible with all included heuristics.

We benchmark \solverName{} on seven families of cryptography-related XNF instances introduced alongside the concurrently-developed XNF solver Xorricane \cite{DBLP:phd/dnb/Danner25, DBLP:journals/jair/DannerK26}, in addition to random instances, Tseitin formulas, perfect matching formulas, and bijective pigeonhole formulas (\cref{ssec:eval}). Experimental results show that \solverName{} significantly outperforms Xorricane, and is competitive with Kissat, CaDiCaL, and CryptoMiniSat on most families. Overall, \solverName{} solves more cryptographic instances than any other solver. On the proof-complexity-inspired families, \solverName{}'s running time appears to scale polynomially on Tseitin formulas, perfect matching formulas, and bijective pigeonhole formulas with $n+1$ pigeons into $n$ holes, demonstrating that $\CDCL(\oplus)$'s theoretical power can be empirically realized. No other tested solver manages to scale efficiently on the latter two proof-complexity-inspired families.

\subsection{Related work}
Solvers that attempt to extend beyond the power of Resolution are common. For example, in addition to CryptoMiniSat and other CNF-XOR solvers, there are also pseudo-Boolean solvers~\cite{DBLP:conf/dac/ChaiK03, DBLP:journals/jsat/SheiniS06, DBLP:conf/sat/VinyalsEGGN18} that use ideas from Cutting Planes proofs~\cite{manual:Gomory58, DBLP:journals/dm/Chvatal73a, DBLP:journals/dam/CookCT87}, and some CDCL-based solvers use bounded variable addition~\cite{DBLP:conf/hvc/MantheyHB12}, which dynamically introduces new variables and is thus part of Extended Resolution~\cite{manual:Tseitin68}. For CNF-XOR in particular, there have also been efforts to forgo CDCL entirely and use stochastic local search, resulting in solvers like xnfSAT (which follows an older naming convention that refers to CNF-XOR as XNF) \cite{DBLP:conf/sat/NawrockiLFHB21}.

While XNF solving has received less attention than CNF and CNF-XOR solving, some proposed algorithms for XNF solving exist.
In their original paper defining $\Res(\oplus)$~\cite{DBLP:conf/mfcs/ItsyksonS14}, Itsykson and Sokolov outlined a backtracking search procedure and proved that with nondeterministic decisions, it is equivalent to the tree-like fragment of $\Res(\oplus)$. 
Hor\'{a}\v{c}ek~\cite{DBLP:phd/dnb/Horacek20} defined an algebraic version of unit propagation in a proof system denoted SRES, inspired by Gr\"obner basis algorithms and equivalent to $\Res(\oplus)$, yielding an analogue of DPLL when combined with the algorithm of~\cite{DBLP:conf/mfcs/ItsyksonS14}. Our unit propagation algorithm is an equivalent, alternative description that  more readily highlights parallels to the classical version and is easier to implement. 

Two other proposed algorithms for XNF solving use different methods entirely. Hor\'{a}\v{c}ek and Kreuzer proposed a search through all possible ways to apply SRES inference rules~\cite{DBLP:conf/scsquare/HoracekK18}, but did not provide an empirical evaluation of this idea. Recently, Andraschko, Danner, and Kreuzer created the 2-Xornado solver, which uses a graph-based approach to solve 2-XNF formulas, where each clause is the disjunction of 2 parity equations~\cite{DBLP:journals/mics/AndraschkoDK24}. They prove that XNF satisfiability can be reduced to 2-XNF, much like the classical reduction from SAT to 3-SAT, so that 2-Xornado can be used to solve any XNF formula via a reduction.

In a concurrent work, Danner and Kreuzer show how to use clause learning with XNF formulas~\cite{DBLP:phd/dnb/Danner25, DBLP:journals/jair/DannerK26}, using an algorithm that is similar to ours. Their solver, Xorricane, implements a restricted version of the algorithm that only makes single-literal decisions and only uses single-literal facts in the trail to make inferences. Xorricane is designed for sparse parity equations, storing each equation as a list of the involved variables. Alongside Xorricane, Danner and Kreuzer also created the first larger-scale XNF benchmark set consisting of various cryptography-related instances. 

While $\CDCL(\oplus)$ is similar to the algorithm described by Danner and Kreuzer, our approach has important differences. We present and view the clause learning algorithm as a canonical translation of classical CDCL by translating the resolution rule into the affine resolution rule. This gives us significantly more theoretical insight, allowing us to prove that the algorithm $p$-simulates $\Res(\oplus)$, and also prove that our clause learning algorithm achieves the largest possible backjump. Unlike Xorricane, our solver \solverName{} faithfully implements $\CDCL(\oplus)$, 
%efficiently 
uses dense parity equations for both decisions and inferred constraints, translates many more heuristics, supports proof logging, and is significantly faster than Xorricane on its own benchmark set.

\subsection*{AI disclosure}

The ideas of \cref{sec:prelim}, \cref{ssec:up}, \cref{ssec:aff-res}, \cref{ssec:learning}, and \cref{ssec:sim-res} were produced independently of AI. The mathematical arguments in \cref{ssec:elim} and \cref{ssec:backjump} were first generated by Claude, and the authors refined the presentation and take responsibility for their verification. The authors acted as primary designers for all of \cref{sec:empr}, while \solverName{}'s actual source code and benchmark harnesses were written interactively through Claude Code. All text in this article has been directly written by the authors, and Claude has been used to proofread without direct edit access.

\subsection*{Relationship with preliminary versions}

A preliminary version of this work appeared in SAT 2026 \cite{DBLP:conf/sat/BeameS26}, whose extended version was posted as \href{https://arxiv.org/abs/2605.15002v1}{\texttt{arXiv:2605.15002v1}} \cite{DBLP:journals/corr/abs-2605-15002}, both associated with \solverName{}'s v0.1 release. Apart from reorganization to improve clarity and readability, this article has four major changes. First, \cref{ssec:elim} slightly generalizes the definition of affine resolution to allow an analogue of variable elimination by resolution. Second, the new \cref{ssec:backjump} discusses the optimality of our generalization of 1-UIP clause learning. While it refutes the conjectured optimality condition from the extended version, it proves that our algorithm allows the largest possible backjumps. Third, \solverName{} v0.2 is a complete rewrite with improved engineering, and most translations of heuristics described in \cref{ssec:heur} are new or improved as well. These result in \solverName{} v0.2 having performance improvements of several orders of magnitude over \solverName{} v0.1, and \cref{ssec:eval} on experimental results includes this comparison. Fourth, the evaluation has been updated to include CaDiCaL, Xorricane, and the latest version of CryptoMiniSat, which has received updates that improve its performance on formulas originally expressed in XNF. Two unrealistic synthetic families on which \solverName{} v0.1 performed well have been dropped for a more principled evaluation.

\section{Preliminaries} \label{sec:prelim}

The set $\mathbb{F}_2 = \{0, 1\}$ is the field of two elements. While parity equations in other SAT solving literature are sometimes written as a possibly-negated XOR of variables, we prefer to write them as actual equations, in order to more easily perform linear algebra. We associate such equations with vectors in $\mathbb{F}_2^{n+1}$ by reading off coefficients, ending with the right-hand side, where $n$ is the number of variables. For example, in a system with 3 variables $x$, $y$, and $z$, the following expression involving only $x$ and $y$ could be represented in the depicted ways. 
\begin{equation*}
    \begin{gathered} \lnot (x \oplus y) \\ \text{XOR form} \end{gathered}
    \quad \Leftrightarrow \quad
    \begin{gathered} \eq{x \oplus y = 0} \\ \text{equation form} \end{gathered}
    \quad \Leftrightarrow \quad
    \begin{gathered} (1, 1, 0, 0) \\ \text{$\mathbb{F}_2^{n+1}$ vector form} \end{gathered}
\end{equation*}

Note that within equations, the symbol $\oplus$ denotes XOR on individual Boolean variables (addition in $\mathbb{F}_2$), whereas the symbol $+$ denotes bitwise XOR on equations (addition in $\mathbb{F}_2^{n+1}$). If $f$ and $g$ are true on some assignment of variables, then $f + g$ is also true. The following is an example of addition.
\begin{equation*}
    \begin{gathered} \eq{x \oplus y = 0} + \eq{y \oplus z = 1} = \eq{x \oplus z = 1} \\ \text{equation form} \end{gathered}
    \quad \Leftrightarrow \quad
    \begin{gathered} (1, 1, 0, 0) + (0, 1, 1, 1) = (1, 0, 1, 1) \\ \text{$\mathbb{F}_2^{n+1}$ vector form} \end{gathered}
\end{equation*}
We also adopt the following notations for false, true, and negation.
\begin{equation*}
    \begin{gathered} \bot = \eq{0 = 1} \\ \text{false} \end{gathered} \qquad
    \begin{gathered} \top = \eq{0 = 0} \\ \text{true} \end{gathered} \qquad 
    \begin{gathered} \lnot f = f + \bot \\ \text{negation} \end{gathered}
\end{equation*}

We typically write everything in equation form, but mathematically treat such equations $f$ as $\mathbb{F}_2^{n+1}$ vectors for the purpose of vector space operations like addition, span, and linear independence. For example, we have $\eq{x \oplus y = 0} \in \linspan(\eq{x = 1}, \eq{y = 1})$, and the set $\{\eq{x = 0}, \eq{x = 1}\}$ is linearly independent. 

When we discuss semantics, an assignment to the variables is denoted $\vec{x} \in \mathbb{F}_2^n$, and such an assignment can either \emph{satisfy} or \emph{falsify} logical expressions such as equations. A set of equations is \emph{consistent} if they have a common solution. For example, the set $\{\eq{x = 0}, \eq{x = 1}\}$ is \emph{inconsistent}. We use the symbol $\Rightarrow$ for semantic implication, where sets on both the left and right of $\Rightarrow$ are read as conjunctions. The expression $A \Rightarrow B$ means that all assignments satisfying $A$ also satisfy $B$. We may also read $A \Rightarrow B$ as \emph{$A$ satisfies $B$}, and read $A \Rightarrow \lnot B$ as \emph{$A$ falsifies $B$}. The following are standard facts that relate linear algebra and semantics.

\begin{fact} \label{prop:linalg}
    Let $U$ be a set of equations and let $f$ be another equation.
    \begin{enumerate}
        \item $U$ is inconsistent if and only if $\bot \in \linspan(U)$.
        \item If $U$ is consistent, then $U \Rightarrow f$ if and only if $f \in \linspan(U)$.
    \end{enumerate}
\end{fact}

In order to reduce confusion and allow consistently reading $\linspan(U)$ as the set of semantic consequences of $U$, we will generally try to avoid taking spans of inconsistent sets of equations. The following notation helps us adhere to this goal, which we will use whenever generalizing a classical CDCL procedure that discusses assigned \emph{variables} rather than assigned \emph{literals}. 

\begin{definition}
    Let $U$ be a consistent set of equations, and consider the coset $\bot + \linspan(U) = \{\lnot f \mid f \in \linspan(U)\}$. Define the subspace of \emph{assigned equations} to be
    \begin{equation*}
        \linspan^\pm(U) = \linspan(U) \cup (\bot + \linspan(U)).
    \end{equation*}
    Equivalently,  $\linspan^\pm(U) = \linspan(U, \bot)$.
\end{definition}

By \cref{prop:linalg}, the subspace $\linspan^\pm(U)$ is exactly the set of equations whose truth value is known under $U$, motivating the name \emph{assigned equations}. For example, with $U = \{\eq{x = 0}, \eq{y = 0}\}$, the equation $\eq{x \oplus y = 1}$ will be falsified on all assignments that satisfy $U$, so we say that $\eq{x \oplus y = 1} \in \linspan^\pm(U)$. This is analogous to a Boolean literal $x$ implying that the value of $\lnot x$ is also known, in particular, known to be false. 

When equations replace literals, we get the concept of a \emph{linear clause}. Importantly, linear clauses allow us to take \emph{disjunctions} of equations, something that is not allowed in CNF-XOR systems that treat standard Boolean clauses and XOR expressions separately. 

\begin{definition}
    A \emph{linear clause} (or just ``clause'' when clear from context) is a disjunction of affine equations $C = h_1 \lor \dots \lor h_m$ over $\mathbb{F}_2$. An \emph{XNF formula} is a conjunction of linear clauses. We write $C \equiv D$ if the two clauses are satisfied by the same assignments.
\end{definition}

Linear clauses strictly generalize Boolean clauses, since one can take affine equations with a single variable. An important distinctive feature is that two syntactically distinct linear clauses $C$ and $D$ may actually have $C \equiv D$.  For example, 
\begin{equation*}
    \eq{x = 1} \lor \eq{y = 1} \equiv \eq{x \oplus y = 1} \lor \eq{y = 1},
\end{equation*}
since both are satisfied by the same assignments. One can check that both clauses are satisfied by $(x, y)$ being $(0, 1)$, $(1, 0)$, or $(1, 1)$.

In order to efficiently determine semantic equivalence, the procedure is to negate the clause and take the span. For example, the negation of $C = \eq{x = 1} \lor \eq{y = 1}$ is $\eq{x = 0} \land \eq{y = 0}$. By \cref{prop:linalg}, $\linspan(\eq{x = 0}, \eq{y = 0})$ captures the semantic information of $C$. 

\begin{definition}
    For a linear clause $C = h_1 \lor \dots \lor h_m$, define its \emph{linear negation} to be the subspace of affine equations 
    \begin{equation*}
        \linneg C = \linspan(\lnot C) = \linspan(\lnot h_1, \dots, \lnot h_m).
    \end{equation*}
\end{definition}

\begin{proposition}\label{prop:equiv}
We can determine clause equivalence using linear algebra through the following properties:
    \begin{enumerate}
    \item A clause $C$ is tautological (satisfied by all assignments) if and only if $\bot \in \linneg{C}$. 
    \item For any non-tautological clauses $C$ and $D$, we have $C \equiv D$ if and only if $\linneg C = \linneg D$. 
    \end{enumerate}
\end{proposition}

\begin{proof}
    For the first part, observe that $C$ is by definition tautological if and only if $\linneg{C}$ is falsified by all assignments, i.e., inconsistent. Then, just apply \cref{prop:linalg}(1).
    
    For the second part, the set of assignments falsifying $C$ is exactly the set of common solutions to $\linneg{C}$. Thus, if $\linneg{C} = \linneg{D}$, then they have the same common solutions, hence $C \equiv D$. Conversely, if $C \equiv D$, then $\linneg{C}$ and $\linneg{D}$ have the same set of common solutions. Applying \cref{prop:linalg}(2), we get $\linneg{C} \subseteq \linneg{D}$ and $\linneg{D} \subseteq \linneg{C}$.
\end{proof}

In our implementation, we also store a clause $C$ as a basis of the system of equations $\linneg{C}$. Although this is the opposite of most classical CDCL solvers that store clause literals positively, we store the negated equations because they are more semantically meaningful, and thus easier to use in algorithms.

As the final preliminary notion, we note that it is no longer meaningful to define the width of a clause by the number of disjuncts, as it is possible to have linear dependencies between the equations. Consider the following example.
\begin{align*}
    \linneg{\eq{x = 1} \lor \eq{y = 1} \lor \eq{x \oplus y = 1}} & = \linspan(\eq{x = 0}, \eq{y = 0}, \eq{x \oplus y = 0}) \\
    & = \linspan(\eq{x = 0}, \eq{y = 0})
\end{align*}
Despite the clause syntactically having 3 disjuncts, it is equivalent to a clause with only 2 disjuncts, as witnessed by its linear negation having dimension 2. Thus, we get the following definitions.

\begin{definition}
    Define the \emph{width} of a clause to be $\width(C) = \dim(\linneg{C})$. A \emph{$k$-XNF formula} is an XNF formula where every clause has width at most $k$. A \emph{unit clause} or \emph{unit} is a non-tautological clause of width 0 or 1, i.e., any equation except $\top$, including $\bot$. 
\end{definition}

The word \emph{unit} is the reason that we typically use the letter $U$ for the trail, which is just a list of units. For a consistent trail $U$, we will also frequently want to refer to clauses $C$ that are falsified by the trail $U$. (Such a clause would imply a conflict.) As mentioned earlier, $C$ being falsified by $U$ can be written as $U \Rightarrow \lnot C$. By \cref{prop:linalg}, this is equivalent to $\linneg{C} \subseteq \linspan(U)$. This linear-algebraic form will appear often and should be read as ``$C$ is falsified by $U$.''

\section{The \texorpdfstring{$\CDCL(\oplus)$}{CDCL(⊕)} algorithm} \label{sec:cdcl-par}

The $\CDCL(\oplus)$ algorithm mimics classical CDCL, but with reasoning over XNF. Its basic form is given in \cref{alg:cdcl-parity}, where subroutines for unit propagation, clause learning, and decision heuristics perform XNF reasoning. The algorithm is otherwise identical to classical CDCL, and its correctness can also be argued similarly to the correctness of the classical version. We include restarts in this description of basic $\CDCL(\oplus)$ because they are needed for the main simulation result of \cref{ssec:sim-res}. 

\begin{algorithm}[t]
\caption{Basic $\CDCL(\oplus)$}\label{alg:cdcl-parity}
\begin{algorithmic}[1]
\Require A set of linear clauses $\Delta$
\Ensure Whether $\Delta$ is satisfiable
\State Initialize $U \leftarrow \emptyset$, $\mathcal{R} \leftarrow \emptyset$.
  \Comment{Trail units and corresponding reason clauses}
\While{true}
    \State Run unit propagation on $\Delta$ (\cref{ssec:up}) and update $U$ and $\mathcal{R}$ accordingly.
    \If{the contradiction $\bot$ was propagated}
        \If{$U$ contains a decision}
            \State Find an asserting clause $C$ (\cref{ssec:learning}) and add it to $\Delta$.
            \State Revert $U$ and $\mathcal{R}$ to the asserting level of $C$.
        \Else
            \State \Return ``unsatisfiable''
        \EndIf
    \Else
        \If{$\linspan^\pm(U) \neq \mathbb{F}_2^{n+1}$ (i.e., some equation is not yet assigned)}
            \State Restart or decide any $f$ with
                   $f \notin \linspan^\pm(U)$ (\cref{ssec:heur}) to add to $U$.
        \Else
            \State \Return ``satisfiable'' and the solution to the system of equations $U$
        \EndIf
    \EndIf
\EndWhile
\end{algorithmic}
\end{algorithm}
\newpage

\begin{convention} \label{conv:trail}
    In all discussion about $\CDCL(\oplus)$, we use the following conventions. We have not yet defined the precise relationship between trail units and reason clauses, but like classical CDCL, each propagated unit comes from a reason clause.
\begin{itemize}
    \item Of the trail $U$, the $i$-th trail unit is denoted $f_i$, and the first $i$ trail units are denoted $U_i = \{f_1, \dots, f_i\}$.
    \item Of the reasons $\mathcal{R}$, the reason clause corresponding to $f_i$ is denoted $R_i$, and the reason clauses corresponding to the first $i$ trail units are denoted $\mathcal{R}_{\le i} = \{R_1, \dots, R_i\}$.
    \item If $f_i$ is a decision, no corresponding reason clause exists. When we write a list of reason clauses $R_1, \dots, R_k$, this list implicitly skips over decisions. In other words, when $f_i$ is a decision, $\mathcal{R}_{\le i} = \mathcal{R}_{\le i-1}$. 
    \item Note that the trail $U$ is always linearly independent (analogous to not repeating literals), and it is consistent (analogous to never having both a literal and its negation) until unit propagation propagates $\bot$.
    \item Note that while $U$ denotes the trail in $\CDCL(\oplus)$, we often also use $U$ more broadly to refer to a set of equations, in the mathematical spirit of stating concepts in more generality. In such cases, it is always fine to continue imagining $U$ as a trail.
\end{itemize}
\end{convention}

In the rest of this section, we will explain the unit propagation and clause learning subroutines. (For decision heuristics, see \cref{ssec:heur}.) For our generalizations of these core classical concepts, we attempt to mimic semantics when possible, meaning that we consider what is logically implied by some clauses and/or equations, in order to have strong justification for our definitions. The reader interested in refreshing their memory on the semantics of classical CDCL may consider reading \cref{sec:classic}.

\subsection{Unit propagation} \label{ssec:up}

Unit propagation in classical CDCL looks at a Boolean clause $C = \ell_1 \lor \dots \lor \ell_m$ and checks how many literals $\ell_i$ are \emph{falsified} by the trail $U$. If all literals are falsified, the solver deduces a contradiction $\bot$, and if all literals but one are falsified, the solver deduces the remaining literal and adds it to the trail. To generalize this, recall the equivalent semantic definition, that unit propagation on a single clause $C$ simply deduces all units that are semantically implied by $C$ and the trail $U$. We do this in the example below.

\begin{example} \label{ex:unitprop}
    Suppose that $U = \{\eq{x = 0}, \eq{z = 0}\}$, and we are processing the clause $C = \eq{x \oplus y \oplus z = 1} \lor \eq{y = 1}$. By plugging the known information into the clause, we get $\eq{y = 1} \lor \eq{y = 1}$, so the unit $\eq{y = 1}$ is implied true. This shows that if we want to determine if a clause reduces to a unit, it is insufficient to simply check for negations of known units, and some algebraic reasoning is necessary.

    Also, because we deduce $\eq{y = 1}$ and already know $U$, the result is that we know all equations in $\linspan(U, \eq{y = 1}) = \linspan(\eq{x = 0}, \eq{y = 1}, \eq{z = 0})$ to be true, such as $\eq{x \oplus y \oplus z = 1}$. Then, instead of choosing to deduce $\eq{y = 1}$, one can see that it would have been equivalent to deduce from $C$ any equation in $\linspan(U, \eq{y = 1}) \setminus \linspan(U)$, since this would produce the same span when added to $U$.
\end{example}

To generalize the example above, we start with a systematic way of plugging equations into clauses, described in \cref{alg:red}. We call this \emph{reduction}. It uses the concept of the \emph{leading term} of an equation: the first coordinate in its $\mathbb{F}_2^{n+1}$ vector form that is a 1. For example, with variables $x$, $y$, and $z$, the equation $\eq{y \oplus z = 0}$ can be represented as the vector $(0, 1, 1, 0)$, and the leading term is the second coordinate corresponding to $y$.

\begin{algorithm}[h]
\caption{Reduction of a single equation}\label{alg:red}
\begin{algorithmic}[1]
\Require Set of equations $U$ and equation $g$
\Ensure The reduction of $g$ by $U$
\State Let $M$ be a matrix whose rows are equations in $U$.
\State Let $h_1, \dots, h_k$ be the non-zero rows of $M$ in row-echelon form.
\For{$i = 1, \dots, k$}
    \If{$g$ contains the leading term of $h_i$}
        \State $g \gets g + h_i$
    \EndIf
\EndFor
\State \Return $g$
\end{algorithmic}
\end{algorithm}

\begin{definition}
    Let $U$ be a set of equations and let $g$ be any other equation. The output of \cref{alg:red} on input $U$ and $g$ is called the \emph{reduction of $g$ by $U$}, denoted $\red_U(g)$. For a clause $C = g_1 \lor \dots \lor g_m$, the \emph{reduction of $C$ by $U$} is given by $\red_U(C) = \red_U(g_1) \lor \dots \lor \red_U(g_m)$. 
\end{definition}

For example, where $U = \{\eq{x = 0}, \eq{z = 0}\}$ and $g = \eq{x \oplus y \oplus z = 1}$, we may compute that $\red_U(g) = \eq{y = 1}$. In particular, $g$ contains the leading term $x$ of $\eq{x = 0}$, so we add $\eq{x = 0}$ to $g$ in order to eliminate $x$, and then repeat with $z$. Thus, reduction is exactly the formal way to ``plug things into an equation.'' (One may also note that reduction is the same concept as the normal form with respect to a Gröbner basis for $U$, though we will not use this term going forward.)

\begin{proposition} \label{prop:red}
    For any set of equations $U$ and equation $g$, the following properties of reduction hold:
    \begin{enumerate}
        \item The equation $\red_U(g)$ depends only on $g$ and $\linspan(U)$. 
        \item Reduction is linear: $\red_U(f + g) = \red_U(f) + \red_U(g)$.
        \item If $U$ is consistent, the subspace $\linneg{\red_U(C)}$ depends only on $\linneg{C}$ and $\linspan(U)$. 
    \end{enumerate}
\end{proposition}
\begin{proof}
    The leading terms of a row-echelon basis are uniquely determined for a subspace $\linspan(U)$ after fixing the standard basis. The key property is that there is a unique equation in $g + \linspan(U) = \{g + u \mid u \in \linspan(U)\}$ that is all-zero in those leading term positions, which the algorithm finds. It is unique because if both $g + u$ and $g + v$ are all-zero in those positions, then so is $u + v$, but it is impossible for a non-zero vector in $\linspan(U)$ to be zero at all leading terms of a row-echelon basis.

    This immediately proves the first part, and the second part also follows from this property because $\red_U(f) + \red_U(g)$ must be all-zeros on the leading term positions, making it the unique candidate for $\red_U(f + g)$. The third part follows from applying linearity twice: first to push negations inside $\linneg{\red_U(C)}$ (using consistency to say that $\red_U(\bot) = \bot$), then again to swap the order of linear combination and reduction.
\end{proof}

After we plug equations into a clause, it remains to check if the clause has become a unit. Recall that this notion of ``unit'' relies on the semantic width of the clause given by $\width(C) = \dim(\linneg{C})$, to avoid issues with duplicates or linear dependencies. This is implemented in \cref{alg:unitprop}, which performs propagation based on casework on $\width(\red_U(C))$, similar to how the classical algorithm counts the number of unfalsified literals. 

\begin{algorithm}
\caption{Unit propagation on a particular clause}\label{alg:unitprop}
\begin{algorithmic}[1]
\Require Set of consistent units $U$ and clause $C$
\Ensure The output of unit propagation on $C$ starting from $U$
\If{$\red_U(C)$ is tautological (i.e., $\bot \in \linneg{\red_U(C)}$)}
    \State \Return $\emptyset$
\ElsIf{$\width(\red_U(C)) = 0$}
    \State \Return $\{\bot\}$
\ElsIf{$\width(\red_U(C)) = 1$}
    \State Let $f$ be the equation such that $f \equiv \red_U(C)$.
    \State \Return $\{f\}$
\Else \Comment{$\width(\red_U(C)) \ge 2$}
    \State \Return $\emptyset$
\EndIf
\end{algorithmic}
\end{algorithm}

It is important to distinguish the single equation output by a successful unit propagation from all of the other equations that are now known from the linear span. Continuing \cref{ex:unitprop}, the algorithm would compute exactly $\red_U(C) = \eq{y = 1} \lor \eq{y = 1}$ and output $\{f\} = \{\eq{y = 1}\}$ as desired, resulting in the updated set of units $\{\eq{x = 0}, \eq{y = 1}, \eq{z = 0}\}$. This process not only gives us $f = \eq{y=1}$, but also tells us that the equations $\eq{x \oplus y = 1}$, $\eq{x \oplus y \oplus z = 1}$, and $\eq{y \oplus z = 1}$ must be true. Furthermore, an alternative algorithm taking any of these 4 equations would have resulted in the same final subspace. Thus, we distinguish the equation $f$ that is \emph{output} from unit propagation from the broader set of equations $\linspan(U, f) \setminus \linspan(U)$ that we say are \emph{propagated} and now known.

Additionally, we remark that item (3) of \cref{prop:red} immediately gives that the output of \cref{alg:unitprop} depends only on $\linspan(U)$ and $\linneg{C}$. A nice property is that the leading term of the output of \cref{alg:unitprop} is distinct from existing leading terms in $U$, which can be useful for keeping $U$ in row-echelon form. 

\begin{definition}
    We say that unit propagation on a clause $C$ starting from a consistent set of equations $U$ \emph{outputs $f$} if $\{f\}$ (including possibly $f = \bot$) is the output of \cref{alg:unitprop} with inputs $U$ and $C$. 
    
    We say that \emph{$C$ propagates $g$ (from $U$)} if the output $f$ satisfies $f \neq \bot$ and $g \in \linspan(U, f) \setminus \linspan(U)$. We also say that \emph{$C$ propagates $\bot$ (from $U$)} if the output $f$ is $\bot$.
\end{definition}

The following characterization of unit propagation is also often useful. It corresponds more closely to the classical intuition of needing to falsify all but one literal in a clause. 

\begin{proposition}[Structure of propagating clauses]\label{prop:up-decomp}
    Consider an equation $g \not \in \linspan^\pm(U)$. (In particular, $g \neq \bot$.) A clause $C$ propagates $g$ if and only if 
    \begin{equation*}
        C \equiv C' \lor (g + u)
    \end{equation*}
    for some clause $C'$ falsified by $U$ and equation $u \in \linspan(U)$. Furthermore,
    \begin{equation*}
        \linneg{C'} = \linneg{C} \cap \linspan(U),
    \end{equation*}
    so the choice of $C'$ is semantically canonical.
\end{proposition}
\begin{proof}
    Suppose that $C$ propagates $g$. Then $g = v + f$ for some $v \in \linspan(U)$, where $\red_U(C) \equiv f$. By definition of reduction, after reordering the equations in the clause,
    \begin{equation*}
        C = \lnot u_{1} \lor \dots \lor \lnot u_k \lor (u_{k+1} + f) \lor \dots \lor (u_m + f)
    \end{equation*}
    where each $u_i \in \linspan(U)$ and $k < m$. (This is exactly the form such that plugging in $U$ leaves every disjunct as either $f$ or $\bot$.) Because we can add $(u_j + f) + (u_m + f) = (u_j + u_m)$ to cancel out $f$, one can check that the following has the same linear negation as $C$ and is thus equivalent to $C$ via \cref{prop:equiv}(2), allowing us to define $C'$ as indicated.
    \begin{equation*}
        C \equiv \underbrace{ \lnot u_{1} \lor \dots \lor \lnot u_k \lor \lnot (u_{k+1} + u_m) \lor \dots \lor \lnot (u_{m-1} + u_m) }_{\displaystyle C'} {} \lor (u_m + f)
    \end{equation*}
    Now, $\linneg{C'} \subseteq \linspan(U)$, so $C'$ is falsified by $U$. The desired $u$ such that $C \equiv C' \lor (g+u)$ is given by $u = v + u_m$. Conversely, one can just trace the unit propagation algorithm on $C' \lor (g + u)$ to see that $\red_U(g)$ will be exactly the output, and hence the equation $g$ will be propagated.

    For the claim that $\linneg{C'} = \linneg{C} \cap \linspan(U)$, the $\subseteq$ direction is clear by the above. For the $\supseteq$ direction, any equation in $\linneg{C} = \linneg{C' \lor (g + u)}$ may be written as $c' + \alpha \cdot \lnot(g + u)$ where $c' \in \linneg{C'}$ and $\alpha \in \{0, 1\}$. For such an equation to also lie in $\linspan(U)$, it must be that $\alpha = 0$, leaving only $c'$. 
\end{proof}

We can also prove the semantic properties of unit propagation previously alluded to:   it deduces every unit consequence of a clause and a set of units. Note that we use the term ``completeness'' loosely here, referring narrowly to deriving units from a single clause, certainly \emph{not} completeness of deriving logical consequences of many clauses.

\begin{proposition}[Soundness and completeness of unit propagation] \label{prop:semantic} Let $U$ be a consistent set of equations, let $C$ be a clause, and let $g$ be any equation. 
\begin{enumerate}
    \item (soundness) If $C$ propagates $g$ from $U$, then $C \land U \Rightarrow g$, including when $g = \bot$.
    \item (completeness) If $C \land U \Rightarrow \bot$, then $C$ propagates $\bot$ from $U$. Otherwise, whenever $C \land U \Rightarrow g$ and $U \not \Rightarrow g$, the clause $C$ propagates $g$ from $U$.
\end{enumerate}
\end{proposition}
\begin{proof}
    The algorithm is clearly sound. For completeness, we analyze in cases similar to the algorithm. Note that $C \land U$ is semantically equivalent to $\red_U(C) \land U$. 
    
    If $\red_U(C)$ is tautological, then $C \land U$ is semantically equivalent to $U$. In this case, we do not want to propagate $\bot$ because $U$ is consistent, and there are no other equations $g$ to propagate by the contradictory hypotheses $C \land U \Rightarrow g$ and $U \not \Rightarrow g$.
    
    If $\red_U(C)$ has width 0, this means $\red_U(C) = \bot$ and we are exactly in the case of $C \land U \Rightarrow \bot$. We propagate $\bot$ as desired. 
    
    Lastly, suppose that $\red_U(C)$ has width $\ge 1$. Then $\red_U(C)$ has an equation with leading term that is not $\bot$ and distinct from all of $U$, so $C \land U$ is satisfiable and $C \land U \not \Rightarrow \bot$. We desire to propagate equations $g$ satisfying $C \land U \Rightarrow g$, which rearranges to $U \land \lnot g \Rightarrow \lnot C$, which by \cref{prop:linalg}(2) means that $\linneg{C} \subseteq \linspan(U, \lnot g)$, here using $U \not \Rightarrow g$ to say that $U \cup \{\lnot g\}$ is consistent. Because $\width(\red_U(C)) \ge 1$, we know $\linneg C$ is not contained in $\linspan(U)$, but it is simultaneously contained in all such subspaces $\linspan(U, \lnot g)$ that are each only 1 dimension larger than $\linspan(U)$, so all of these subspaces (if any exist) must actually be the same space $V$. Then $\red_U(C)$ must have width 1, and the algorithm outputs $\{f\}$ where $f$ also satisfies $\linspan(U, \lnot f) = V$. We conclude that $g \in \linspan(U, f) \setminus \linspan(U)$ as desired. 
\end{proof}

Finally, we give the procedure for unit propagation on a set of clauses $\Delta$ until fixed point in \cref{alg:up-fp}. We phrase this procedure in the context of \cref{alg:cdcl-parity}, where the solver state consists of a set of clauses $\Delta$, a trail $U = U_i$ for some $i$, and reason clauses $\mathcal{R} = \mathcal{R}_{\le i}$ for the same $i$. It runs unit propagation on a single clause in a loop and records new units and the clauses used to deduce those new units, just like the classical version.

\begin{algorithm}
\caption{Unit propagation to fixed point}\label{alg:up-fp}
\begin{algorithmic}[1]
\Require A solver state with clauses $\Delta$, trail $U$, and reasons $\mathcal{R}$
\Ensure An updated solver state reflecting unit propagation on $\Delta$ starting from $U$
\State Let $i$ be such that $U = U_i$ and $\mathcal{R} = \mathcal{R}_{\le i}$. (That is, $i = |U|$.)
\Repeat
    \State \texttt{changed} $\gets$ false
    \For{each $C \in \Delta$}
        \State $F \gets$ output of \cref{alg:unitprop} on $C$ and $U$
        \If{$F = \{f\}$ where $f \neq \bot$}
            \State Choose any representative $f_{i+1} \in \linspan(U, f) \setminus \linspan(U)$.
            \State Record $R_{i+1} \gets C$ as the reason clause for $f_{i+1}$.
            \State $U \gets U \cup \{f_{i+1}\}$ and $\mathcal{R} \gets \mathcal{R} \cup \{R_{i+1}\}$
            \State \texttt{changed} $\gets$ true and $i \gets i + 1$
        \ElsIf{$F = \{\bot\}$}
            \State Record $R_\bot \gets C$ as the reason clause for $\bot$.
            \State $U \gets U \cup \{\bot\}$ and $\mathcal{R} \gets \mathcal{R} \cup \{R_\bot\}$
            \State \textbf{break} out of the repeat loop
        \EndIf
    \EndFor
\Until{ \texttt{changed} is false}
\end{algorithmic}
\end{algorithm}

The one new feature not present in classical unit propagation is the choice of representative $f_{i+1}$ in line 7 of \cref{alg:up-fp}. Of course, one can always choose the equation $f$ (the output of \cref{alg:unitprop}) to be the representative. However, we allow flexibility in choosing the representative because different representatives can be useful in different ways. Note that because unit propagation depends only on $\linspan(U)$ and not the specific equations syntactically in $U$, the choice of representative does not affect the output of any algorithm.

For example, suppose that $U = \{\eq{x = 0}, \eq{z = 0}\}$ and we run \cref{alg:unitprop} on the clause $R_i = \eq{x \oplus y = 0} \lor \eq{y \oplus z = 0}$. The output is $f = \eq{y = 0}$, which we can certainly record as the representative. However, it is impossible to express $R_i$ with $\eq{y = 0}$ as one of its disjuncts, because the representative $\eq{y = 0}$ implicitly requires the rest of the trail to already exist. If we want the deduction to always be contained in its reason clause, we can use the procedure in the proof of \cref{prop:up-decomp} to write either $R_i \equiv \eq{x \oplus z = 1} \lor \eq{x \oplus y = 0}$ or $R_i \equiv \eq{x \oplus z = 1} \lor \eq{y \oplus z = 0}$, such that $y$ appears only in the second equation of the clause. Then we take $R_i' = \eq{x \oplus z = 1}$ and we can choose either $f_i = \eq{x \oplus y = 0}$ or $f_i = \eq{y \oplus z = 0}$ as the representative. Either choice produces the same span as $\eq{y = 0}$ when added to $U$. This motivates the following two possible conventions. 

\begin{description}
    \item[Reduced trail] Always simply choose $f_i = f$, the output of \cref{alg:unitprop}. This output has the property of already being reduced by the existing trail, so it has a distinct leading term and can be easily added to the trail while keeping the trail in a row-echelon basis, which is desirable for \cref{alg:red}. Our implementation relies on the trail having distinct leading terms, which this guarantees.
    \item[Horn-form trail] By \cref{prop:up-decomp}, a reason clause $R_i$ may be written as $R_i' \lor (f + u)$, where $f$ was the output of \cref{alg:unitprop} and $u \in \linspan(U_{i-1})$. This option chooses $f_i = f + u$, allowing reason clauses to be written as $R_i' \lor f_i$. This choice allows reason clauses to be viewed as Horn clauses, where $R_i'$ being falsified leads to the deduction of $f_i$. This is useful in theoretical analyses to avoid clutter with extra terms like $u$. 
\end{description}

Lastly, we conclude with the standard fact that unit propagation to fixed point is \emph{confluent}, meaning that the final state is independent of the order in which the clauses in $\Delta$ are chosen. For the proof, we think of unit propagation as a rewriting system for subspaces of equations (e.g., $\linspan(U)$), where an arrow goes from $\linspan(U)$ to $\linspan(U, f)$ (with $f \not \in \linspan^\pm(U)$) if there exists a clause $C$ such that $C$ propagates $f$ from $U$. We also identify all inconsistent subspaces as equivalent, and draw an arrow from $\linspan(U)$ to the inconsistent state if there exists a clause $C$ such that $C$ propagates $\bot$ from $U$.

\begin{proposition}[Confluence of unit propagation]
For any set of clauses $\Delta$ and consistent set of equations $U$, suppose that $U_*$ and $U_*'$ are two final lists of trail equations from using different clause orders in \cref{alg:up-fp}. Then either $\bot \in \linspan(U_*) \cap \linspan(U_*')$ or $\linspan(U_*) = \linspan(U_*')$.
\end{proposition}
\begin{proof}
    By Newman's lemma (see, e.g., Lemma 2.7.2 in \cite{DBLP:books/daglib/0092409}), it suffices to prove that the unit propagation rewriting relation (on the system described above) is terminating and locally confluent. Termination is clear because following each arrow increases the dimension of the subspace by one, for a maximum of $n + 1$ steps.

    Local confluence means that for any two arrows from the same $U$, there are sequences of more arrows that lead to the exact same state (remembering that we have identified all inconsistent subspaces to be equivalent). To prove this, just note that unit propagation continues to work from larger starting spaces: if $C$ propagates $g$ from $U$, then $C$ also propagates $g$ from $V \supseteq \linspan(U)$, unless $g \in V$ already or $C$ is falsified by $V$. (This follows from \cref{prop:semantic}.) Some simple case analysis is needed to deal with inconsistent states, but the main case is just that if $C$ propagates $f$ and $D$ propagates $g$, then unit propagation on $C$ from $\linspan(U, g)$ and on $D$ from $\linspan(U, f)$ both result in arriving at $\linspan(U, f, g)$. 
\end{proof}

\subsection{Affine resolution} \label{ssec:aff-res}

A fundamental concept in SAT solving is the resolution rule. In this rule, when Boolean clauses $C$ and $D$ can be written as $C = C' \lor \ell$ and $D = D' \lor \lnot \ell$ for some literal $\ell$, one can ``resolve on $\ell$'' to deduce $C' \lor D'$. Such a deduction is often written in the following notation.
\begin{prooftree}
    \AxiomC{$C' \lor \ell$ \vphantom{()}}
    \AxiomC{$D' \lor \lnot \ell$ \vphantom{()}}
    \BinaryInfC{$C' \lor D'$ \vphantom{()}}
\end{prooftree}

Note that the phrase ``on $\ell$'' is actually redundant. If $C$ contains both $\ell_1$ and $\ell_2$, and $D$ contains both $\lnot \ell_1$ and $\lnot \ell_2$, then any attempt at resolving them will create a tautological clause. Thus, the existence of a nontrivial resolvent depends only on the clauses $C$ and $D$, not the literal being resolved on. 

We can generalize resolution to our XNF setting with linear clauses, and we call it \emph{affine resolution}, or just ``resolution'' when clear from context. The equation being resolved away will be redundant information, much like $\ell$. 

\begin{definition}
    Let $U$ be a consistent set of equations, let $g \not \in \linspan^\pm(U)$, and let $C$ and $D$ be clauses such that $C$ propagates $g$ and $D$ propagates $\lnot g$ starting from $U$. Equivalently by \cref{prop:up-decomp}, $C \equiv C' \lor (g + u)$ and $D \equiv D' \lor (\lnot g + v)$ for some clauses $C'$ and $D'$ falsified by $U$ and some equations $u, v \in \linspan(U)$. The \emph{affine resolution rule} derives the clause $C' \lor D' \lor \lnot(u + v)$.
    \begin{prooftree}
    \AxiomC{$C' \lor (g + u)$}
    \AxiomC{$D' \lor (\lnot g + v)$}
    \BinaryInfC{$C' \lor D' \lor \lnot(u + v)$}
    \end{prooftree}
    We call this \emph{resolving (on $g$) with respect to $U$}, and the resolvent is falsified by $U$.
\end{definition}

Affine resolution generalizes classical resolution, which corresponds to the special case of $u = v = \top$, where the resolvent $C' \lor D' \lor \bot$ would be equivalent to $C' \lor D'$. The rule is sound because if $C$ were only true because $(g + u)$ was true, and $D$ were only true because $(\lnot g + v)$ was true, then $g + u + \lnot g + v = \lnot (u + v)$ must be true. 

Following the earlier remark, although the result of affine resolution does not depend on $g$, there is a new dependence on $U$. Classically, the choice of $U$ was irrelevant because the decompositions $C \equiv C' \lor \ell$ and $D \equiv D' \lor \lnot \ell$ do not depend on $U$. (In the classical case, the requirement that $U$ is consistent and falsifies $C'$ and $D'$ only enforces that the resolvent is non-tautological.) In the XNF setting, the clauses $C'$ and $D'$ are given by $\linneg{C'} = \linneg{C} \cap \linspan(U)$ and $\linneg{D'} = \linneg{D} \cap \linspan(U)$, so they can certainly depend on $U$, and thus the outcome of affine resolution can change. The following gives an example.

\begin{example}
    Suppose that one wishes to resolve the following clauses.
    \begin{equation*}
        C = \eq{x = 0} \lor \eq{z = 0} \qquad \qquad D = \eq{y = 0} \lor \eq{z = 1}
    \end{equation*}
    If we take $U = \{\eq{x = 1}, \eq{y = 1}\}$ and $g = \eq{z = 0}$, this resolves to $\eq{x = 0} \lor \eq{y = 0}$, exactly like classical resolution. However, if we keep the same $g$ but take $U = \{\eq{x \oplus z = 0}, \eq{y = 1}\}$, this requires writing $C \equiv \eq{x \oplus z = 1} \lor \eq{z = 0}$, where $C' = \eq{x \oplus z = 1}$ is falsified by the trail. The clause $D$ is already written in the correct form, where $D' = \eq{y = 0}$ is already falsified by the trail, so this results in a final resolvent of $\eq{x \oplus z = 1} \lor \eq{y = 0}$. Both resolvents are sound deductions, yet they are different.
\end{example}

\begin{example} \label{ex:aff-res}
    The intuition for affine resolution is the same as for classical resolution, in that something becomes eliminated. Let $U = \{\eq{x = 0}, \eq{y = 0}, \eq{z = 0}, \eq{w = 0}\}$, and suppose that we wish to resolve the following clauses to get rid of $g = \eq{a = 1}$. 
    \begin{equation*}
        C = \eq{x \oplus a = 1} \lor \eq{y \oplus a = 1} \qquad \qquad D = \eq{z \oplus a = 0} \lor \eq{w \oplus a = 0}
    \end{equation*}
    Resolution is possible because $C$ propagates $g = \eq{a = 1}$ and $D$ propagates $\lnot g = \eq{a = 0}$ after plugging in $U$. First, we follow \cref{prop:up-decomp} to create the following decompositions.
    \begin{align*}
        C & \equiv \underbrace{\eq{x \oplus y = 1}}_{\displaystyle C'} {} \lor (\underbrace{\eq{a = 1}}_{\displaystyle g \vphantom{C'}} {} + {} \underbrace{\eq{y = 0}}_{\displaystyle u \vphantom{C'}}) \\
        D & \equiv \underbrace{\eq{z \oplus w = 1}}_{\displaystyle D'} {} \lor (\underbrace{\eq{a = 0}}_{\displaystyle \lnot g \vphantom{D'}} {} + {} \underbrace{\eq{w = 0}}_{\displaystyle v \vphantom{D'}})
    \end{align*}
    Then we apply affine resolution below, where the equations that change have been highlighted.
    
    \begin{prooftree}
        \AxiomC{$\eq{x \oplus y = 1} \lor \highlight{(\eq{a = 1} + \eq{y = 0})}$}
        \AxiomC{$\eq{z \oplus w = 1} \lor \highlight{(\eq{a = 0} + \eq{w = 0})}$}
        \BinaryInfC{$\eq{x \oplus y = 1} \lor \eq{z \oplus w = 1} \lor \highlight{\eq{y \oplus w = 1}}$}
    \end{prooftree}
    The resolvent is $C' \lor D' \lor \lnot(u + v) = \eq{x \oplus y = 1} \lor \eq{z \oplus w = 1} \lor \eq{y \oplus w = 1}$. Notice that the resolvent of a width-$k$ clause and a width-$\ell$ clause can have width up to $k + \ell - 1$ as it is here, in contrast to $k + \ell - 2$ for classical resolution. 
\end{example}

While affine resolution as defined above requires $C$ and $D$ to unit propagate $g$ and $\lnot g$ respectively, we remark that we will relax these requirements in later sections. For example, a relaxed definition could allow both to unit propagate $g$. Then $C \equiv C' \lor (g + u)$ and $D \equiv D' \lor (g + v)$ could ``resolve'' to give $C' \lor D' \lor (u + v)$. While this would not be useful for classical CDCL where $u = v = \top$ and the resolvent would be tautological, it typically produces a non-tautological clause in the XNF setting. Such ``resolvents'' are necessary for an equivalent of variable elimination in \cref{ssec:elim}. 

Even more generally, one could have no preconditions and start with any clauses $C \lor f_C$ and $D \lor f_D$, then deduce $C \lor D \lor (f_C + f_D)$. (Affine resolution would be the special case of $f_C = g + u$ and $f_D = \lnot g + v$, resulting in $g$ being canceled out in $f_C + f_D$.) We call the general principle the \emph{addition rule}, and we discuss it in more detail in \cref{ssec:sim-res}, where its simplicity is useful for theoretical arguments. 

We end with a short note formally establishing that the result of affine resolution depends only on $C$, $D$, and $U$, and in particular does not depend on $g$. 

\begin{proposition} \label{prop:aff-struct}
    Let $U$ be a consistent set of equations. Following the definition of affine resolution, let
    $C \equiv C' \lor (g + u)$ and $D \equiv D' \lor (\lnot g + v)$ for some clauses $C'$ and $D'$ falsified by $U$ and some equations $u, v \in \linspan(U)$.
    Let $c \in \linneg{C} \setminus \linspan(U)$ and $d \in \linneg{D} \setminus \linspan(U)$ be arbitrary. Then the subspace corresponding to the resolvent is
    \begin{equation*}
        \linspan(\linneg{C'}, \linneg{D'}, u + v) = \linspan(\linneg{C'}, \linneg{D'}, \lnot(c + d)).
    \end{equation*}
    In particular, because $\linneg{C'} = \linneg{C} \cap \linspan(U)$ and $\linneg{D'} = \linneg{D} \cap \linspan(U)$, the result of affine resolution depends only on $C$, $D$, and $U$.
\end{proposition}
\begin{proof}
    Because $\linneg{C} = \linspan(\linneg{C'}, \lnot g + u)$ and $\linneg{D} = \linspan(\linneg{D'}, g + v)$, we can write $c = c' + \lnot g + u$ and $d = d' + g + v$ for some $c' \in \linneg{C'}$ and $d' \in \linneg{D'}$. Then $\lnot (c + d) = c' + d' + (u + v)$, so the claim follows easily.
\end{proof}

\subsection{Conflict analysis} \label{ssec:learning}

We turn our attention to conflict analysis in $\CDCL(\oplus)$, where the primary goal is to generalize the 1-UIP clause learning algorithm. We already have our main tool of affine resolution, but we would ideally also like to develop a broader theory of \emph{conflict clauses}, which allows us to explain what clauses are learnable at all, and also allows us to prove an optimality result in \cref{ssec:backjump}. 

Classical CDCL typically defines conflict clauses using \emph{conflict cuts} in an \emph{implication graph},  where trail units are the vertices and we have directed edges $(x, y)$ if $x$ is part of the reason for deducing $y$ in unit propagation. This presents a problem for us, because we typically think of the trail as a nested sequence of subspaces $\{\top\} = \linspan(U_0) \subsetneq \linspan(U_1) \subsetneq \dots \subsetneq \linspan(U_m)$ that increases by one dimension at each step, which does not immediately lend itself well to a graph structure. 

If we did try to use a graph structure to analyze $\CDCL(\oplus)$, there would be some immediate issues. For example, continuing \cref{ex:unitprop} where the clause $C = \eq{x \oplus y \oplus z = 1} \lor \eq{y = 1}$ propagates $\eq{y = 1}$ from $U = \{\eq{x = 0}, \eq{z = 0}\}$, it would not be clear which edges should point into $\eq{y = 1}$. If one creates edges from both $\eq{x = 0}$ and $\eq{z=0}$, attempting to capture that these two units were used to propagate $\eq{y = 1}$, the graph only encodes the information $\eq{x = 1} \lor \eq{y = 1} \lor \eq{z = 1}$, which is weaker than $C$. 

Instead, to define conflict clauses and asserting clauses in a way that fits neatly with the subspace view of the trail, we can use semantically motivated definitions of conflict clauses and asserting clauses that are equivalent to their implication-graph definitions in classical CDCL. (These equivalences can be found in \cref{sec:classic}.)

\begin{definition}
    A clause $C$ is a \emph{conflict clause} with respect to a trail $U = \{f_1, \dots, f_m, \bot\}$ if $C$ is falsified by $U_m$ and unit propagation (to fixed point) with the trail's reasons $\mathcal{R} = \{R_1, \dots, R_m, R_\bot\}$ as the set of clauses, starting with units $\linneg C$, finds a contradiction.
    
    We consider the trail as partitioned into \emph{levels} by decisions, where each decision starts a new level. The $k$-th decision starts level $k$. The \emph{asserting level} of a conflict clause $C$ is the smallest level $k$ such that $C$ propagates some equation $g$ (possibly $\bot$) starting from the units of level $k$ and below. A conflict clause $C$ is \emph{asserting} if its asserting level is less than the total number of decisions $d$, and the various $g$ that can be propagated are called the \emph{asserted equations}. 
\end{definition}

Note that our definition of conflict clause relies on unit propagation, which we have already established to be a semantically meaningful notion. One could imagine a different definition where we use the criterion of $\mathcal{R}$ semantically implying $C$. Although this would be entirely equivalent for Boolean clauses (again, \cref{sec:classic}), the following example shows that this is a different notion for linear clauses. Our unit-propagation-based definition will generally be the more useful definition.

\begin{example} \label{ex:sem-confl} \allowdisplaybreaks
    Consider the following sequence of events.
    \begin{alignat*}{2}
    f_1 &= \eq{a = 1} &\qquad& \text{decision} \\
    f_2 &= \eq{b = 1} &&       \text{from } R_2 = \eq{a = 0} \lor \eq{b = 1} \\
    f_3 &= \eq{c = 1} &&       \text{from } R_3 = \eq{b = 0} \lor \eq{c = 1} \\
    f_4 &= \eq{d = 1} &&       \text{from } R_4 = \eq{c = 0} \lor \eq{d = 1} \\
    \bot &            &&       \text{from } R_\bot = \eq{c = 0} \lor \eq{d = 0}
    \end{alignat*}
    This trail is entirely in CNF. One can check that the unit linear clause $C = \eq{a \oplus b \oplus c = 0}$ is falsified by $U_4$ (in fact, by $U_3$), and is semantically implied by the set of reasons because $R_\bot$ and $R_4$ imply $\eq{c = 0}$, which implies $\eq{b = 0}$ and $\eq{a = 0}$. However, this deduction cannot be verified through unit propagation starting from $\eq{a \oplus b \oplus c = 1}$, so $C$ is \emph{not} a conflict clause by our definition.
\end{example}

Affine resolution then allows us to easily state the equivalent of the 1-UIP clause learning algorithm. In classical CDCL, a 1-UIP asserting clause is found by taking $R_\bot$ and going backwards, successively resolving with other reason clauses on the propagated variables. 
In $\CDCL(\oplus)$, we do the exact same thing, replacing classical resolution with affine resolution. This is written slightly more generally in \cref{alg:assert} for easier reference in future sections, where we may want to refer to this algorithm with various stopping conditions. Taking ``$C$ is asserting'' to be the stopping condition gives exactly the 1-UIP clause learning algorithm. 

\begin{algorithm}[t]
\caption{Resolution on the reversed trail}\label{alg:assert}
\begin{algorithmic}[1]
\Require Units $U = \{f_1, \dots, f_m, \bot\}$ and corresponding reasons $\mathcal{R} = \{R_1, \dots, R_m, R_\bot\}$
\Ensure A conflict clause
\State Initialize $C \gets R_\bot$.
\For{$i \gets m, m-1, \dots, 1$}
    \If{a stopping condition holds}
        \State \Return $C$
    \EndIf
    \If{$f_i$ is a decision}
        \State \Return fail \Comment{Impossible for appropriate stopping conditions}
    \EndIf
    \If{$C$ propagates $\lnot f_i$ from $U_{i-1}$} 
        \State $C \gets$ the resolvent of $C$ and $R_i$ (on $f_i$) with respect to $U_{i-1}$
    \EndIf
\EndFor
\end{algorithmic}
\end{algorithm}

\begin{proposition} \label{prop:assert-correct}
    The following properties hold for \cref{alg:assert}.
    \begin{enumerate}
        \item At the start of every iteration, $U_i$ falsifies $C$. 
        \item The clause $C$ is always a conflict clause.
        \item If ``$C$ is asserting'' is the stopping condition, then the algorithm never fails, and thus computes an asserting clause.
    \end{enumerate}
\end{proposition}
\begin{proof}
    To prove (1), the base case is that $U_m$ falsifies $R_\bot$ by definition of $R_\bot$. Then, after each resolution, $U_{i-1}$ falsifies $C$ by the definition of affine resolution, so the property continues to hold as $i$ decreases by 1. If no resolution step happens at step $i$, meaning $C$ does not propagate $\lnot f_i$ from $U_{i-1}$ even though $U_i$ falsified $C$, then by completeness of unit propagation in \cref{prop:semantic}, this is only possible when $U_{i-1}$ already falsifies $C$. 
    
    To prove (2), it remains to show that unit propagation from $\linneg {C}$ and $\mathcal{R}$ implies a contradiction. Our strategy is to replay the same unit propagation as the trail. We prove inductively that one can derive the linear negation of the previous conflict clause, and knowing the linear negation of the first conflict clause $R_\bot$ certainly derives a contradiction. 
    
    Suppose that $C = C_{\text{prev}}' \lor R_i' \lor \lnot (u + v)$ is the resolvent of $C_{\text{prev}} \equiv C_{\text{prev}}' \lor (\lnot f_i + u)$ and $R_i \equiv R_i' \lor (f_i + v)$ where, as usual, $C_{\text{prev}}'$ and $R_i'$ are falsified by $U_{i-1}$ and $u, v \in \linspan(U_{i-1})$. Because $\linneg{C}$ contains $\linneg{R_i'}$, we know $R_i$ will either propagate $\bot$, in which case we are immediately done, or it will propagate $f_i + v$. Then $f_i + u = (f_i + v) + (u + v)$ is known, and $\linneg{C_{\text{prev}}} =  \linneg{C_{\text{prev}}' \lor (\lnot f_i + u)}$ is known as desired.

    Finally, to show (3), simply note that by (1), at the point where $i$ is the index of the last decision, we have $C$ falsified by $U_i = U_{i-1} \cup \{f_i\}$ where $U_{i-1}$ is the previous level. Thus, $C$ propagates $\lnot f_i$ from $U_{i-1}$, and is asserting. (If we do not assume that every decision level ends at a unit propagation fixed point, it is also possible for $C$ to propagate $\bot$, but we accept this as asserting as well.)
\end{proof}

Recall that the classical 1-UIP clause learning algorithm learns something more specific than just an asserting clause: a unique implication point (UIP) is defined as a literal in the implication graph that lies on every path from the last decision to the contradiction $\bot$. The 1-UIP algorithm finds an asserting clause that immediately propagates the negation of the UIP closest to $\bot$. For \cref{alg:assert} to keep the 1-UIP name, we should also develop a theory of UIPs without implication graphs, and prove similar optimality properties to that of the classical 1-UIP algorithm. This theory is actually substantially complicated, so we defer it to \cref{ssec:backjump}. 

In any case, we end with a computationally more precise framing of 1-UIP clause learning, in which the procedure for affine resolution is inlined without needing to reference the proof of \cref{prop:up-decomp}. This is shown in \cref{alg:trail-basis}. We use the following concept: given a trail $U = \{f_1, \dots, f_m, \bot\}$, the \emph{trail basis} refers to the trail in reverse order: $(\bot, f_m, \dots, f_1)$. When we write a subspace of $\linspan(U)$ in the trail basis, this means that each vector is represented as a linear combination of trail elements, where the first entry is the coefficient of $\bot$, the second entry is the coefficient of $f_m$, and so on. 

The key idea is that in a row-echelon form matrix, if $g$ is the first coordinate, then $g$ cannot appear after the first row, exactly like how $C \equiv C' \lor (g + u)$ from \cref{prop:up-decomp} only has $g$ in one equation. Additionally, in the trail basis, each $\linneg{R_i}$'s first row must be $\bot + f_i + (\text{others})$, starting with $\bot$, and be otherwise made of earlier trail equations in order to propagate $f_i$. In contrast, $\linneg{C}$ may start with $f_i$, but may never involve $\bot$. Thus, to resolve these two clauses, we can simply combine both first rows, throwing away $\bot + f_i$ from the first row of $\linneg{R_i}$ and $f_i$ from the first row of $\linneg{C}$, then concatenate with the rest of the two clauses. The stopping condition ``$C$ is asserting'' translates to the second row of $\linneg{C}$ starting with something in the previous level. Overall, this is written in \cref{alg:trail-basis}, which computes the same thing as \cref{alg:assert} stopping at asserting clauses.

\begin{algorithm}
\caption{1-UIP clause learning in the trail basis}\label{alg:trail-basis}
\begin{algorithmic}[1]
\Require Units $U = \{f_1, \dots, f_m, \bot\}$ and corresponding reasons $\mathcal{R} = \{R_1, \dots, R_m, R_\bot\}$
\Ensure Identical to \cref{alg:assert} with stopping condition ``$C$ is asserting''
\State Initialize $C \gets R_\bot$.
\State Let $M_C$ be the row-echelon form matrix of $\linneg{C}$ in the trail basis.
\For{$i \gets m, m-1, \dots, 1$}
    \If{the second row of $M_C$ has no coordinate in the last level}
        \State \Return the clause corresponding to $M_C$ (in normal coordinates)
    \EndIf
    \If{the first row of $M_C$ starts with $f_i$}
        \State Let $M_{R_i}$ be the row-echelon form matrix of $\linneg{R_i}$ in the trail basis.
        \State Denote $M_C = \begin{bmatrix} f_i + u \\ M_C' \end{bmatrix}$ and
               $M_{R_i} = \begin{bmatrix} \bot + f_i + v \\ M_{R_i}' \end{bmatrix}$ (i.e., pop the first row of each).
        \State $M_C \gets$ the concatenation of $M_C'$, $M_{R_i}'$, and $u + v$
        \State Rewrite $M_C$ in row-echelon form.
    \EndIf
\EndFor
\end{algorithmic}
\end{algorithm}

\begin{example}
\allowdisplaybreaks
This extended example will cover the entire $\CDCL(\oplus)$ procedure, and will use \cref{alg:trail-basis} at the end. Consider three clauses over variables $a, b, c, d, e$:
\begin{align*}
C_1 &= \eq{a \oplus b \oplus e = 1} \lor \eq{b \oplus c \oplus d \oplus e = 0} \\
C_2 &= \eq{b \oplus d = 1} \lor \eq{d \oplus e = 1} \\
C_3 &= \eq{b \oplus e = 0} \lor \eq{c \oplus d = 0}
\end{align*}
We encode each equation as a vector in $\mathbb{F}_2^6$ by reading off the coefficients, ending with the right-hand side. Each clause $C_i$ is stored as a basis of $\linneg{C_i}$ (a list of vectors):
\begin{align*}
\linneg{C_1} &= \linspan(\mathtt{110010},\, \mathtt{011111}) \\
\linneg{C_2} &= \linspan(\mathtt{010100},\, \mathtt{000110}) \\
\linneg{C_3} &= \linspan(\mathtt{010011},\, \mathtt{001101})
\end{align*}
Suppose that $\CDCL(\oplus)$ decides $f_1 = \eq{b \oplus c = 0}$, i.e., $f_1 = \mathtt{011000}$. To check unit propagation, we reduce each basis vector of each $\linneg{C_i}$ by the known units via \cref{alg:red}. If the remaining width is at least 2, the clause does not propagate.
\begin{align*}
\linneg{C_1} \text{ reduced by } \{f_1\} &= \linspan(\mathtt{101010},\, \mathtt{000111}) \\
\linneg{C_2} \text{ reduced by } \{f_1\} &= \linspan(\mathtt{001100},\, \mathtt{000110}) \\
\linneg{C_3} \text{ reduced by } \{f_1\} &= \linspan(\mathtt{001011},\, \mathtt{001101})
\end{align*}
No clause propagates. Next, suppose that $\CDCL(\oplus)$ decides $f_2 = \eq{a \oplus b \oplus d = 1}$, i.e., $f_2 = \mathtt{110101}$. Then,
\begin{equation*}
\linneg{C_1} \text{ reduced by } \{f_1, f_2\} = \linspan(\mathtt{000111},\, \mathtt{000111}).
\end{equation*}
Because the two equations are identical, this is width 1, and so remembering that $C_1$ is stored in negated form, unit propagation outputs $f_3 = \eq{d \oplus e = 0}$, i.e., $f_3 = \mathtt{000110}$. Continuing,
\begin{equation*}
\linneg{C_2} \text{ reduced by } \{f_1, f_2, f_3\} = \linspan(\mathtt{001010},\, \mathtt{000000}),
\end{equation*}
and we propagate $f_4 = \eq{c \oplus e = 1}$, i.e., $f_4 = \mathtt{001011}$. Finally,
\begin{equation*}
\linneg{C_3} \text{ reduced by } \{f_1, f_2, f_3, f_4\} = \linspan(\mathtt{000000},\, \mathtt{000000}).
\end{equation*}
All basis vectors reduce to zero, so $\linneg{C_3} \subseteq \linspan(f_1, f_2, f_3, f_4)$, and unit propagation finds a contradiction, with $R_\bot = C_3$.

Proceeding to conflict analysis, we rewrite the linear negation of each clause to be in row-echelon form in the trail basis $(\bot, f_4, f_3, f_2, f_1)$:
\begin{align*}
\linneg{C_1} &= \linspan(\bot + f_3 + f_2, f_2 + f_1) \\
\linneg{C_2} &= \linspan(\bot + f_4 + f_3 + f_1, f_3)\\
\linneg{C_3} &= \linspan(f_4 + f_1, f_3 + f_1)
\end{align*}
The unit $f_1$ is at the first decision level, and $f_2, f_3, f_4$ are all at the second decision level. First, $C_3$ is not asserting because unit propagation cannot derive anything from $C_3$ and $f_1$ alone. 

We resolve $C_3$ with $C_2$ on $f_4$, keeping equations $f_3 + f_1$ from $C_3$ and $f_3$ from $C_2$, and combining the first equation of both clauses without $\bot$ or $f_4$ to get $(f_1) + (f_3 + f_1) = f_3$. Thus, the final clause is just 
\begin{equation*}
    \linneg{C} = \linspan(f_3 + f_1, f_3, f_3).
\end{equation*}
We place this in row-echelon form to get
\begin{equation*}
     \linspan(f_3 + f_1, f_1) \qquad \Leftrightarrow \qquad C = \lnot(f_3 + f_1) \lor \lnot f_1.
\end{equation*}
Unit propagation can now derive $\lnot f_3$ from $f_1$ and $C$, so $C$ is asserting and conflict analysis stops. Translating back to the original variables, we learn the clause
\begin{equation*}
C = \eq{b \oplus c \oplus d \oplus e = 1} \lor \eq{b \oplus c = 1}.
\end{equation*}
Notice that conflict analysis did not use $C_1$, indicating that it is certainly possible for the resolution walk to end before reaching the decision.
\end{example}

\subsection{Projection onto subspaces} \label{ssec:elim}

In classical SAT solving, variable elimination refers to replacing all clauses containing the literal $x$ and all clauses containing $\lnot x$ with all of their pairwise resolvents, and the resulting set of clauses is equisatisfiable with the original set. While the procedure was used in the original Davis--Putnam procedure for SAT solving \cite{DBLP:journals/jacm/DavisP60}, it is no longer featured in the core loop of CDCL today, though it remains part of the bounded variable elimination (BVE) heuristic \cite{DBLP:conf/sat/SubbarayanP04, DBLP:conf/sat/EenB05} and will also be used as a theoretical tool in \cref{ssec:backjump}. We therefore dedicate a subsection to generalizing this idea, maintaining the equisatisfiability condition. 

While we clearly need to use affine resolution, what we defined earlier is also clearly insufficient. Recall that classically, given clauses $C = x \lor y$ and $D = \lnot x \lor \lnot y$, variable elimination on $x$ will attempt to resolve these clauses because $C$ contains $x$ and $D$ contains $\lnot x$. Only after resolving do we realize that the resolvent $y \lor \lnot y$ is tautological, and we drop this resolvent. (This discussion is similar to work on blocked clauses \cite{DBLP:journals/tcs/Kullmann99}.) However, when we defined affine resolution earlier, we required $C$ to propagate $x$ and $D$ to propagate $\lnot x$ from units $U$. But then $U$ has to have both $y$ and $\lnot y$, so it is inconsistent. 

Thus, we should generalize affine resolution such that whenever we write $C \equiv C' \lor (g + u)$, we allow not only $\linneg{C'} \subseteq \linspan(U)$, but more generally $\linneg{C'} \subseteq \linspan^\pm(U)$. We note the peculiar consequence that it is now possible to have $C \equiv C' \lor (g + u) \equiv C' \lor (\lnot g + v)$ for $u, v \in \linspan(U)$. For example, with $U = \{\eq{x = 0}, \eq{y = 0}\}$ and $g = \eq{a = 0}$, we have
\begin{equation*}
    \eq{x \oplus y = 0} \lor \eq{a \oplus x = 0} \equiv \eq{x \oplus y = 0} \lor \eq{a \oplus y = 1},
\end{equation*}
where $\eq{x = 0}$ and $\eq{y = 0}$ both belong to $\linspan(U)$, so that in the second equation of each form of this clause, the first form leaves $g$ and the second form leaves $\lnot g$. Because the distinction between $g$ and $\lnot g$ is no longer meaningful, we allow $u \in \linspan^\pm(U)$ instead, so that we can always write the final equation as $g + u$, where a negation could be folded into $u$ if necessary. We summarize the definition as follows.

\begin{definition}
    Let $U$ be a consistent set of equations and let $g \not \in \linspan^\pm(U)$. Suppose that $C \equiv C' \lor (g + u)$ and $D \equiv D' \lor ( g + v)$ where $\linneg{C'}, \linneg{D'} \subseteq \linspan^\pm(U)$ and $u, v \in \linspan^\pm(U)$. The \emph{generalized affine resolution rule} derives the clause $C' \lor D' \lor (u + v)$. 
    \begin{prooftree}
    \AxiomC{$C' \lor (g + u)$}
    \AxiomC{$D' \lor (g + v)$}
    \BinaryInfC{$C' \lor D' \lor (u + v)$}
    \end{prooftree}
    We also call this \emph{resolving on $g$ with respect to $U$}, as it strictly generalizes the earlier meaning of this term.
\end{definition}

The following proposition establishes that, like (non-generalized) affine resolution, the result of generalized affine resolution does not depend on $g$. Thus, instead of saying ``eliminating $g$,'' we often prefer to say ``projecting onto $\linspan^\pm(U)$,'' to focus on the actually meaningful part.

\begin{proposition}
     Following the definition of generalized affine resolution, let $c \in \linneg{C} \setminus \linspan^\pm(U)$ and $d \in \linneg{D} \setminus \linspan^\pm(U)$ be arbitrary. Then the subspace corresponding to the resolvent is
    \begin{equation*}
        \linspan(\linneg{C'}, \linneg{D'}, \lnot(u + v)) = \linspan(\linneg{C'}, \linneg{D'}, \lnot(c + d)).
    \end{equation*}
    In particular, because $\linneg{C'} = \linneg{C} \cap \linspan^\pm(U)$ and $\linneg{D'} = \linneg{D} \cap \linspan^\pm(U)$, the result of generalized affine resolution depends only on $C$, $D$, and $U$.
\end{proposition}
\begin{proof}
    Identical to \cref{prop:aff-struct} for the main argument, and identical to \cref{prop:up-decomp} for $\linneg{C'} = \linneg{C} \cap \linspan^\pm(U)$.
\end{proof}

It remains to discuss how to use generalized affine resolution to perform variable elimination. First, given a set of clauses $\Delta$, we need to identify $U$ and $g$ such that all clauses $C \in \Delta$ satisfy $\linneg{C} \subseteq \linspan^\pm(U, g)$, in other words, that $U$ and $g$ determine whether every equation in $C$ is true or false. In a typical solver state, this would require taking $U$ to be such that $\linspan^\pm(U)$ is a hyperplane, meaning that $U \cup \{g, \bot\}$ is a basis of $\mathbb{F}_2^{n+1}$. 

Once we have established $U$ and $g$, the clauses $\Gamma = \{C \in \Delta \mid \linneg{C} \not \subseteq \linspan^\pm(U)\}$ are the ones that involve $g$ and will hence be affected by the elimination procedure. Classically, the set $\Gamma$ is partitioned into two sets based on the polarity of $g$. However, as mentioned earlier, there is no meaningful polarity of $g$ for us. Instead of partitioning $\Gamma$, we simply consider all $\binom{\lvert \Gamma \rvert}{2}$ possible resolvents. In resolving $C' \lor (g +u)$ with $D' \lor (g + v)$, the classical cases are $u, v \in \{\top, \bot\}$, and the classical algorithm saves time by avoiding even trying $u = v = \top$ and $u = v = \bot$, which result in tautological resolvents. We typically have nontrivial $u + v$ and need to include these extra resolvents.

\begin{example}
    Consider the following clauses.
    \begin{equation*}
        C_1 = \eq{x = 1} \lor \eq{a = 1} \qquad C_2 = \eq{y = 1} \lor \eq{a = 0} \qquad C_3 = \eq{x \oplus y = 0} \lor \eq{ a \oplus  y = 1}
    \end{equation*}
    Suppose that our goal is to generalize the idea of ``eliminating $a$.'' We can take $U = \{\eq{x = 0}, \eq{y = 0}\}$ and $g = \eq{a = 0}$, and affine resolve on $g$ with respect to $U$. (Since we care about $\linspan^\pm$, the precise right-hand-side constants don't matter, so other options are possible.) Note that this choice of $U$ indeed allows all clauses to live in $\linspan^\pm(U, g)$. Each of the three clauses is already conveniently written to only involve $a$ in its second equation, so we can compute the resolvents to be the following, where $R_{ij}$ denotes the resolvent of $C_i$ and $C_j$.
    \begin{align*}
        R_{12} & = \eq{x = 1} \lor \eq{y = 1} \\
        R_{13} & = \eq{x = 1} \lor \eq{x \oplus y = 0} \lor \eq{y = 0} \equiv \eq{x = 1} \lor \eq{y = 0} \\
        R_{23} & = \eq{y = 1} \lor \eq{x \oplus y = 0} \lor \eq{y = 1} \equiv \eq{x = 0} \lor \eq{y = 1}
    \end{align*}
    We see that all three resolvents are actually non-tautological and irredundant. This is noteworthy because three Boolean clauses can classically only have up to two non-tautological resolvents. 
\end{example}

\begin{definition}
     Let $U$ be a consistent set of equations and $g \not \in \linspan^\pm(U)$. Suppose that $\Delta$ is a set of clauses such that $\linneg{C} \subseteq \linspan^\pm(U, g)$ for all $C \in \Delta$. We can \emph{project $\Delta$ onto $\linspan^\pm(U)$} by considering the clauses
    \begin{equation*}
        \Gamma = \{C \in \Delta \mid \linneg{C} \not \subseteq \linspan^\pm(U)\}
    \end{equation*}
    (i.e., all clauses involving $g$) and replacing them with their pairwise resolvents as defined above, with tautologies discarded. 
\end{definition}

Again, when $\Delta$ is the entire clause database, it is typical for $U \cup \{g, \bot\}$ to be a basis of $\mathbb{F}_2^{n+1}$ in order to satisfy $\linneg{C} \subseteq \linspan^\pm(U, g)$ for all $C \in \Delta$. However, smaller $\Delta$ and smaller $U$ can be useful for theoretical arguments. We remark that the 1-UIP clause learning algorithm can be viewed as successive projection of the reason clauses $\mathcal{R}$ onto subspaces $\linspan^\pm(U_m), \linspan^\pm(U_{m-1}), \dots$, where the derived conflict clause $C$ in the algorithm's main loop is exactly the resolvent of the only two clauses involving $f_i$, namely the reason clause $R_i$ and the previous conflict clause. In \cref{ssec:backjump}, we will frequently use this perspective.

We end with the important property that projection yields equisatisfiability. 
Moreover, we show that if satisfying assignments exist, they are related by a linear-algebraic projection, giving even more motivation for this choice of vocabulary. Let $\Sol(\Delta)$ denote the set of assignments that satisfy $\Delta$. 

\begin{proposition}[Solutions after projection] \label{prop:proj-sol}
    Let $U$ be a consistent set of equations, let $g \not \in \linspan^\pm(U)$, and let $\Delta$ be a set of clauses such that $\linneg{C} \subseteq \linspan^\pm(U, g)$ for all $C \in \Delta$. Let $\Delta'$ be the projection of $\Delta$ onto $\linspan^\pm(U)$. Then there exists a vector $\vec{z} \in \mathbb{F}_2^n$ such that 
    \begin{equation*}
        \Sol(\Delta') = \Sol(\Delta) \cup \{\vec{x} + \vec{z} \mid \vec{x} \in \Sol(\Delta)\}.
    \end{equation*}
\end{proposition}
\begin{proof}
    Let the vector $\vec{z}$ be any solution to the system of equations formed by the left-hand sides of $U$ set to 0 and the left-hand side of $g$ set to 1. Then for any starting assignment, adding $\vec{z}$ preserves the truth value of every equation in $\linspan^\pm(U)$ and flips the value of every equation in $\linspan^\pm(U, g) \setminus \linspan^\pm(U)$. 

    To prove the $\supseteq$ direction, by virtue of generalized affine resolution being sound, the assignments satisfying $\Delta$ must also satisfy $\Delta'$. By the above remark, because the clauses after projection are all contained in $\linspan^\pm(U)$, the same assignments plus $\vec{z}$ must satisfy $\Delta'$ as well. 
    
    To prove the $\subseteq$ direction, let $\vec{x} \in \mathbb{F}_2^n$ be an assignment satisfying $\Delta'$, and we claim that either $\vec{x}$ satisfies $\Delta$ or $\vec{x} + \vec{z}$ satisfies $\Delta$. Clauses unchanged by projection live in $\linspan^\pm(U)$, and are thus satisfied by both $\vec{x}$ and $\vec{x} + \vec{z}$ using the above remark. Consider removed clauses $C \equiv C' \lor (g + u)$, where $\linneg{C'} \subseteq \linspan^\pm(U)$ and $u \in \linspan^\pm(U)$. If $\vec{x}$ satisfies $C'$, then again $\vec{x} + \vec{z}$ will also satisfy $C'$ and hence $C$, so we are good. Consider the remaining such clauses where $\vec{x}$ does not satisfy $C'$. For any two such clauses $C_i \equiv C_i' \lor (g + u_i)$ and $C_j \equiv C_j' \lor (g + u_j)$, we know that $\vec{x}$ must satisfy the resolvent $C_i' \lor C_j' \lor (u_i +u_j)$, and hence must satisfy $(u_i + u_j)$. Thus, all $u_i$, and hence all $g + u_i$, actually have the same truth value under $\vec{x}$. By the above remark, switching between $\vec{x}$ and $\vec{x} + \vec{z}$ flips the value of all $g + u_i$, so one of the two choices simultaneously makes all $g + u_i$ true. 
\end{proof}
\begin{corollary}[Equisatisfiability after projection] \label{prop:proj-eqsat}
    Let $\Delta'$ be the projection of $\Delta$ onto $\linspan^\pm(U)$. Then $\Delta'$ is satisfiable if and only if $\Delta$ is satisfiable.
\end{corollary}
\begin{proof}
    Obvious from the characterization of solutions after projection.
\end{proof}

\section{Theoretical analysis of \texorpdfstring{$\CDCL(\oplus)$}{CDCL(⊕)}} \label{sec:theo}

In this section, we collect two independent theoretical results on $\CDCL(\oplus)$. The first result proves an optimality condition for our equivalent of the 1-UIP clause learning algorithm, which justifies the name 1-UIP. The second result first connects affine resolution to the proof system $\Res(\oplus)$ (``Res-parity'') that has been studied extensively in proof complexity~\cite{DBLP:conf/stoc/EfremenkoGI24, DBLP:conf/coco/BhattacharyaCD24, DBLP:conf/stoc/AlekseevI25}, and which inspired the name $\CDCL(\oplus)$. It follows from this connection that $\CDCL(\oplus)$ easily produces $\Res(\oplus)$ proofs, and also simulates all $\Res(\oplus)$ proofs using polynomially many nondeterministic decisions and restarts, just like the relationship between CDCL and Resolution. 

\subsection{Backjump optimality of clause learning} \label{ssec:backjump}

The goal of this subsection is to develop a sense in which \cref{alg:trail-basis} can honestly be called 1-UIP clause learning, in that the output of the algorithm meets some optimality criterion based on an appropriate generalization of UIPs to XNF reasoning. Because a classical UIP (defined via an implication graph) is exactly the negation of an asserted literal of an asserting clause (see \cref{sec:classic}), one reasonable optimality condition could be that whenever $f$ is the equation propagable from an earlier level, $\lnot f$ is as close to the conflict as possible, now measured by the order of trail equations rather than an implication graph. Of course, $\lnot f$ may only lie in $\linspan(U)$ and not syntactically in the trail $U$ itself, so we should look at the earliest trail prefix $U_t = \{f_1, \dots, f_t\}$ such that $\lnot f \in \linspan(U_t)$ and associate $\lnot f$ with index $t$. However, this optimality condition is \emph{not} maintained by our algorithms, refuting a conjecture of \cite{DBLP:journals/corr/abs-2605-15002}, as shown in the simple example below. 

\begin{example} \label{ex:out-of-order}
Consider the following three clauses.
\begin{align*}
    D_1 &= \eq{a=1} \lor \eq{b = 0} \\
    D_2 & = \eq{a \oplus b = 1} \lor \eq{c = 0} \\
    D_3 & = \eq{b = 1} \lor \eq{c = 1}
\end{align*}
If we decide $\eq{a=0}$, then unit propagation with $D_1$ gives $\eq{b = 0}$, and then unit propagation with $D_2$ gives $\eq{c=0}$, and then $D_3$ is a conflict. One can check that our 1-UIP clause learning algorithm returns $\eq{a=1}$ as the asserting clause, which is farthest from the conflict. But $C = \eq{b \oplus c = 1}$ is also an asserting clause: It satisfies $\linneg{C} \subseteq \linspan(\eq{a = 0}, \eq{b = 0}, \eq{c = 0})$ and is associated with the index right before conflict. And, unit propagation from $\eq{b \oplus c = 0}$ deduces $\eq{c = 1}$ from $D_3$, followed by $\eq{a \oplus b = 1}$ from $D_2$, and a conflict is reached with $D_1$. 
\end{example}

The situation in \cref{ex:out-of-order} cannot occur in classical CDCL, where it can be shown that reason clauses can only be used with the same antecedents-and-propagated-literal decomposition that they were originally used with. In the example above, the equation $\eq{b \oplus c = 0}$ allowed the same clauses to be used in a way that is different from the original propagation. 

Nevertheless, there is an arguably more important optimality condition satisfied by classical 1-UIP clause learning, that among all asserting clauses that can be learned, the clause found by the 1-UIP algorithm has the largest backjump. (While \cite{DBLP:conf/sat/AudemardBHJS08} proved this only for an implication-graph-specific definition of asserting clause, it is easy to generalize. See \cref{sec:classic}.) \cref{ex:out-of-order} is not a counterexample for this idea, because both $\eq{a = 1}$ and $\eq{b \oplus c = 1}$ result in backjumps to the root. This property turns out to remain generally true for our XNF-based 1-UIP clause learning, and is the reason we retain this name.

\begin{convention}
    Continue to assume the conventions regarding the trail from \cref{conv:trail}. We add a few more conventions and notations to the list.
    \begin{itemize}
        \item Let $L_k$ denote the set of trail units at decision level $\le k$. 
        \item Let the total number of decisions be $d$.
        \item Define the \emph{position} $\pos(g)$ of an equation $g \in \linspan^\pm(U_m)$ to be the least index $t$ such that $g \in \linspan^\pm(U_t)$. (In particular, for $g \in \linspan(U_m)$ and $t$ not the position of a decision, the clause $R_t$ propagates $g$ from $U_{t-1}$ if and only if $\pos(g) = t$.) 
        \item Define $R_i'$ by $\linneg{R_i'} = \linneg{R_i} \cap \linspan(U_{i-1})$, and assume the Horn-form trail convention by which trail representatives $f_i$ are chosen such that reason clauses may be written as $R_i \equiv R_i' \lor f_i$.  
        \item In \cref{alg:assert}, we refer to the value of $C$ at the start of iteration $i$ by $C_i$, and we define $C_i'$ by $\linneg{C_i'} = \linneg{C_i} \cap \linspan(U_{i-1})$. In particular, $C_{i-1}$ is typically the resolvent of $C_i$ and $R_i$, or we have $C_i = C_{i-1}$ in iterations that skip resolving with $R_i$. 
    \end{itemize}
\end{convention}

We make the following fixed-point assumption as well. Note that even for classical CDCL, this fixed-point assumption is necessary for the 1-UIP algorithm to find an asserting clause with largest backjump.

\begin{enumerate}[{(FP)}]\def\theenumi{\textup{(FP)}}
    \item\label[assumption]{as:fp} Assume that unit propagation reaches fixed point
          at the end of every level, so that each $L_k$ cannot unit propagate any further using the clause database, which in particular, contains the reasons $\mathcal{R}$.
\end{enumerate}

First, we define UIPs. Generally, a UIP is an equation that allows deduction of $\bot$ when paired with an earlier level. However, a wrinkle is that there are many different kinds of ``deduction'' that can occur. 

\begin{definition}
    Let $k < d$ be a level before the last level. We say that an equation $g \in \linspan(U_m)$ is
    \begin{itemize}
        \item a \emph{replay UIP for level $k$} if a trail replay starting with $L_k$ and $g$ finds a contradiction, where a \emph{trail replay} is unit propagation on $R_1, \dots, R_m, R_\bot$ once each, in that original order,
        \item a \emph{unit propagation UIP for level $k$} if unit propagation on $\mathcal{R}$ starting from $L_k$ and $g$ to fixed point finds a contradiction (with clauses being chosen in any order), and
        \item a \emph{semantic UIP for level $k$} if $L_k \cup \mathcal{R} \cup \{g\}$ is unsatisfiable. 
    \end{itemize}
\end{definition}

Note that UIPs cannot belong to $\linspan(L_k)$. For replay UIPs and unit propagation UIPs, this is clear by the fixed-point assumption \cref{as:fp}, which prohibits deducing $\bot$ via unit propagation from any earlier levels. For semantic UIPs, \cref{as:fp} cannot be immediately applied, but we will show that semantic UIPs nevertheless cannot belong to $\linspan(L_k)$ in \cref{cor:uip-not-in-lk}. It will generally be a useful tool to have a translation of \cref{as:fp} for use in semantic settings.

A naive translation of the assumption could ask that $L_k$ and $\mathcal{R}$ (or even $L_k$ and the clause database, which contains $\mathcal{R}$) have no semantic consequences that are equations outside of $\linspan(L_k)$. However, that would be clearly false in general, since any unit asserting clause is implied by $\mathcal{R}$ alone and is outside of $\linspan(L_k)$ for all $k$. We can rescue a weaker version that considers $\mathcal{R}_{\le m}$ instead of $\mathcal{R}$, and this is enough for us.

\begin{proposition}[Semantic fixed point property]\label{prop:sem-fp}
    Assume \cref{as:fp} and let $k < d$ be any level except the last, so that $L_k$ is missing at least one decision.
    \begin{enumerate}
        \item Suppose that $(L_k \land \mathcal{R}_{\le m}) \Rightarrow g$ for some equation $g$. Then $L_k \Rightarrow g$.  
        \item For any $i \le m$ and any clause $C$ not falsified by $L_k$, the set $L_k \cup \mathcal{R}_{\le i} \cup \{C\}$ is satisfiable. In particular, applied at $i = m$ and $C = R_\bot$, we conclude that $L_k \cup \mathcal{R}$ is satisfiable.
    \end{enumerate}
\end{proposition}
\begin{proof}
    We note that (2) follows easily from (1):
    The set $L_k \cup \mathcal{R}_{\le i} \cup \{C\}$ being unsatisfiable can be rearranged into $(L_k \land \mathcal{R}_{\le i}) \Rightarrow \linneg{C}$. Then (1) would imply that $C$ is falsified by $L_k$ alone, which is a contradiction. 
    
    For the main proof of (1), we first observe that any $g$ such that $(L_k \land \mathcal{R}_{\le m}) \Rightarrow g$ must satisfy $g \in \linspan(U_m)$, because every assignment that satisfies $U_m$ also satisfies $L_k \cup \mathcal{R}_{\le m}$, and hence satisfies $g$. This allows us to consider $\pos(g)$. 
    
    Suppose for contradiction that there exists some
    $g$ such that $(L_k \land \mathcal{R}_{\le m}) \Rightarrow g$ but
    $L_k \not \Rightarrow g$.
    Among such choices of $g$, choose the one with smallest value of
    $t = \pos(g)$. Position $t$ must come after level $k$ because $L_k \not \Rightarrow g$. Write $g = f_t + u$ for some $u \in \linspan(U_{t-1})$.

    The assumption that $(L_k \land \mathcal{R}_{\le m}) \Rightarrow g$ is equivalent to the set 
    $\Delta = L_k  \cup \mathcal{R}_{\le m} \cup \{\lnot g\}$ being unsatisfiable.
    Project this set onto $\linspan^\pm(U_{m-1}), \dots, \linspan^\pm(U_{t-1})$ in sequence. The result of the projection $\Delta'$ is given by:
    \begin{equation*}
        \Delta' = \begin{cases}
            L_k \cup \mathcal{R}_{\le t-1} & \text{if $f_t$ is a decision} \\
            L_k \cup \mathcal{R}_{\le t-1} \cup \{R_t' \lor \lnot u\} & \text{if $f_t$ is propagated}
        \end{cases}
    \end{equation*}
    which can be seen as follows: Recall that $g \in \linspan^\pm(U_t)$. Thus, projecting onto the subspaces $\linspan^\pm(U_{m-1}), \dots, \linspan^\pm(U_t)$ simply removes $R_m, \dots, R_{t+1}$, with $g$ unaffected. In the final projection onto $\linspan^\pm(U_{t-1})$, suppose that $f_t$ is a decision. Then $\mathcal{R}_{\le t-1}=\mathcal{R}_{\le t}$, and $g$ is the only item involving $f_t$ which is not in $\linspan^\pm(U_{t-1})$, so it is removed. Otherwise, $f_t$ was propagated by $R_t \equiv R_t' \lor f_t$. The resolvent of $R_t$ and $\lnot g = \lnot f_t + u$ is exactly $R_t' \lor \lnot u$, as desired. In all cases, the final projection is unsatisfiable by repeated application of \cref{prop:proj-eqsat}.
    
    If $f_t$ was a decision, the fact that $L_k \cup \mathcal{R}_{\le t-1}$ is unsatisfiable contradicts the fact that $L_k \cup \mathcal{R}_{\le m}$ is satisfied by $U_m$.

    If $f_t$ was propagated, we can rearrange the fact that the final projection is unsatisfiable to get that $(L_k \land \mathcal{R}_{\le t-1}) \Rightarrow \linneg{R_t' \lor \lnot u}\Rightarrow \linneg{R_t'}$.
     Consider any $h\in \linneg{R_t'} \subseteq \linspan(U_{t-1})$.
     Then $(L_k \land \mathcal{R}_{\le t-1})\Rightarrow h$,
     and by the minimality of $t=\pos(g)$, we get
     $L_k\Rightarrow h$. 
     This means that $L_k$ falsifies $R_t'$,
     which contradicts \cref{as:fp} because there should not be any propagations from $R_t$ given $L_k$. 
\end{proof}

\begin{corollary} \label{cor:uip-not-in-lk}
    Assume \cref{as:fp}. No UIP $g$ (of any kind) for level $k$ has $g \in \linspan(L_k)$. 
\end{corollary}
\begin{proof}
    For replay and unit propagation UIPs, this follows easily from \cref{as:fp}. For semantic UIPs, the set $L_k \cup \mathcal{R}$ is satisfiable by \cref{prop:sem-fp}(2), while $L_k \cup \mathcal{R} \cup \{g\}$ should be unsatisfiable.
\end{proof}

To motivate why we care about all three variants of UIPs, consider the following. It is immediate that an equation $g$ is a unit propagation UIP for level $k$ if and only if $\lnot g$ is one of the asserted equations for an asserting clause with asserting level $\le k$. Thus, the goal is to find the smallest $k$ for which a unit propagation UIP exists. However, the 1-UIP clause learning algorithm actually corresponds more closely to replay UIPs, since trail replay is exactly the inverse process of resolution in reversed trail order. 
The various possible orders of unit propagation make it difficult to reason about unit propagation UIPs directly, so it turns out to be useful to reason about semantic UIPs instead. 

These three different kinds of UIPs are all distinct for $\CDCL(\oplus)$ (\cref{ex:out-of-order} gives a unit propagation UIP that is not a replay UIP, and \cref{ex:sem-confl} gives a semantic UIP that is not a unit propagation UIP), but our ultimate claim is that their \emph{existence} for each level must be the same. All three notions are equivalent for classical CDCL (\cref{sec:classic}). In order to develop an intuition for these three kinds of UIPs, we spend some time understanding their structure and relationships before stating and proving the main theorem.

Our first characterization defines upper and lower limits for what replay from $g$ can propagate. In particular, replay from $g$ alone cannot propagate with reasons at or before $\pos(g)$, and afterwards, each $R_i$ can only propagate $f_i$ (no restriction on what $R_\bot$ can propagate). For example, continuing \cref{ex:out-of-order}, consider replay from $\eq{b \oplus c = 0}$. Then $D_1 = \eq{a = 1} \lor \eq{b = 0}$ fails to propagate, $D_2 = \eq{a \oplus b = 1} \lor \eq{c = 0}$ fails to propagate, and $D_3 = \eq{b = 1} \lor \eq{c = 1}$ propagates $\eq{b = 1}$ (or $\eq{c = 1}$) but fails to find a contradiction. This is unlike full unit propagation, which finds a contradiction. 

However, suppose that one starts replay with not just $g$, but also with all semantic consequences of $g$ with reasons $\mathcal{R}_{\le t}$ and $L_k$, where $t = \pos(g)$. These are what replay missed by not being able to propagate anything using $R_1, \dots, R_t$ in sequence. Our claim is that assuming these earlier misses alone, replay is semantically complete and will not miss any later deductions. In the same example, one can see that $g = \eq{b \oplus c = 0}$ with clauses $D_1 = \eq{a = 1} \lor \eq{b = 0}$ and $D_2 = \eq{a \oplus b = 1} \lor \eq{c = 0}$ semantically implies $\eq{b = 0}$ (and $\eq{c = 0}$), because $\eq{b=1}$ implies $\eq{c = 1}$ using $g$ and $\eq{a = 1}$ using $D_1$, which contradicts $D_2$. Adding $\eq{b = 0}$ to the starting units during replay certainly allows contradiction to be derived upon reaching $D_3$.

\begin{definition}
    Let $g \in \linspan(U_m)$ and suppose that position $t = \pos(g)$ occurs after level $k$. Let $i \in \{t, \dots, m\}$. We define
    \begin{equation*}
        \Conseq_{\le i}(g) = \{h \in \mathbb{F}_2^{n+1} \mid (L_k \land \mathcal{R}_{\le i} \land g) \Rightarrow h\}.
    \end{equation*}
\end{definition}
\begin{fact}
    We always have $\Conseq_{\le i}(g) \subseteq \linspan(U_i)$. 
\end{fact}
\begin{proof}
    Since $t \le i$, we just observe that $U_i$ satisfies $L_k \land \mathcal{R}_{\le i} \land g$, and thus $U_i$ satisfies every consequence.
\end{proof}

\begin{proposition}[Upper and lower limits of replay] \label{prop:replay}
Assume \cref{as:fp}. Let $g \in \linspan(U_m)$ and suppose that position $t = \pos(g)$ occurs after level $k$. 
\begin{enumerate}
    \item In trail replay from $L_k$ and $g$, the clauses 
    $\mathcal{R}_{\le t}$ do not propagate any equations, and if a later clause $R_i$ with $i > t$ does propagate an equation at all, then $R_i$ propagates $f_i$. (We do not constrain what $R_\bot$ propagates.)
    \item In trail replay from $\Conseq_{\le t}(g)$, let $U_i'$ denote the trail after propagation on $R_i$.
    Then for all $i \in \{t, \dots, m\}$, we have $\linspan(U_i') = \Conseq_{\le i}(g)$.
\end{enumerate}
\end{proposition}

\begin{proof}
    To prove (1), first consider clauses $R_i$ in $\mathcal{R}_{\le t}$.  Because $g$'s position is later in the trail, only $L_k$ can be helpful for propagating from $R_i$. However, $R_i$ cannot propagate from $L_k$ by \cref{as:fp}. Additionally, $R_t$ cannot propagate anything because $g = f_t + u$ for $u \in \linspan(U_{t-1})$ has the wrong polarity of $f_t$ to help propagate something from $R_t$. For later clauses that do propagate, if $U_i'$ denotes the trail after propagation on $R_i$ we inductively have $\linspan(U_{i-1}') \subseteq \linspan(U_{i-1})$, so the propagation must be caused by $R_i'$ being falsified, thus propagating $f_i$. 

    To prove (2), we first observe that (1) together with soundness implies that $\linspan(U_i') \subseteq \Conseq_{\le i}(g)$. To show that $\linspan(U_i') \supseteq \Conseq_{\le i}(g)$, we induct on $i$. For the base case $i = t$, just note that the trail does not grow while processing $R_1, \dots, R_t$ by definition of $\Conseq_{\le t}(g)$. For $i > t$, consider $h \in \Conseq_{\le i}(g) \subseteq \linspan(U_i)$. Rearranging the definition, we get that the set $\Delta = L_k \cup \mathcal{R}_{\le i} \cup \{g, \lnot h\}$ is unsatisfiable. Project this set onto $\linspan^\pm(U_{i-1})$. The result of the projection $\Delta'$ is given by:
    \begin{equation*}
        \Delta' = \begin{cases}
            L_k \cup \mathcal{R}_{\le i -1} \cup \{g, \lnot h\}  & \text{if $h \in \linspan(U_{i-1})$} \\
            L_k \cup \mathcal{R}_{\le i-1} \cup \{g\} & \text{if $h \not \in \linspan(U_{i-1})$ and $f_i$ is a decision} \\
            L_k \cup \mathcal{R}_{\le i-1} \cup \{g, R_i' \lor \lnot u\} & \text{if $h \not \in \linspan(U_{i-1})$ and $f_i$ is propagated}
        \end{cases}
    \end{equation*}
    where $h = f_i + u$ for some $u \in \linspan(U_{i-1})$ in the case that $h \not \in \linspan(U_{i-1})$, which can be seen as follows: 
    First note that $g$ is always unaffected by projection because $i-1 \ge t$ in the inductive step. In the case that $h \in \linspan(U_{i-1})$, the equation $\lnot h$ is unaffected by projection, and $R_i$ is the only remaining item involving $f_i$, so it can be deleted. In the case that $h = f_i + u$ for a decision $f_i$, we have $\mathcal{R}_{\le i} = \mathcal{R}_{\le i-1}$ and only $\lnot h$ involves $f_i$, so it is deleted. The final case is the only case where resolution happens, where $R_i \equiv R_i' \lor f_i$ resolves with $\lnot h = \lnot f_i + u$ to get $R_i' \lor \lnot u$. In all cases, the final projection is unsatisfiable by repeated application of \cref{prop:proj-eqsat}.
    
    Continuing from the first case where $h \in \linspan(U_{i-1})$, we rearrange $L_k \cup \mathcal{R}_{\le i -1} \cup \{g, \lnot h\}$ being unsatisfiable to get that $(L_k \land \mathcal{R}_{\le i-1} \land g) \Rightarrow h$, or in other words, $h \in \Conseq_{\le i-1}(g)$. Thus, by induction, $h \in \linspan(U_{i-1}') \subseteq \linspan(U_i')$. 
    
    In the case where $h \not \in \linspan(U_{i-1})$ and $f_i$ is a decision, we have that $L_k \cup \mathcal{R}_{\le i-1} \cup \{g\}$ is unsatisfiable. But one can check that each part of $L_k \cup \mathcal{R}_{\le i-1} \cup \{g\}$ is satisfied by $U_{i-1}$, so this case cannot occur. 
    
    In the final case where $h = f_i + u$ and $f_i$ is propagated from $R_i \equiv R_i' \lor f_i$, we have that $L_k \cup \mathcal{R}_{\le i-1} \cup \{g, R_i' \lor \lnot u\}$ is unsatisfiable. This can be rearranged into $(L_k \land \mathcal{R}_{\le i - 1} \land g ) \Rightarrow \linneg{R_i' \lor \lnot u} \Rightarrow (\linneg{R_i'} \land u)$. By induction, this means that $R_i'$ is falsified by $U_{i-1}'$ and $u \in \linspan(U_{i-1}')$. Thus, $R_i$ propagates $f_i$ from $U_{i-1}'$, and $h = f_i + u \in \linspan(U_i')$. 
\end{proof}

\begin{corollary}[Condition for semantic UIPs to be replay UIPs] \label{cor:sem-to-replay}
    Assume \cref{as:fp}. Let $g \in \linspan(U_m)$ and suppose that $t = \pos(g)$ occurs after level $k$. If $\Conseq_{\le t}(g)$ is trivial, meaning
    \begin{equation*}
        \Conseq_{\le t}(g) = \linspan(L_k, g),
    \end{equation*}
    then $g$ is a replay UIP for level $k$ if and only if $g$ is a semantic UIP for level $k$.
\end{corollary}
\begin{proof}
    One direction is trivial. If $g$ is a semantic UIP, then replay from $L_k$ and $g$ will find $\Conseq_{\le m}(g)$ after processing $R_m$, and adding $R_\bot$ makes this unsatisfiable by virtue of $g$ being a semantic UIP. By completeness of unit propagation on a single clause and \cref{prop:replay}(2), replay will find a contradiction when processing $R_\bot$.
\end{proof}

Eventually, we will use \cref{cor:sem-to-replay} to find a replay UIP given a semantic UIP. Our next property connects replay UIPs to the 1-UIP clause learning algorithm, so that together, the existence of a semantic UIP will allow us to say something about the asserting clause learned by the 1-UIP clause learning algorithm, in particular, that it has the largest possible backjump. While admittedly vague at the moment, remember that the condition about $\Conseq_{\le t}(g)$ being trivial in \cref{cor:sem-to-replay} will be important for later. 

More specifically, the next property says that replay UIPs correspond precisely to asserting clauses learned by the resolution process in \cref{alg:assert}. To emphasize, they do not correspond to \emph{all} asserting clauses, but they specifically do correspond to the ones learnable by \cref{alg:assert}. As a corollary, replay UIPs are unique at each position up to $L_k$, and replay from earlier ones derives all later ones.

\begin{proposition}[Replay UIPs and resolution on the trail] \label{prop:replay-res}
Assume \cref{as:fp}. Recall that $C_i$ denotes the clause $C$ in iteration $i$ of \cref{alg:assert}, such that either $C_{i-1}$ is the resolvent of $C_i$ and $R_i$, or $C_i = C_{i-1}$ in iterations that skip resolution. 
\begin{enumerate}
    \item Consider a run of \cref{alg:assert} (instantiated with any stopping condition) whose output $C_t$ is asserting with asserting level $k$. If $C_t$ propagates $\lnot g$ from $L_k$, then $g$ is a replay UIP for level $k$.
    \item Let $g$ be a replay UIP for level $k$, and let $t = \pos(g)$. Then the equations $f_m, \dots, f_{t+1}$ are not decisions and running \cref{alg:assert}  instantiated with stopping condition ``$i \le t$,'' which resolves on $f_m, \dots, f_{t+1}$, yields an output $C_t$ that is an asserting clause that propagates $\lnot g$ from $L_k$ or earlier. Furthermore, for all $i \ge t$, if $C_i$ is asserting, then its asserting level is at most $k$.
\end{enumerate}
\end{proposition}

\begin{proof}
    To prove (1), assume that $C_t$ propagates $\lnot g$ from $L_k$, so $L_k \cup \{g\}$ falsifies $C_t$. Consider replay from $g$ and $L_k$, and let $U_i'$ denote the trail after propagation on $R_i$. We claim that $U_i'$ falsifies $C_i$ for all $i \ge t$ by induction, where the base case with $U_t' = L_k \cup \{g\}$ is above. Then for the inductive step, the case $C_i = C_{i-1}$ is trivial, and whenever $U_{i-1}'$ falsifies $C_{i-1} = C_i' \lor R_i' \lor \lnot u$ where $C_i \equiv C_i' \lor (\lnot f_i + u)$ for some $u \in \linspan(U_{i-1})$, we see that $C_i'$ is already falsified, $u$ is already satisfied, and $f_i$ is satisfied by propagating from $R_i$, now that $R_i'$ is falsified. Thus, $U_m'$ falsifies $C_m = R_\bot$, so replay finds a contradiction. 

    To prove (2), again let $U_i'$ denote the trail after propagation on $R_i$. Indeed, $f_m, \dots, f_{t+1}$ cannot be decisions, because replay finding a contradiction from an earlier level is a violation of \cref{as:fp}. To show that $C_t$ propagates $\lnot g$ from $L_k$, it suffices to again inductively show the following property.
    \begin{claim} \label{clm:replay-inv}
        For all $i \ge t$, we have that $U_i'$ falsifies $C_i$.
    \end{claim}
    In particular, applied at $i = t$, we conclude that $\linspan(U_t') = \linspan(L_k, g)$ falsifies $C_t$. This is almost enough to conclude that $C_t$ propagates $\lnot g$ from $L_k$, except we also need that $C_t$ is not falsified by $L_k$ alone. This follows simply because $C_t \equiv C_j' \lor R_j' \lor \lnot v$ for some $j$ and some $v \in \linspan(U_{j-1})$ (where $j = t+1$ unless resolution was skipped at this step), and $R_j$ did not unit propagate from $L_k$ by \cref{as:fp}. Additionally, all asserting clauses $C_i$ with $i \ge t$ are falsified by $U_m'$, and because $U_m'$ cannot involve equations at positions after level $k$ and before the last level (in fact, before $t$) by \cref{prop:replay}(1), we conclude that their asserting level must be at most $k$. 

    \begin{proof}[Proof of \cref{clm:replay-inv}]
        The base case is $C_m = R_\bot$, where replay is required to falsify $R_\bot$ and not anything earlier because $L_k \cup \mathcal{R}_{\le m} \cup \{g\}$ is satisfied by $U_m$. Then going down, in the case that $C_{i-1} = C_i$, this situation can only occur when $C_i$ is falsified by $U_{i-1}$. By \cref{prop:replay}(1), the clause $R_i$ can only propagate $f_i$, so propagating with it in replay was useless, and $U_{i-1}'$ falsifies $C_{i-1}$. In the case that $C_{i-1} = C_i' \lor R_i' \lor \lnot u$ where $C_i \equiv C_i' \lor (\lnot f_i + u)$ for some $u \in \linspan(U_{i-1})$, to show that $C_{i-1}$ is falsified by $U_{i-1}'$, it suffices to check that $C_i'$, $R_i'$, and $\lnot u$ are all falsified by $U_{i-1}'$. 
        
        By induction, $C_i$ is falsified by $U_i'$, so $C_i'$ is also falsified by $U_i'$. Again by \cref{prop:replay}(1), the only effect of $R_i$ in replay is to deduce $f_i$, and $C_i'$ does not involve $f_i$, so $C_i'$ is falsified by $U_{i-1}'$. But $C_i$ does involve $f_i$, and is falsified by $U_i'$ by induction, so it must have been propagation on $R_i$ that added $f_i$ to $U_{i-1}'$, thus $R_i'$ is falsified. Finally, $f_i + u \in \linspan(U_{i}')$ because $C_i \equiv C_i' \lor (\lnot f_i + u)$ is falsified, hence $u = f_i + (f_i + u) \in \linspan(U_{i}')$. Since $u \in \linspan(U_{i-1})$, we conclude $u \in \linspan(U_{i-1}')$ as desired.
    \end{proof}
    
    This concludes the main proof as well.
\end{proof}

\begin{corollary}[Structure of replay UIPs]
    Assume \cref{as:fp}. Let $g$ and $h$ be replay UIPs for level $k$. 
    \begin{enumerate}
        \item If $\pos(g) < \pos(h)$, then replay starting from $L_k$ and $g$ derives $h$.
        \item If $\pos(g) = \pos(h)$, then $g$ and $h$ are equivalent up to $\linspan(L_k)$; that is,
        $\linspan(L_k, g)=\linspan(L_k, h)$.
    \end{enumerate}
\end{corollary}
\begin{proof}
    By \cref{prop:replay-res}, both replay UIPs can be found by the same run of \cref{alg:assert}, so that replay from the earlier one finds the later one. If they are at the same position, the same clause $C_t$ propagates both $\lnot g$ and $\lnot h$ from $L_k$, so $g + h \in \linspan(L_k)$.
\end{proof}

The above corollary states that up to equivalence under $L_k$, replay UIPs are unique per position and hence linearly independent. Next, we paint the picture regarding semantic UIPs, showing that they generally form an \emph{affine subspace}, and are thus generally exponentially more numerous than replay UIPs. To understand this structure, we actually have to look at negations of UIPs. 

For example, continuing \cref{ex:out-of-order}, we saw that $\eq{a = 0}$ and $\eq{b \oplus c = 0}$ are semantic UIPs, so let us consider the span of their negations, $V = \linspan(\eq{a = 1}, \eq{b \oplus c = 1})$. We have that $\eq{a \oplus b \oplus c = 0} \in V$, and one can check that although its negation $\eq{a \oplus b \oplus c = 1}$ lives outside of the original trail $U = \{\eq{a = 0}, \eq{b = 0}, \eq{c = 0}\}$ and hence cannot be a UIP by definition, it still acts like a ``non-trail UIP,'' in the sense that it is incompatible with the clauses. In particular, it allows $D_2 = \eq{a \oplus b = 1} \lor \eq{c = 0}$ to imply $\eq{c = 0}$, which allows $D_3 = \eq{b = 1} \lor \eq{c = 1}$ to imply $\eq{b = 1}$, which allows $D_1 = \eq{a = 1} \lor \eq{b = 0}$ to imply a contradiction. One can similarly consider $\bot$ to be another trivial ``non-trail UIP,'' corresponding to $\top \in V$. So, the linear subspace $V$ consists of negations of equations that are incompatible with the clauses, and the affine subspace containing negations of the true semantic UIPs, here $\{\eq{a = 1}, \eq{b \oplus c = 1}\}$, is exactly half of $V$.

\begin{proposition}[Structure of semantic UIPs] \label{prop:sem-uip-struc} Assume \cref{as:fp}. Let $k < d$ be any level before the last. Let 
\begin{equation*}
    V = \{\lnot g \in \mathbb{F}_2^{n+1} \mid L_k \cup \mathcal{R} \cup \{g\} \text{ is unsatisfiable}\}.
\end{equation*}
Define $W^+ = V \cap \linspan(U_m)$ and $W^- = V \cap (\bot+\linspan(U_m))$, where $\bot + \linspan(U_m)$ is the coset $\{\lnot f \mid f \in \linspan(U_m)\}$.
    \begin{enumerate}
        \item $V$ is a linear subspace of $\linspan^\pm(U_m)$, and hence $V$ is the disjoint union of $W^+$ and $W^-$.
        \item If some semantic UIP $g$ exists, then the semantic UIPs are exactly $\{g + h \mid h \in W^+\}$, an affine subspace of $\linspan(U_m)$,
        and $W^-$ is the set of negated semantic UIPs.
    \end{enumerate}
\end{proposition}
\begin{proof}
    To prove (1), first note that $L_k \cup \mathcal{R}$ is satisfiable by \cref{prop:sem-fp}. Furthermore, all equations discussed in $L_k \cup \mathcal{R}$ belong to $\linspan^\pm(U_m)$. Thus, for $L_k \cup \mathcal{R} \cup \{g\}$ to be unsatisfiable, it must be that $g \in \linspan^\pm(U_m)$. It remains to show that $V$ is a linear subspace. We just note that $V$ is closed under addition because if $L_k \land \mathcal{R} \Rightarrow \lnot g$ and $L_k \land \mathcal{R} \Rightarrow \lnot h$, then $L_k \land \mathcal{R} \Rightarrow \lnot g + \lnot h$. To prove (2), just recall that a semantic UIP is exactly an equation $g \in \linspan(U_m)$ such that $L_k \cup \mathcal{R} \cup \{g\}$ is unsatisfiable.
\end{proof}

A key result for us is that \emph{earliest} semantic UIPs $g$, meaning those of minimal position $t = \pos(g)$, satisfy the previously mentioned requirement in \cref{cor:sem-to-replay} that $\Conseq_{\le t}(g)$ is trivial, so that these are also replay UIPs and hence also unique up to $\linspan(L_k)$. The proof of the following proposition requires a few new ideas, so we delay it for now.

\begin{proposition}[Earliest semantic UIP has trivial $\Conseq_{\le t}(g)$] \label{prop:early-no-conseq}
    Assume \cref{as:fp}. Let $g$ be an earliest semantic UIP for level $k$ and let $t = \pos(g)$. Then,
    \begin{equation*}
        \Conseq_{\le t}(g) = \linspan(L_k, g).
    \end{equation*}
\end{proposition}

With these structural tools, the main theorem, that the 1-UIP clause learning algorithm finds an asserting clause of largest possible backjump, follows easily.

\begin{theorem}[Backjump optimality theorem] \label{thm:backjump}
    Assume \cref{as:fp}. Let $k < d$ be a level before the last level. The following are equivalent:
    \begin{enumerate}
        \item A replay UIP for level $k$ exists.
        \item A unit propagation UIP for level $k$ exists.
        \item A semantic UIP for level $k$ exists.
        \item The 1-UIP clause learning algorithm outputs a clause with asserting level $\le k$.
    \end{enumerate}
    Thus, the clause learned by the 1-UIP clause learning algorithm has the lowest asserting level among all asserting clauses. 
\end{theorem}
\begin{proof}
    The implications $(1) \Rightarrow (2) \Rightarrow (3)$ are obvious. The implication $(3) \Rightarrow (1)$ is given by considering an earliest semantic UIP of level $k$, which by \cref{prop:early-no-conseq} and \cref{cor:sem-to-replay} is a replay UIP of level $k$. The equivalence $(4) \Leftrightarrow (1)$ is given by \cref{prop:replay-res}. The final claim that we learn a clause of lowest asserting level follows from $(2) \Leftrightarrow (4)$, which we now know to be true, along with the easy correspondence between unit propagation UIPs and asserting clauses. 
\end{proof}

Before we jump into the proof of \cref{prop:early-no-conseq}, we sketch some of its ideas to motivate a key definition.
To show that $\Conseq_{\le t}(g) \subseteq \linspan(L_k, g)$, we suppose for contradiction that there exists $f \in \Conseq_{\le t}(g) \setminus \linspan(L_k, g)$, and we use such an $f$ to find a semantic UIP earlier than $g$, thus deriving a contradiction. While it's not too hard to find such an $f$ earlier than $g$, we run into trouble trying to turn such an earlier $f$ into a UIP. 

We resolve this issue by forcing $f$ to be a UIP in an alternate universe with a different set of reason clauses. In particular, if the reason clauses $\mathcal{R}$ were swapped out for $\mathcal{R}_{\le t} \cup \{\lnot f\}$, then $f$ would trivially be a ``UIP,'' meaning that $f$ with the new set of reasons (and $L_k$) is unsatisfiable. This universe is nice because it allows $g$ to remain a semantic UIP here as well, since $f$ is a consequence of $g$ using $L_k \cup \mathcal{R}_{\le t}$. After establishing the existence of UIPs in the $\mathcal{R}_{\le t} \cup \{\lnot f\}$ universe, we can now take the earliest UIP $g'$ in this universe, and get the final contradiction from showing that $g'$ is a UIP in the original universe. 

The final part of the above argument will rely on a key property, that if $g'$ is an earliest UIP in the $\mathcal{R}_{\le t} \cup \{\lnot f\}$ universe, then any other UIP $h'$ in this universe (e.g., $f$ or $g$) satisfies $h' \in \Conseq_{\le t}(g')$. In other words, earlier semantic UIPs imply later ones. We will prove this property by induction, which will require us to strengthen the inductive hypothesis to consider not just $\mathcal{R}_{\le t} \cup \{\lnot f\}$, but any universe of the form $\mathcal{R}_{\le i} \cup \{C\}$ where $C$ is a clause falsified by $U_i$ but not by $L_k$. This motivates the following definition.

\begin{definition}
    Let $C$ be a clause falsified by $U_i$ but not by $L_k$. We say that an equation $g \in \linspan(U_i)$ is a \emph{semantic UIP with respect to $\mathcal{R}_{\le i} \cup \{C\}$ for level $k$} if 
    $L_k\cup \mathcal{R}_{\le i} \cup \{C\} \cup \{g\}$ is unsatisfiable.
\end{definition}

As mentioned earlier, the following is the key property that we will need in our arguments. We state it generally, and we will apply it with $i = t$ and $C = \lnot f$.

\begin{proposition}[Earliest semantic UIP implies later UIPs] \label{prop:early-impl-later}
    Assume \cref{as:fp}. Let $C$ be a clause falsified by $U_i$ but not by $L_k$. Let $g \in \linspan(U_i)$ be an earliest semantic UIP with respect to $\mathcal{R}_{\le i} \cup \{C\}$ for level $k$. Then for every other such UIP $h \in \linspan(U_i)$, we have $h \in \Conseq_{\le i}(g)$.
\end{proposition}

Before we prove \cref{prop:early-impl-later}, we show how it can be used to derive \cref{prop:early-no-conseq}. As a reminder, the goal is to prove that for an earliest semantic UIP $g$ (with respect to $\mathcal{R}$), we have $\Conseq_{\le t}(g) = \linspan(L_k, g)$. 

\begin{proof}[Proof of \cref{prop:early-no-conseq}]
Let $g$ be an earliest semantic UIP for level $k$ (with respect to $\mathcal{R}$) and let $t = \pos(g)$. 
Suppose for contradiction that there exists $f^* \in \Conseq_{\le t}(g) \setminus \linspan(L_k, g)$. Observe that also $f^* + g \in \Conseq_{\le t}(g)$. Furthermore, both $f^*$ and $f^* + g$ belong to $\linspan(U_t)$, and $\pos(g) = t$, hence one of $f^*$ and $f^* + g$ must belong to $\linspan(U_{t-1})$. Denote this equation in $\linspan(U_{t-1})$ by $f$. 

Note that $f$ is trivially a semantic UIP for level $k$ with respect to
$\mathcal{R}_{\le t} \cup \{\lnot f\}$. This is nice because $g$ is not only a semantic UIP with respect to $\mathcal{R}$, but also a semantic UIP with respect to $\mathcal{R}_{\le t} \cup \{\lnot f\}$ because $f \in \Conseq_{\le t}(g)$. In any case, now that we have established the existence of such UIPs, there must be some earliest semantic UIP $g'$ with respect to
$\mathcal{R}_{\le t} \cup \{\lnot f\}$. Furthermore, $\pos(g') \le \pos(f) \le t - 1$. 

We now apply \cref{prop:early-impl-later} with
$i = t$ and $C=\lnot f$ to derive that $g \in \Conseq_{\le t}(g')$. 
Since $g$ is a semantic UIP for level $k$ with respect to $\mathcal{R}$, 
$L_k\cup \mathcal{R}\cup \{g\}$ is unsatisfiable. Now that $g$ is a consequence of $g'$ assuming $L_k$ and $\mathcal{R}_{\le t} \subseteq \mathcal{R}$, we conclude that
$L_k\cup \mathcal{R}\cup \{g'\}$ must also be unsatisfiable, so that $g'$ is a semantic UIP for level $k$
with respect to $\mathcal{R}$. The fact that 
$\pos(g')\le t-1$ now contradicts the choice of $g$ 
as the earliest semantic UIP.
\end{proof}

We finish this section with the proof of
\cref{prop:early-impl-later}, which is the last 
remaining loose end of the proof of \cref{thm:backjump}.
As a reminder, we suppose that $C$ is falsified by $U_i$ but not by $L_k$, and the goal is to prove that when $g \in \linspan(U_i)$ is the earliest semantic UIP with respect to $\mathcal{R}_{\le i} \cup \{C\}$, any other such UIP $h \in \linspan(U_i)$ satisfies $h \in \Conseq_{\le i}(g)$. 

\begin{proof}[Proof of \cref{prop:early-impl-later}]
    We induct on $i$ starting at level $k$. 
    Within level $k$ as the base case, there is nothing to prove, because at such an $i$, no clause $C$ can be falsified by $U_i$ but not $L_k$. For the induction then, let $h \in \linspan(U_i)$ be a semantic UIP with respect to $\mathcal{R}_{\le i} \cup \{C\}$. 
    We wish to prove that $h \in \Conseq_{\le i}(g)$.
    We consider two cases based on $\pos(h)$. 
    
    \medskip
    \noindent\textsc{Case $\pos(h)\le i-1$:}
    We will need to apply the fact that this case holds in the second case also, so we state it as a formal claim:

    \begin{claim}\label{clm:early-pos}
        When $\pos(h) \le i-1$, we have $h \in \Conseq_{\le i}(g)$.
    \end{claim}
    
    \begin{proof}[Proof of claim]
        Project $\Delta = \mathcal{R}_{\le i} \cup \{C\}$ onto $\linspan^\pm(U_{i-1})$. The result of projection $\Delta'$ is given by:
        \begin{equation*}
            \Delta' = \begin{cases}
                \mathcal{R}_{\le i-1} \cup \{C\} & \text{if $C$ is falsified by $U_{i-1}$} \\
                \mathcal{R}_{\le i-1} & \text{if $C$ is not falsified by $U_{i-1}$ and $f_i$ is a decision} \\
                \mathcal{R}_{\le i-1} \cup \{D\} & \text{if $C$ is not falsified by $U_{i-1}$ and $f_i$ is propagated}
            \end{cases}
        \end{equation*}
        where $D$ is the resolvent of $C$ and $R_i$ in the final case; these computations are obvious. 
        
        A key remark is that replacing $\Delta$ by its projection $\Delta'$ onto  the subspace $\linspan^\pm(U_{i-1})$ leaves all semantic UIPs in that target space as semantic UIPs with respect to $\Delta'$. Specifically, because $\pos(g) \le \pos(h) \le i-1$, the sets $L_k \cup \Delta \cup \{g\}$ and $L_k \cup \Delta \cup \{h\}$ project to $L_k \cup \Delta' \cup \{g\}$ and $L_k \cup \Delta' \cup \{h\}$, respectively, preserving unsatisfiability by \cref{prop:proj-eqsat}. Hence, $g$ and $h$ are semantic UIPs with respect to $\Delta'$, and in fact, the same argument shows that $g$ is still an earliest semantic UIP with respect to $\Delta'$. This is nearly sufficient for us to use induction to conclude that $h\in \Conseq_{\le i-1}(g) \subseteq \Conseq_{\le i}(g)$, which would finish this case. It only remains to check that the projection $\Delta'$ is of the right form, meaning
        that $\Delta' = \mathcal{R}_{\le i-1} \cup \{C^*\}$ for some $C^*$ that is falsified by $U_{i-1}$ but not by $L_k$. 
        
        The case $\Delta' = \mathcal{R}_{\le i-1} \cup \{C\}$ is of the right form by assumption. The case $\Delta' = \mathcal{R}_{\le i-1}$ actually cannot happen, because the existence of semantic UIPs in $\linspan(U_i)$ with respect to $\Delta' = \mathcal{R}_{\le i-1}$, like $g$ or $h$, contradicts the fact that $U_i$ simultaneously satisfies such UIPs and $L_k \cup \mathcal{R}_{\le i-1}$. Finally, to show that the case $\Delta' = \mathcal{R}_{\le i-1} \cup \{D\}$ is of the right form, first note that $D$ is falsified by $U_{i-1}$ as a property of resolution. To show that $D$ is not falsified by $L_k$, consider the set $L_k \cup \mathcal{R}_{\le i} \cup \{C\}$. Because $C$ is not falsified by $L_k$, \cref{prop:sem-fp} implies that this set is satisfiable, and its projection onto $\linspan^\pm(U_{i-1})$ is $L_k \cup \mathcal{R}_{\le i-1} \cup \{D\}$. Because \cref{prop:proj-eqsat} guarantees equisatisfiability, $D$ is not falsified by $L_k$ as desired.
    \end{proof}

    \noindent\textsc{Case $\pos(h) = i$:}
        As in the previous case, we would like to apply induction, but it is no longer true that $g$ and $h$ remain semantic UIPs with respect to the projection of $\Delta = \mathcal{R}_{\le i} \cup \{C\}$ onto $\linspan^\pm(U_{i-1})$. Instead, write $h = f_i + h'$ for some $h' \in \linspan(U_{i-1})$. We make the following claim.
        
        \begin{claim} \label{clm:proj-prop}
            Either the claim $h \in \Conseq_{\le i}(g)$ is true without induction, or there exists a set $\Delta'$ of the form $\Delta' = \mathcal{R}_{\le i-1} \cup \{C^*\}$ for some clause $C^*$ such that the following properties also hold.
            \begin{enumerate}
                \item The equation $h'$ is a semantic UIP with respect to $\Delta'$.
                \item The clause $C^*$ is falsified by $U_{i-1}$ and not by $L_k$.
                \item The set $L_k \cup \Delta' \cup \{R_i'\}$ is unsatisfiable, where $R_i \equiv R_i' \lor f_i$.
                \item Write $C \equiv C' \lor (\lnot f_i + v)$ where $\linneg{C'} \subseteq \linspan(U_{i-1})$ and $v \in \linspan(U_{i-1})$. Then either $C^* = C$ or $C^* = C' \lor (\lnot h' + v)$.
            \end{enumerate}
        \end{claim}

        \begin{proof}[Proof of claim]%[Proof of \cref{clm:proj-prop}]
            We set up some notation. In the case that $C$ is not falsified by $U_{i-1}$, we write $C \equiv C' \lor (\lnot f_i + v)$ where $\linneg{C'} \subseteq \linspan(U_{i-1})$ and $v \in \linspan(U_{i-1})$. As usual, write $R_i \equiv R_i' \lor f_i$, and as a reminder, we have $h = f_i + h'$.
            
            Because $h$ is a semantic UIP with respect to $\Delta$, the set $L_k \cup \mathcal{R}_{\le i} \cup \{C\} \cup \{h\}$ is unsatisfiable. We project this onto $\linspan^\pm(U_{i-1})$. By \cref{prop:proj-eqsat}, the set $L_k \cup \mathcal{R}_{\le i-1} \cup \Gamma$ is unsatisfiable, where
            \begin{equation*}
                \Gamma = \begin{cases}
                    \{C\} & \text{if $C$ is falsified by $U_{i-1}$ and $f_i$ is a decision} \\
                    \{C, R_i' \lor h'\}& \text{if $C$ is falsified by $U_{i-1}$ and $f_i$ is propagated} \\
                    \{C'\lor (\lnot h' + v)\}& \text{if $C$ is not falsified by $U_{i-1}$ and $f_i$ is a decision} \\
                    \{C' \lor (\lnot h' + v), R_i' \lor h'\}& \text{if $C$ is not falsified by $U_{i-1}$ and $f_i$ is propagated}
                \end{cases}
            \end{equation*}
            as can be seen by the following:
            The computations for the first three cases are easy. In the last case, we have written the resolvent of $C$ and $h$, as well as the resolvent of $R_i$ and $h$, but not the resolvent of $C$ and $R_i$. This is because the resolvent of $C$ and $R_i$ is $C' \lor R_i' \lor \lnot v$, which is implied by the two other resolvents and can thus be safely omitted for the purpose of equisatisfiability. It remains to use $\Gamma$ to pick $\Delta'$ and show the four desired properties.

            The first case, where $\Gamma = \{C\}$, is actually not possible. This is because $L_k \cup \mathcal{R}_{\le i-1} \cup \{C\}$ being unsatisfiable, where $C$ is not falsified by $L_k$ by assumption, contradicts \cref{prop:sem-fp}.

            In the second case, where $\Gamma = \{C, R_i' \lor h'\}$, we take $\Delta' = \mathcal{R}_{\le i-1} \cup \{C\}$. We verify the four properties: For property (1), $h'$ is a semantic UIP with respect to $\Delta'$ because $L_k \cup \mathcal{R}_{\le i-1} \cup \Gamma$ being unsatisfiable implies that $L_k \cup \Delta' \cup \{h'\}$ is unsatisfiable, where $R_i' \lor h'$ has been made stronger by replacing it with $h'$. For property (2), $C$ is falsified by $U_{i-1}$ for this case, and is not falsified by $L_k$ by definition of $C$. For property (3), $L_k \cup \Delta' \cup \{R_i'\}$ is unsatisfiable because $L_k \cup \mathcal{R}_{\le i-1} \cup \Gamma$ is unsatisfiable and $R_i'$ is stronger than $R_i' \lor h'$. Property (4) is obvious, with $C^* = C$ in the statement of the property. 

            The third case, where $\Gamma = \{C' \lor (\lnot h' + v)\}$, is similar to the first case, being not possible because of \cref{prop:sem-fp} \emph{unless} $C' \lor (\lnot h' + v)$ is falsified by $L_k$. This means that $C'$ is falsified by $L_k$ and $h' + v \in \linspan(L_k)$. Then, we can show that $h \in \Conseq_{\le i}(g)$ without induction: Note that $C \equiv C' \lor (\lnot f_i + v)$ propagates $(\lnot f_i + v) + (h' + v) = \lnot h$ from $L_k$. Recall that because $g$ is a semantic UIP with respect to $\Delta$, we know that $L_k \cup \mathcal{R}_{\le i} \cup \{C\} \cup \{g\}$ is unsatisfiable. Because $C$ propagates $\lnot h$, we see that $L_k \cup \mathcal{R}_{\le i} \cup \{\lnot h\} \cup \{g\}$ is unsatisfiable, or, in other words, $h \in \Conseq_{\le i}(g)$. 

            In the fourth case, where $\Gamma = \{C' \lor (\lnot h' + v), R_i' \lor h'\}$, we take $\Delta' = \mathcal{R}_{\le i-1} \cup \{C' \lor (\lnot h' + v)\}$. We verify the four properties. For property (1), $h'$ is a semantic UIP with respect to $\Delta'$ for the same reason as in the second case. For property (2), $C' \lor (\lnot h' + v)$ is falsified by $U_{i-1}$ as the result of resolving $C$ with $h$ on $f_i$, and either it is not falsified by $L_k$ as desired for this property, or it could be falsified by $L_k$, but then we would conclude $h \in \Conseq_{\le i}(g)$ as in the third case. For property (3), $L_k \cup \Delta' \cup \{R_i'\}$ is unsatisfiable by the same argument as the second case. Property (4) is obvious, with $C^* = C' \lor (\lnot h' + v)$ in the statement of the property.
        \end{proof}

        We now prove the case $\pos(h)=i$ using this claim. Properties (1) and (2) in \cref{clm:proj-prop} allow us to take $g'$ to be the earliest semantic UIP with respect to $\Delta'$, then apply induction to derive UIPs as consequences of $g'$. As a reminder, the final goal is to prove that $h \in \Conseq_{\le i}(g)$. Under the assumptions $L_k \land \mathcal{R}_{\le i}$, we will show that $g \Rightarrow g' \Rightarrow h$, in other words, that $h \in \Conseq_{\le i}(g')$ and $g' \in \Conseq_{\le i}(g)$.

        To show that $h \in \Conseq_{\le i}(g')$, note that because $h = f_i + h'$ and $h' \in \Conseq_{\le i-1}(g')$ by induction, it suffices to show $f_i \in \Conseq_{\le i}(g')$. The idea is to get $f_i$ by propagating it from $R_i$. We notice that for any $r \in \linneg{R_i'}$, we have $(L_k \land \Delta') \Rightarrow r$ by property (3) of \cref{clm:proj-prop}. Because $h'$ is a semantic UIP with respect to $\Delta'$, we conclude that $(L_k \land \Delta') \Rightarrow \lnot h' + r$, so that $h' + r$ is also a semantic UIP with respect to $\Delta'$. (This is similar to \cref{prop:sem-uip-struc}. In the discussion preceding the proposition, $\lnot r$ would be a ``non-trail UIP,'' so a UIP plus $r$ is still a UIP.) Thus $h' + r \in \Conseq_{\le i}(g')$. Because $h' \in \Conseq_{\le i}(g')$ as well, we get $r \in \Conseq_{\le i}(g')$ for all $r \in \linneg{R_i'}$, and thus we also know $f_i$ via $R_i$ as desired.

        To show that $g' \in \Conseq_{\le i}(g)$, recall that $\pos(g') \le \pos(h') \le i-1$. This would actually allow us to apply \cref{clm:early-pos} to conclude that $g' \in \Conseq_{\le i}(g)$, as long as we can show the required hypothesis that $g'$ is a semantic UIP with respect to the original $\Delta$. Following property (4) of \cref{clm:proj-prop}, in the case that $C^* = C$, we have $\Delta' \subseteq \Delta$, so this is trivial. In the case that $C^* = C' \lor (\lnot h' + v)$, because $g'$ is a semantic UIP with respect to $\Delta'$, we know that $(L_k \land \mathcal{R}_{\le i-1} \land g') \Rightarrow (\linneg{C'} \land (h' + v))$, and it suffices to show that $(L_k \land \mathcal{R}_{\le i} \land g') \Rightarrow (\linneg{C'} \land (f_i + v))$.
        These are already nearly identical. In particular, because $h = h' + f_i$, it suffices to know that $h \in \Conseq_{\le i}(g')$. We do know this, as shown in the previous paragraph.
        This completes the proof of the $\pos(h) = i$ case.
        
        The proposition follows by induction. 
\end{proof}

\subsection{Simulation of the \texorpdfstring{$\Res(\oplus)$}{Res(⊕)} proof system} \label{ssec:sim-res}

While the affine resolution rule we introduced for $\CDCL(\oplus)$ is new, XNF reasoning has been studied under the $\Res(\oplus)$ proof system since 2014 \cite{DBLP:conf/mfcs/ItsyksonS14,DBLP:journals/apal/ItsyksonS20}, primarily focused on lower bounds \cite{DBLP:conf/stoc/EfremenkoGI24, DBLP:conf/coco/BhattacharyaCD24, DBLP:conf/stoc/AlekseevI25}. We show that our affine resolution rule is actually equivalent to $\Res(\oplus)$, and consequently, that the relationship between $\CDCL(\oplus)$ and $\Res(\oplus)$ mirrors the classical relationship between CDCL and Resolution.  
%In the definition below, what we call ``ordinary resolution'' was introduced as just ``resolution'' in \cite{DBLP:conf/mfcs/ItsyksonS14}, but we include the adjective to disambiguate it from affine resolution.

\begin{definition}[\cite{DBLP:conf/mfcs/ItsyksonS14}]
    The $\Res(\oplus)$ proof system has non-tautological linear clauses as lines, and the following two inference rules:
    \begin{itemize}
        \item (ordinary resolution\footnote{We include the ``ordinary'' adjective, which was not used in the original papers \cite{DBLP:conf/mfcs/ItsyksonS14, DBLP:journals/apal/ItsyksonS20}, to distinguish this rule from affine resolution.}) From $C \lor f$ and $D \lor \lnot f$, deduce $C \lor D$.
        \item (weakening) From $C$, deduce any $D$ such that $C \Rightarrow D$, equivalently $\linneg C \subseteq \linneg D$.
    \end{itemize}
\end{definition}

For example, one can weaken $\eq{x \oplus y = 1}$ to deduce $\eq{x = 1} \lor \eq{y = 1}$. While weakening is also a feature sometimes included in standard Resolution (from $C$, deriving $C \lor \ell$ for any literal $\ell$), a standard inductive argument shows that any proof involving weakening can be adjusted into a proof that weakens at the very last step, so that weakening does not add any deductive power. However, in $\Res(\oplus)$, weakening cannot be removed: on CNF input, $\Res(\oplus)$ without weakening keeps the entire proof in CNF and is thus subject to Resolution lower bounds, despite $\Res(\oplus)$ having polynomial-size proofs for problems like Tseitin formulas that are hard for Resolution \cite{DBLP:journals/jacm/Urquhart87}.

The necessity of weakening and the sheer number of possible deductions allowed by the rule make it initially difficult to imagine what a solver that simulates $\Res(\oplus)$ might look like. Although there are results that show how weakening can be limited~\cite{DBLP:conf/mfcs/ItsyksonS14, DBLP:conf/scsquare/HoracekK18}, a main result of ours is to show that general weakening is not necessary for $\Res(\oplus)$ at all. Instead, it suffices to have a ``change of basis'' rule that allows switching between equivalent clauses, along with upgrading ordinary resolution into affine resolution. All of the various proof rules that we consider are summarized in \cref{fig:res-summary}. We now define one more variant of resolution called ``addition,'' which was briefly mentioned in \cref{ssec:aff-res}, then give the main equivalence theorem.

\begin{figure}[t]
\centering
\small
\begin{tblr}{
  colspec = {Q[c,m,wd=2.4cm] Q[c,m] Q[c,m]},
  hline{2} = {1}{-}{}, hline{2} = {2}{-}{},
  vline{-} = {2}{abovepos = 1},
  hlines, vlines,
  rowsep  = 4pt,
  colsep  = 0pt,
  column{1} = {colsep=4pt},
  cell{4,6,9,11}{-} = {bg=black!5},
  hline{4,6,9,11}     = {2-3}{dashed},
  hline{4,6,9,11}     = {1}{wd=0pt},
  cell{4,6,9,11}{1}   = {halign=r},
  cell{2,3,5,7,8,10}{2-3} = {cmd=\mbox},
}
Rule & Lines are linear clauses & Lines are subspaces \\
Ordinary resolution
  &
  \AxiomC{$C' \lor f$}%
  \AxiomC{$D' \lor \lnot f$}%
  \BinaryInfC{$C' \lor D'$}%
  \DisplayProof
  &
  \AxiomC{$\linspan(V', \lnot f)$}%
  \AxiomC{$\linspan(W', f)$}%
  \BinaryInfC{$\linspan(V', W')$}%
  \DisplayProof
  \\
Affine resolution
  &
  \AxiomC{$C' \lor (g + u)$}%
  \AxiomC{$D' \lor (\lnot g + v)$}%
  \BinaryInfC{$C' \lor D' \lor \lnot(u + v)$}%
  \DisplayProof
  &
  \AxiomC{$\linspan(V', \lnot g + u)$}%
  \AxiomC{$\linspan(W', g + v)$}%
  \BinaryInfC{$\linspan(V', W', u + v)$}%
  \DisplayProof
  \\
\dots where
  &
  {$\linneg{C'}, \linneg{D'} \subseteq \linspan(U)$ \\ $u, v \in \linspan(U)$}
  &
  {$V', W' \subseteq \linspan(U)$ \\ $u, v \in \linspan(U)$}
  \\
Generalized affine resolution
  &
  \AxiomC{$C' \lor (g + u)$}%
  \AxiomC{$D' \lor (g + v)$}%
  \BinaryInfC{$C' \lor D' \lor (u + v)$}%
  \DisplayProof
  &
  \AxiomC{$\linspan(V', \lnot g + u)$}%
  \AxiomC{$\linspan(W', \lnot g + v)$}%
  \BinaryInfC{$\linspan(V', W', \lnot(u + v))$}%
  \DisplayProof
  \\
\dots where
  &
  {$\linneg{C'}, \linneg{D'} \subseteq \linspan^\pm(U)$ \\ $u, v \in \linspan^\pm(U)$}
  &
  {$V', W' \subseteq \linspan^\pm(U)$ \\ $u, v \in \linspan^\pm(U)$}
  \\
Addition
  &
  \AxiomC{$C' \lor f$}%
  \AxiomC{$D' \lor g$}%
  \BinaryInfC{$C' \lor D' \lor (f + g)$}%
  \DisplayProof
  &
  \AxiomC{$\linspan(V', \lnot f)$}%
  \AxiomC{$\linspan(W', \lnot g)$}%
  \BinaryInfC{$\linspan(V', W', \lnot(f + g))$}%
  \DisplayProof
  \\
Change of basis
  &
  \AxiomC{$C$}%
  \UnaryInfC{$D$}%
  \DisplayProof
  &
  N/A (identity)
  \\
\dots where
  &
  $C \equiv D$
  &
  ---
  \\
Weakening
  &
  \AxiomC{$C$}%
  \UnaryInfC{$D$}%
  \DisplayProof
  &
  \AxiomC{$V$}%
  \UnaryInfC{$W$}%
  \DisplayProof
  \\
\dots where
  &
  $C \Rightarrow D$
  &
  $V \subseteq W$
  \\
\end{tblr}

\caption{The table above summarizes various proof rules, including preconditions to apply each rule when applicable. Note that in (generalized) affine resolution, $U \cup \{g\}$ is any linearly independent and consistent set of equations. Also, although $U$ and $g$ are useful bookkeeping devices in $\CDCL(\oplus)$, a formal proof system does not need to track $U$ and $g$ because each proof step can be verified similarly to the addition rule, by checking every pair of disjuncts between the two premises.}
\label{fig:res-summary}
\end{figure}

\begin{definition}
    We define the \emph{addition} and \emph{change of basis} inference rules as follows:
    \begin{itemize}
        \item (addition) From $C \lor f$ and $D \lor g$, deduce $C \lor D \lor (f + g)$.
        \item (change of basis) From $C$, deduce any $D$ such that $C \equiv D$, equivalently $\linneg C = \linneg D$.
    \end{itemize}
\end{definition}

It is clear that (generalized) affine resolution is just a special case of addition where a particular equation is canceled out in the sum $f + g$. The addition rule is also an $\mathbb{F}_2$ version of what \cite{DBLP:journals/apal/RazT08} called simply ``Resolution'' for the R(lin) proof system, which operates on disjunctions of affine equations over $\mathbb{Z}$. We specifically note the importance of this $\mathbb{F}_2$ version in connection with $\Res(\oplus)$. 

\begin{theorem}[Equivalent definitions of $\Res(\oplus)$] \label{thm:res-parity}
    Where proof lines are non-tautological linear clauses, the following inference rules produce polynomially equivalent proof systems (for the purpose of refutation). The simulations preserve graph structure, and thus tree-like and other fragments are also equivalent. 
    \begin{enumerate}
        \item Affine resolution and change of basis
        \item Generalized affine resolution and change of basis
        \item Addition and change of basis
        \item Ordinary resolution and weakening
    \end{enumerate}
\end{theorem}
\begin{proof}
    It is obvious that (3) simulates (2), which simulates (1). One can also quickly check that (4) simulates (3) because weakening is strictly more powerful than change of basis, and the addition rule can be simulated as follows. (We tag each deduction with an abbreviated name of the rule used.)
    \begin{center}
        \AxiomC{$C' \lor f$}%
        \AxiomC{$D' \lor g$}%
        \reason{add}
        \BinaryInfC{$C' \lor D' \lor (f + g)$}%
        \DisplayProof
        \quad $\Rightarrow$ \quad 
        \AxiomC{$C' \lor f$}%
        \reason{weak}
        \UnaryInfC{$C' \lor (f + g) \lor \lnot g$}%
        \AxiomC{$D' \lor g$}%
        \reason{ord-res}
        \BinaryInfC{$C' \lor D' \lor (f + g)$}%
        \DisplayProof
    \end{center}
    If the weakening of $C' \lor f$ above is not possible because $C' \lor (f + g) \lor \lnot g$ is tautological, this means that $g \Rightarrow C' \lor (f + g)$. Hence, $D' \lor g$ weakens to $D' \lor (C' \lor (f + g))$ directly.

    It remains to prove that (1) simulates (4). We will inductively construct such a proof, where each new line (using affine resolution and change of basis) will be at least as strong as the corresponding line (using ordinary resolution and weakening). The base cases are the starting clauses, which remain the same. In the inductive step, let $A \lor B$ be the first not-yet-converted clause that was derived via ordinary resolution from $A \lor f$ and $B \lor \lnot f$. Visually, we have the following.
    \begin{prooftree}
        \AxiomC{$C$}
        \reason{weak}
        \UnaryInfC{$A \lor f$}
        \AxiomC{$D$}
        \reason{weak}
        \UnaryInfC{$B \lor \lnot f$}
        \reason{ord-res}
        \BinaryInfC{$A \lor B$} 
    \end{prooftree}
    In particular, we note that a chain of 0 or more weakenings can be collapsed into a single weakening (possibly the identity weakening). Then because all clauses derived by resolution before $A \lor B$ have been converted by induction, and because our inductive assumption is that each old clause is a weakening of the converted new clause, we can assume that $C$ and $D$ above represent clauses that have already been converted by induction.

    Let $U$ be a basis of $\linneg{A \lor B}$. If $C$ is falsified by $U$ or $D$ is falsified by $U$, then we can simply assign $C$ or $D$ to be the new proof line associated with $A \lor B$. Otherwise, when neither is falsified by $U$, one can check that $f \not \in \linspan^\pm(U)$, and that starting from $U$ the clause $C$ unit propagates $f$ and the clause $D$ unit propagates $\lnot f$  (e.g., by the semantic characterization of unit propagation in \cref{prop:semantic}). Then by \cref{prop:up-decomp}, we can write $C \equiv C' \lor (f + u)$ and $D \equiv D' \lor (\lnot f + v)$ where $C'$ and $D'$ are falsified by $U$ and $u, v \in \linspan(U)$. The clause $C' \lor D' \lor \lnot(u + v)$ is falsified by $U$, and is thus the strengthening of $A \lor B$ that we wanted. To summarize, we have the following.
    \begin{prooftree}
        \AxiomC{$C$}
        \reason{chg-basis}
        \UnaryInfC{$C' \lor (f + u)$}
        \AxiomC{$D$}
        \reason{chg-basis}
        \UnaryInfC{$D' \lor (\lnot f + v)$}
        \reason{aff-res}
        \BinaryInfC{$C' \lor D' \lor \lnot(u + v)$} 
    \end{prooftree}
    This is exactly a proof using affine resolution and change of basis.
\end{proof}

It remains to describe the connection between $\CDCL(\oplus)$ and $\Res(\oplus)$. To recap, there are two important ways that classical CDCL relates to Resolution. 
First, CDCL with any asserting clause learning procedure produces Resolution proofs of size at most the number of unit propagations~\cite{DBLP:journals/jair/BeameKS04}. 
More generally, this can be viewed as an exact equivalence between unit propagation and a fragment of Resolution known as input Resolution (or sometimes trivial Resolution)~\cite{DBLP:journals/jacm/Chang70, DBLP:journals/jair/BeameKS04}. 
Second is the converse direction, that any problem with a Resolution proof of unsatisfiability of length $M$ also admits a sequence of decisions and restarts for CDCL that returns unsatisfiable after learning at most polynomially many clauses: originally $O(n^3 M)$ clauses in~\cite{DBLP:journals/ai/PipatsrisawatD11}, improved to $O(n^2M)$ in~\cite{DBLP:journals/lmcs/BeyersdorffB23}. 
We generalize both of these to $\CDCL(\oplus)$ and $\Res(\oplus)$.  

For simplicity, for the rest of this section, we consider $\Res(\oplus)$ to have subspaces of equations as proof lines, as illustrated in the right column of \cref{fig:res-summary}. This means that instead of a clause $C$, we write the subspace $\linneg{C}$, and we can assume that addition is the only proof rule because change of basis is just the identity on subspaces. We assume that subspaces are always consistent, analogous to clauses being non-tautological.

\begin{definition} \label{def:input-res}
An \emph{input $\Res(\oplus)$ proof} of a clause $C$ from a set of clauses $\Delta$ is a valid $\Res(\oplus)$ proof of the following form, where each $C_i \in \Delta$ and each $V_i$ is derived via addition from $V_{i-1}$ and $\linneg{C_{i-1}}$. 
\begin{prooftree}
    \AxiomC{$V_1 = \linneg{C_0}$}
    \AxiomC{$\linneg{C_1}$}
    \reason{add}
    \BinaryInfC{$V_2$}
    \AxiomC{$\linneg{C_2}$}
    \reason{add}
    \BinaryInfC{$\vdots$}
    \AxiomC{$\linneg{C_{k}}$}
    \reason{add}
    \BinaryInfC{$V_{k+1} = \linneg{C}$}
\end{prooftree}
The length of the proof is $k$.
\end{definition}

Input $\Res(\oplus)$ is exactly the shape of the resolution process in the unit propagation and 1-UIP clause learning algorithms. While it is obvious that 1-UIP clause learning produces clauses derivable by input $\Res(\oplus)$ proofs, this is also true for any algorithm that learns conflict clauses. Note that this is essentially only because we \emph{defined} conflict clauses by unit propagation (as opposed to semantics, see \cref{ex:sem-confl}). We state the proposition slightly more generally, so that we do not need to discuss the trail. This mimics a classical result from \cite{DBLP:journals/jair/BeameKS04}, while sidestepping their use of implication graphs.

\begin{definition}
    A subspace of equations $V$ is \emph{conflict-like with respect to a set of clauses $\Delta$} if unit propagation starting with units $V$ and clauses $\Delta$ produces a contradiction.
\end{definition}

\begin{proposition}[Input $\Res(\oplus)$ and unit propagation] \label{prop:up-rp}
    If a clause $C$ has an input $\Res(\oplus)$ proof from $\Delta$, then $\linneg C$ is conflict-like with respect to $\Delta$. Conversely, if $\linneg C$ is conflict-like, there is an input $\Res(\oplus)$ proof of $C'$ such that $\linneg{C'} \subseteq \linneg{C}$. The length of the proof is the number of unit propagations made.
\end{proposition}
\begin{proof}
    Suppose that $C$ has an input $\Res(\oplus)$ proof from $\Delta$, and let subspaces $V_1, \dots, V_{k+1}$ and clauses $C_0, \dots, C_k$ be as in \cref{def:input-res}. We claim that $C_i$ propagates $V_i$ from $V_{i+1}$, or more formally, that $C_i$ propagates an equation $g$ such that $V_i \subseteq \linspan(V_{i+1}, g)$. (Then by induction, $C_1$ propagates $\linneg{C_0}$, so that $C_0$ can propagate $\bot$ as desired, and we conclude by confluence.) Denote this $\Res(\oplus)$ step as follows.
    \begin{prooftree}
        \AxiomC{$V_i = \linspan(V_i', \lnot f)$}
        \AxiomC{$\linneg{C_i} = \linspan(W', \lnot g)$}
        \reason{add}
        \BinaryInfC{$V_{i+1} = \linspan(V_i', W', \lnot(f+g))$}
    \end{prooftree}
    It is clear that $C_i$ propagates $g$ (or $\bot$, in which case the induction ends early) from $V_{i+1} \supseteq W'$. Then, $V_i \subseteq \linspan(V_{i+1}, g)$, as desired. 

    Conversely, suppose that $\linneg{C}$ is conflict-like, so that unit propagation produces $\bot$. To construct the input $\Res(\oplus)$ proof, simply take this trail and let $C'$ be the output of \cref{alg:assert} run until reaching a decision (formally, with stopping condition ``$f_i$ is a decision''). We clearly have an input $\Res(\oplus)$ proof of $C'$, and the algorithm guarantees that $U_i$ falsifies $C'$ when $U_i = \linneg{C}$ is the set of decisions, so $\linneg{C'} \subseteq \linneg{C}$ as desired.
\end{proof}

\begin{corollary}[$\Res(\oplus)$ simulates $\CDCL(\oplus)$]\label{thm:cdclbound}
    Let $\Delta$ be an unsatisfiable set of clauses and consider $\CDCL(\oplus)$ with any asserting clause learning scheme. The execution of the algorithm on $\Delta$ may be converted into a $\Res(\oplus)$ proof of unsatisfiability with at most as many additions as the number of unit propagations made.
\end{corollary}

For the converse relationship, the original simulation of Resolution by CDCL in~\cite{DBLP:journals/ai/PipatsrisawatD11} used a concept called \emph{absorption}~\cite{DBLP:conf/aaai/PipatsrisawatD08a}. A Boolean clause $C$ is \emph{absorbed} by a set of clauses $\Delta$ if for all sets of starting units $U$, unit propagation on the set $\Delta \cup \{C\}$ deduces the same units as $\Delta$ alone. (In some treatments, a clause that is not absorbed is called ``1-empowered,'' though we won't use this phrase.) This definition translates easily to $\CDCL(\oplus)$ and its more powerful unit propagation algorithm, but it turns out to be the wrong definition to use, as can be seen from the following example.

\begin{example}
    Suppose that $\Delta = \{\eq{x = 1} \lor \eq{z = 0}, \eq{y = 1} \lor \eq{z = 1}\}$ and we are considering $C = \eq{x = 1} \lor \eq{y = 1}$. This same example is a valid input for both classical unit propagation and the unit propagation in $\CDCL(\oplus)$.  Classically, $C$ would be absorbed by $\Delta$ because the only ways to use $C$ are when $\eq{x = 0}$ or $\eq{y = 0}$ is part of the starting units $U$, and in both of these cases, the deduction using $C$ can also be deduced using the two clauses in $\Delta$. 
    
    However, in our version, suppose that $U = \{\eq{x\oplus y = 0}\}$. We see that $C$ is actually not absorbed, because $C$ can propagate, for example, $\eq{y = 1}$, whereas nothing can be propagated from $\Delta$. In other words, it is much harder to absorb clauses in our version, because unit propagation is more powerful. This prevents us from tracing the argument in~\cite{DBLP:journals/ai/PipatsrisawatD11} directly.
\end{example}

Instead of absorbing clauses, we absorb specific \emph{decompositions} of clauses, which identify a specific disjunct, just as when we write $R_i \equiv R_i' \lor f_i$. Such a decomposition is absorbed if the inference $\lnot R_i' \Rightarrow f_i$ can be performed through unit propagation using $\Delta$ instead. Absorbing a clause is equivalent to absorbing every possible decomposition. While there are at most $\width(C)$ possible classical decompositions for a Boolean clause $C$, there are exponentially many possible decompositions for linear clauses.   

\begin{definition}
    For a subspace $V$, a \emph{decomposition} is a pair $(V', \lnot f)$ such that $V = \linspan(V', \lnot f)$. The decomposition is \emph{absorbed} by a set of clauses $\Delta$ if unit propagation with units $V'$ and clauses $\Delta$ deduces either $\bot$ or $f$. 
\end{definition}

The addition rule of $\Res(\oplus)$ naturally specifies decompositions of its antecedents; the decompositions used in a given $\Res(\oplus)$ proof are exactly what we absorb. Then, at a high level, the rest of the construction of decisions and restarts to absorb these decompositions closely follows the original argument in~\cite{DBLP:journals/ai/PipatsrisawatD11}, and in fact improves it by a multiplicative factor equal to the number of variables, $n$. 

To summarize the improvement, given a clause $C$ in a Resolution proof of unsatisfiability from $\Delta$, the argument in~\cite{DBLP:journals/ai/PipatsrisawatD11} forces the solver to absorb an entire conflict-like clause $C$, equivalent to absorbing all $n$ decompositions in the Boolean case. Instead, we only absorb the decompositions that are used in applications of the addition rule, which on average is fewer than 2 decompositions per clause, a factor $n$ improvement. To keep this article self-contained, we give the full argument below, following the same definition and lemma structure of \cite{DBLP:journals/ai/PipatsrisawatD11}, which starts by defining which decompositions we would like to absorb.

\begin{definition}
   A decomposition $(V', \lnot f)$ of a conflict-like subspace $V$ is \emph{asserting-like} for a set of clauses $\Delta$ if $(V', \lnot f)$ is not absorbed by $\Delta$.
\end{definition}

\begin{lemma}[Existence of asserting-like decompositions] \label{lem:findassertlike}
    Suppose that $\Delta$ is unsatisfiable and unit propagation cannot derive $\bot$. Then every $\Res(\oplus)$ refutation of $\Delta$ contains a deduction 
    \begin{prooftree}
        \AxiomC{$\linspan(V', \lnot f)$}
        \AxiomC{$\linspan(W', \lnot g)$}
        \reason{add}
        \BinaryInfC{$\linspan(V', W', \lnot(f + g))$} 
    \end{prooftree}
    such that at least one of $(V', \lnot f)$ and $(W', \lnot g)$ is asserting-like for $\Delta$. 
\end{lemma}
\begin{proof}
    The starting subspaces $\linneg{C}$ for $C \in \Delta$ are trivially conflict-like, and the final subspace $\{\top\}$ is not conflict-like by the assumption that unit propagation fails to derive $\bot$. Consider the deduction of the first subspace in the proof that is not conflict-like, with notation as in the lemma statement. Then, $(V', \lnot f)$ and $(W', \lnot g)$ are conflict-like, and we claim that at least one is not absorbed by $\Delta$. Suppose for contradiction that both are absorbed, so that propagation from $V'$ and $W'$ either deduces $\bot$, or deduces $f$ and $g$, respectively. All cases contradict $\linspan(V', W', \lnot(f + g))$ not being conflict-like. 
\end{proof}

Reading \cref{lem:findassertlike} in contrapositive form, we see that, once every asserting-like decomposition in a proof is absorbed by $\Delta$, unit propagation will deduce a contradiction from $\Delta$. Thus, the main idea is to construct a run of $\CDCL(\oplus)$ that expands $\Delta$ to absorb every asserting-like decomposition in a proof. Let $n$ be the number of variables.

For the rest of this subsection, assume that unit propagation in $\CDCL(\oplus)$ reaches a fixed point after every decision, and consider any asserting clause learning scheme, not necessarily 1-UIP clause learning. Assume that once clauses are learned, they are never deleted.

\begin{lemma}[Forced absorption lemma] \label{lem:absorb}
    Let $(V', \lnot f)$ be an asserting-like decomposition of a subspace $V$ under a set of clauses $\Delta$. There exists a decision sequence that results in the clause database of $\CDCL(\oplus)$ absorbing $(V', \lnot f)$ while learning at most $n^2$ clauses.
\end{lemma}

\begin{proof}
    Until $(V', \lnot f)$ is absorbed, we repeatedly perform the following loop as our decision sequence for $\CDCL(\oplus)$: restart, then decide any basis of $V'$, then $\lnot f$. (If some decision is already in the trail $\linspan(U)$, we skip it. It is not possible for any decision to be in $\linspan^\pm(U) \setminus \linspan(U)$ because $(V', \lnot f)$ is not absorbed yet.) We claim that we learn at most  $n$ clauses per loop iteration, and the loop terminates within $n$ iterations.
    
    For the number of clauses per loop iteration, first note that we never deduce a contradiction and perform conflict analysis before deciding $\lnot f$ because $(V', \lnot f)$ is not absorbed, and we always deduce a contradiction after deciding $\lnot f$ because $V$ is conflict-like. After reaching a contradiction, $\CDCL(\oplus)$ can reach a contradiction and learn clauses many more times before needing the next decision. Each contradiction backjumps by at least one level, so it takes at most $n$ learned clauses to reach the next decision. 

    For the number of loop iterations, consider the \emph{first} asserting clause derived in each iteration. By the fixed-point assumption, each allows a new equation to be deduced from decisions $V'$, and these equations must be linearly independent to be new, and hence this can occur at most $n$ times.
\end{proof}

\begin{theorem}[$\CDCL(\oplus)$ simulates $\Res(\oplus)$] \label{thm:full-sim}
    Given a $\Res(\oplus)$ proof of unsatisfiability of a set of clauses $\Delta$ using $M$ additions, there is a sequence of decisions and restarts such that $\CDCL(\oplus)$ learns at most $2n^2 M$ clauses before concluding unsatisfiability. 
\end{theorem}

\begin{proof}
    The decision sequence is constructed adaptively while running $\CDCL(\oplus)$. Until unit propagation can immediately derive contradiction from $\Delta$, we repeatedly perform the following loop: use \cref{lem:findassertlike} to find an asserting-like decomposition $(V', \lnot f)$, then use \cref{lem:absorb} to absorb $(V', \lnot f)$ by learning at most $n^2$ new clauses.
    
    Note that $\Delta$ is growing throughout this process, but \cref{lem:findassertlike} may view the original proof of unsatisfiability as always starting with the original $\Delta$ and ignoring learned clauses. Since each absorbed decomposition must stay absorbed as $\Delta$ grows, and each addition uses 2 decompositions, the procedure must terminate within $2M$ iterations. 
\end{proof}

Since each clause has an input $\Res(\oplus)$ proof of length at most $n$ by \cref{prop:up-rp}  and the linear independence of $\CDCL(\oplus)$ trails, the translation of a $\Res(\oplus)$ proof to a $\CDCL(\oplus)$ execution and back to $\Res(\oplus)$ incurs a multiplicative $O(n^3)$ overhead. When this argument is translated line by line back to the Resolution and classical CDCL case, this $O(n^3)$ factor matches what is proven in~\cite{DBLP:journals/lmcs/BeyersdorffB23}.

\begin{corollary} \label{cor:full-sim-resolution}
    \cref{thm:full-sim} is also true for Resolution and classical CDCL. 
\end{corollary}

\section{Empirical evaluation of \texorpdfstring{\solverName{}}{Xorcle}} \label{sec:empr}

We implement $\CDCL(\oplus)$ in our solver \solverName{}. The current version of \solverName{} is v0.2, which includes many features that were absent from v0.1, which was released with the conference version of this article \cite{DBLP:conf/sat/BeameS26}. We begin with an overview of important features in \solverName{} v0.2 beyond the core $\CDCL(\oplus)$ algorithm, then compare how \solverName{} performs against other solvers on a variety of benchmarks.

\subsection{List of features implemented in \texorpdfstring{\solverName{}}{Xorcle}} \label{ssec:heur}

In this subsection, we discuss all major features included in \solverName{} v0.2, including features first introduced in \solverName{} v0.1. A summary of how various heuristics contribute to the final performance of \solverName{} can be found in \cref{sec:comp-heur}. While the feature set has greatly expanded since v0.1, \solverName{} still does not yet translate some of the more recent features in CDCL SAT solving, such as chronological backtracking \cite{DBLP:conf/sat/NadelR18, DBLP:conf/sat/MohleB19}, bounded variable addition \cite{DBLP:conf/hvc/MantheyHB12, DBLP:conf/sat/HaberlandtGH23}, techniques to extract specific structures \cite{DBLP:conf/sat/SoosNC09, DBLP:conf/aaai/SoosM19, DBLP:conf/sat/BiereBLM14}, symmetry breaking \cite{DBLP:conf/kr/CrawfordGLR96, DBLP:conf/sat/Devriendt0BD16, DBLP:journals/jair/AndersBR26}, integration with local search \cite{DBLP:conf/sat/CaiZ21, DBLP:journals/jair/CaiZFB22}, etc. Translations of these remaining features are not available in v0.2 primarily due to limitations on scope and engineering capacity, rather than feasibility. 

Before discussing all the heuristics, a common idea we will need is to associate a clause $C$ falsified by the trail $U$ with a set of exactly $\width(C)$ trail equations. In classical CDCL, there is nothing to do, because the clause already consists of negated trail literals. In $\CDCL(\oplus)$, it is slightly more complicated. Following trail notation from \cref{conv:trail},  consider that a clause $\lnot(f_1 + f_2) \lor \lnot f_3$ has width 2, yet it involves $f_1$, $f_2$, and $f_3$. However, the equation $f_1 + f_2$ is equivalent to $f_2$ modulo $U_1 = \{f_1\}$, so both are in \emph{position} 2 of the trail. A natural solution is to associate $C$ with $\width(C)$ trail \emph{positions} rather than equations. 

Specifically, associate an equation $g \in \linspan^\pm(U_m)$ with position $\pos(g) = t$ if $t$ is the smallest index such that $g \in \linspan^\pm(U_t)$. Associate a conflict clause $C$ with the positions $\pos(C) = \{\pos(g) \mid g \in \linneg{C}\} \setminus \{0\}$. We note that $\lvert\pos(C)\rvert = \width(C)$ because any collection of equations from distinct non-zero positions is linearly independent. The positions of a clause may be calculated efficiently by writing the clause in the trail basis in row-echelon form, similar to how clauses are written in \cref{alg:trail-basis}. For example, if $C = \lnot(f_1 + f_2) \lor \lnot f_3$, then we say that $\pos(C) = \{2, 3\}$.  

\subparagraph*{CSIDS decision procedure} 

In order to unlock proof-complexity-theoretic improvements on CNF instances over classical CDCL, $\CDCL(\oplus)$ should branch on equations, not just variables.  Classical SAT solvers commonly choose variables based on information tracked for each one, using techniques like VSIDS~\cite{DBLP:conf/dac/MoskewiczMZZM01, DBLP:conf/sat/EenS03} and CHB~\cite{DBLP:conf/aaai/LiangGPC16}. However, since we have exponentially many equations, these are difficult to adapt. Instead, we choose an equation using a two-step process similar to some older approaches in classical SAT solving such as BerkMin~\cite{DBLP:conf/date/GoldbergN02} and Clause-Move-To-Front (CMTF)~\cite{DBLP:conf/hvc/GershmanS05}: first choose a clause $C$, then choose an equation associated with $C$.

To choose the clause $C$, we introduce a new heuristic that we call \textbf{CSIDS} (Clause State Independent Decaying Sum), which is inspired by standard VSIDS and interpolates between the intuition behind clause choices in BerkMin and CMTF, generalizing both. CSIDS maintains activity scores for all clauses and selects clauses in order of activity. Starting clauses have base activity 0, new asserting clauses have base activity $A_0 \ge 0$, each clause used during conflict analysis has activity bumped by $\delta \ge 0$, and all activities decay by a multiplicative factor $0 \le \beta \le 1$ at every conflict. Our default parameter choices are $A_0 = 1$, $\delta = 1$, and $\beta = 0.98$. 

We note that choosing $A_0 = 1$, $\delta = 0$, and any $0 < \beta < 1$ results in a BerkMin-like order, always producing clauses in the reverse order that they were learned. (Unlike the original BerkMin heuristic, this falls back to starting clauses if all conflict clauses are satisfied.) We also note that choosing $A_0 = 1$, $\delta = 1$, and any $0 < \beta \le 0.5$ results in a CMTF-like order, guaranteeing that all clauses used in conflict analysis, including the new asserting clause, are moved to the front. (Unlike the original CMTF heuristic, this also moves starting clauses and does not bound the number of clauses moved.)

We choose the clause $C$ with highest activity score that is not already satisfied. From $C$, we should \emph{not} just choose one of the equations syntactically represented in $C$, since that would still only branch on Boolean literals when given CNF input. Let $U$ be the current trail. Our main approach in \solverName{} v0.1 essentially sampled a uniformly random $f \in \linneg{C} \setminus \linspan(U)$. Then, deciding $f$ (``falsify'') always makes progress towards unit propagation with $C$, while deciding $\lnot f$ (``satisfy'') immediately satisfies $C$. This is partially analogous to phase heuristics in classical solvers, but the choice is not arbitrary: $f$ keeps $C$ as a reason in the trail, while $\lnot f$ moves the search to a different region. Both options for this falsify/satisfy setting are available, and ``falsify'' is the default choice.

In \solverName{} v0.2, we improve the method for selecting equations by prioritizing equations that affect multiple clauses. Specifically, let $C_1, C_2, C_3, \dots$ be the not-satisfied clauses sorted by activity score, with highest score first. We initialize $V = \linspan(\linneg{C_1}, U)$ and then, for increasing $i$, we update $V \gets V \cap \linspan(\linneg{C_i}, U)$ whenever this intersection is nontrivial modulo $\linspan(U)$, stopping when $V$ is one-dimensional modulo $\linspan(U)$ or at a fixed depth for $i$, by default 50. After that, we decide a random $f \in V \setminus \linspan(U)$. As desired, such a decision is likely to affect multiple top-activity clauses at once.
We include an option, which is off by default, to use $\linspan^\pm$ instead of $\linspan$ in the above, which is analogous to prioritizing shared variables instead of shared literals, so the final choice $f$ simultaneously satisfies some clauses while moving others towards unit propagation. In that case, the falsify/satisfy setting refers to falsifying or satisfying the equation in the top-activity clause.   

\subparagraph*{Phase saving}
Classical phase saving chooses the polarity of decisions to match previous long trails, recently propagated choices, or other criteria \cite{DBLP:conf/sat/PipatsrisawatD07, DBLP:journals/jair/CaiZFB22}. In \solverName{}, we can treat a previous long trail as a source for phases, and use a recent trail as an approximation for the saved polarity of each variable. Specifically, we save the longest trail seen, with a reset every 1000 conflicts (a ``target trail,'' excluding the decision level of the conflict), and we save the trail immediately before a conflict (a ``recent trail,'' excluding $\bot$ itself). These trails may be used with options ``both,'' ``target-only,'' ``recent-only,'' ``none,'' and ``rotate,'' which rotates between the other four options on a schedule, as a simple alternative to rephasing \cite{DBLP:journals/jair/CaiZFB22}. Our default choice is to rotate.

Note that unlike the classical situation, when the current trail and the saved trail contradict each other, the saved trail generally yields no preference for the phase. For example, suppose that the saved trail is $S = \{\eq{x = 0}, \eq{y = 0}\}$ and the current trail has $U = \{\eq{x = 1}\}$. We remove phase information from the final intersection space $V$ by taking $\linspan^\pm$, in order to allow $S$ to determine the phase. If the final intersection space is $\linspan^\pm(V) = \linspan^\pm(\eq{x = 1}, \eq{y = 0})$, then the intersection with $\linspan(S)$ contains both $\eq{y = 0}$ and $\eq{x \oplus y = 0}$. But $\eq{x \oplus y = 0}$ is equivalent to $\eq{y = 1}$ modulo $U$, so the saved trail yields no preference. To fix this in many cases, we intersect $V$ with  $\linspan^\pm(\linneg{C_i})$ for clauses $C_i$ involved in the intersection computation of CSIDS, including the top-activity clause. Continuing the example above, if such a clause $C_i$ does not involve $x$ at all (meaning that $x$ entered $V$ only through the trail) and $C_i$ does contain either $\eq{y = 0}$ or $\eq{y = 1}$, then consulting the trail would result in deciding $\eq{y = 0}$. 

Formally, let $V$ be the final space from CSIDS, let $S$ be a saved trail (either target or recent), and let $W \gets \linspan^\pm(V)$. Using $S$ to decide a phase means updating 
\begin{equation*}
    W \gets W \cap \linspan^\pm(\linneg{C_i}) \cap \linspan(S)
\end{equation*}
for the highest-activity clause $C_i$ involved in the intersection computation such that this space is not contained in $\linspan^\pm(U)$ and is consistent with $U$. When no such clause exists, we can still attempt to use $S$ by setting $W \gets W \cap \linspan(S)$ if the intersection is not contained in $\linspan^\pm(U)$, which typically does not decide a phase if the current trail contradicts $S$, but can make $W$ smaller, which keeps focus on an active region of the search space. We first try to decide a phase using any saved trail we have (with target trail first if we have both), and if that fails, we try to make $W$ smaller using $S$. Finally, we sample a uniformly random $f \in W \setminus \linspan^\pm(U)$. If $S$ decided the phase, we decide $f$, and if not, we decide $f$ or $\lnot f$ according to the falsify/satisfy setting described in the CSIDS section, where we falsify or satisfy the equation in the top-activity clause. 

\subparagraph*{Watched equations and clause selection} 
Watched literals~\cite{DBLP:conf/dac/MoskewiczMZZM01} optimize classical
unit propagation in two ways:
(1) if neither of the two ``watched'' literals in a clause is falsified then it doesn't propagate, and
(2) the ``watchlists'' of clauses that watch a particular literal allow the solver to quickly find candidate clauses for
propagation. $\CDCL(\oplus)$ can easily watch two equations $f, g \in \linneg C$, with a note that we must also check whether $f + g \in \linspan(U)$ to determine if a watch must be updated. However, watchlists are more difficult to adapt. A new unit alone does not determine which clauses require action---it also depends on the previously known units, since linear combinations are allowed. Thus, clause selection becomes a heuristic choice. \solverName{} v0.2 includes two options: a \textbf{cyclic} selector, with optional \textbf{incomplete propagation}, and an \textbf{exact watchlist} selector. 

The cyclic selector scans clauses in database order, always picking up where it last left off (not restarting from the beginning), except jumping to new asserting clauses and jumping to the top-activity clause after a CSIDS decision. This performed best in v0.1, but results in a large part of \solverName{}'s running time being spent on inspecting clauses that do not contribute to progress. In v0.2, the cyclic selector can save time using incomplete propagation, which detects when unit propagation is unlikely to propagate further and stops there to make a decision, \emph{without} being confident in having reached fixed point. There are no soundness concerns with this change (as long as we remember to propagate to fixed point before declaring SAT), but suboptimal learned clauses are possible. Note that for a reason clause $R_i \equiv R_i' \lor f_i$ that propagated $f_i$ because $R_i'$ was falsified, we can still compute the true decision level of $f_i$ by taking the maximum over the true levels of $\pos(R_i')$. 

The exact watchlist selector guarantees that selected clauses will be useful to look at (i.e., will propagate an equation or at least update a watched equation), similar to classical watchlists, at the cost of significant overhead. In particular, it maintains a copy of every watched equation, \emph{partially reduced} by the current trail and sorted by leading term into buckets, with the second equation also partially reduced against the first to have distinct leading terms. Partial reduction means that we stop eliminating coordinates as soon as we reach one that is not a leading term of a trail equation, in contrast to full reduction as defined in \cref{ssec:up}, which eliminates \emph{every} coordinate that is a leading term of a trail equation. Partial reduction is enough for detecting when an equation reduces to $\top$ or $\bot$, and is enough to only reclassify the bucket corresponding to the leading term of each new trail equation, instead of needing a full loop through all clauses. However, this structure requires frequent upkeep to preserve the semantics described above, such as undoing reductions during a backjump. The exact watchlist selector did not perform well in v0.1 due to poor engineering, but it has been more carefully optimized in v0.2 and is now faster than the cyclic selector on most instances. We enable it by default.

\subparagraph*{LBD-based clause deletion}

In classical SAT solvers, literal block distance (LBD) refers to the number of distinct decision levels that involve the literals of a conflict clause~\cite{DBLP:conf/ijcai/AudemardS09}. In \solverName{}, we can compute LBD for a conflict clause $C$ by counting the number of levels represented in $\pos(C)$, excluding level 0. We also recalculate LBD for a reason clause $R_i$ whenever it is used as a reason during conflict analysis. Instead of storing the LBD for only a single point in time, we found that heuristics generally perform better when we remember an average LBD over recent uses of the clause, calculated with an exponential moving average with small effective window size, by default $1/\alpha = 2$, where $\alpha$ is the smoothing factor.

Clause deletion in \solverName{} is inspired by the three-tier design \cite{DBLP:conf/sat/Oh15} used in modern solvers such as CaDiCaL and Kissat, where clauses are deleted on a schedule of passes. However, rather than three discrete tiers, we prefer a simplified approach, where a clause with recent average LBD $k$  is given $\mathsf{round}(2^{6-k})$ 
deletion passes where it is allowed to be unused. The counter for the number of deletion passes is recalculated and reset every time a clause is used in conflict analysis. The number of conflicts between passes grows proportionally to the square root of the number of conflicts so far, as suggested by \cite{DBLP:conf/ijcai/AudemardS09}.
During a deletion pass, excluding current reasons on the trail, half of the learned clauses that have run out of free passes are evicted from the clause database, chosen by either CSIDS activity (default) or LBD. We observed that this model performed slightly better than the three-tier model.

\subparagraph*{Local clause minimization}

Local self-subsumption minimization, as introduced by MiniSat \cite{DBLP:conf/sat/EenS03, DBLP:conf/sat/SorenssonB09}, is implemented as a continuation of the 1-UIP clause learning algorithm, extended beyond the first-detected asserting clause. Specifically, after the first-detected asserting clause, we continue affine resolution in \cref{alg:assert} going all the way back to the beginning of the trail, but we accept only resolutions that strictly shrink the width of the learned clause. Note that the same resolution-based procedure in classical CDCL exactly matches the more common description of local minimization as a graph-based procedure.

An important detail is that once we begin to skip some resolutions, there is no canonical way to perform further resolution. For example, suppose that resolution on $f_3$ was skipped, and now we have to resolve $C = \lnot f_1 \lor \lnot f_2 \lor \lnot f_3$ and $R_2 = \lnot f_1 \lor f_2$. If we use $U = \{f_1, f_3\}$, then the resolvent is clearly $\lnot f_1 \lor \lnot f_3$. But we could just as well use $U = \{f_1, f_2 + f_3\}$, so that we rewrite $C \equiv \lnot f_1 \lor \lnot f_2 \lor \lnot (f_2 + f_3)$ and get the resolvent $\lnot f_1 \lor \lnot (f_2 + f_3)$. Because we have no reason to prefer one over the other, we choose between them uniformly at random. Specifically, at position $i$ in the trail, we resolve with respect to a uniformly random codimension-1 subspace of $\linspan(U_{i-1}, \linneg{C})$ that contains $\linspan(U_{i-1})$ and does not contain $f_i$. 

We remark that the equivalent of recursive minimization \cite{DBLP:conf/sat/SorenssonB09} appears to be NP-hard for $\CDCL(\oplus)$, as argued in the AI-generated note \cite{rec-min-hard}. One can check that for classical CDCL, recursive minimization finds the unique conflict clause $C^*$ of minimum width that subsumes the output $C$ of the 1-UIP clause learning algorithm. However, at a high level, from an arbitrary CNF instance, \cite{rec-min-hard} constructs an XNF instance such that, when run on $\CDCL(\oplus)$ with an appropriate sequence of decisions, the satisfiability of the original CNF is encoded in the width of the minimum-width conflict clause subsuming the output of the 1-UIP clause learning algorithm. This would imply that recursive minimization is NP-hard, so we make no attempt to do more minimization than what is described above.

\subparagraph*{On-the-fly subsumption}

During conflict analysis, on-the-fly clause subsumption \cite{DBLP:conf/sat/HanS09} additionally checks if a derived conflict clause subsumes the reason clause(s) used immediately to derive it (both reasons for the first conflict clause, and the one reason for the remaining ones). This relies on the simple observation that such a reason $R$ will be subsumed by the conflict clause $C$ if and only if $\width(C) = \width(R) - 1$ (cf.~Lemma 1 of \cite{DBLP:conf/sat/HanS09}). On-the-fly subsumption is checked for all conflict clauses derived during clause learning, including intermediate conflict clauses before reaching the first-detected asserting clause, clauses that were deduced for attempted minimization but ultimately rejected, etc. If an intermediate conflict clause subsumes an earlier clause, we learn \emph{both} the intermediate clause (replacing the subsumed reason) and the final asserting clause, which is unchanged.

\subparagraph*{Forced restarts with replay for inprocessing}
It turns out that restarts generally hurt \solverName{}'s performance. (See the paragraph about restarts under ``Included heuristics that are off by default'' below.) Thus, when performing various inprocessing heuristics that may be triggered by schedules or other inputs, we not only want to avoid waiting for a level 0 backjump, but we also want to avoid the trail changes caused by a forced restart. Therefore, we save the current trail, restart, perform the inprocessing, and replay the decisions and propagations from the saved trail afterwards. 

Specifically, assume that inprocessing guarantees that every old clause is semantically implied by a single new clause and/or units at the root. Associate each old clause with its corresponding new clause, or with the no clause if it is implied solely by units at the root. We can perform replay by making old decisions and attempting unit propagation using the new clauses corresponding to old reason clauses in order, skipping anything that is now given by a root unit. (For engineering simplicity, we also let unit propagation proceed as normal at the end of each level, before each decision.) Then the new trail span after attempting each clause must be a superset of the old trail span at that point, except after reaching a conflict or reaching a decision that would immediately be a conflict. One can prove this by a simple induction argument. 

\subparagraph*{Vivification}
Similar to subsumption or minimization, vivification makes clauses narrower, but by using unit propagation with the entire clause database and checking if the negation of a subclause is enough to unit propagate contradiction \cite{DBLP:conf/ecai/PietteHS08, DBLP:conf/ijcai/LuoLXML17}. In \solverName{}, we have many more choices of subclauses than before, so we just choose a random basis and add basis vectors one at a time. We may repeat this a few more times to see if other random bases are successful. In our testing, we found that vivifying wider clauses has little effect relative to effort, so we attempt to vivify a $\min(1, 2^{3-w})$ fraction of the width-$w$ clauses in every vivification pass, up to some work budget, skipping width-2 clauses and skipping clauses that were already attempted in previous passes. Vivification passes occur at a fixed frequency in number of conflicts via a restart with replay.

\subparagraph*{Implication graph reasoning}

We process clauses of width $\le 2$ using an implication graph (of the 2-SAT style where edges are sufficient conditions). This follows ideas from the 2-XNF solver 2-Xornado \cite{DBLP:journals/mics/AndraschkoDK24}, and supersedes failed literal probing and equivalent literal substitution in traditional SAT solvers \cite{DBLP:conf/ictai/LynceS03, DBLP:conf/sat/HeuleJB11}. Each width-2 clause $f \lor g$ contributes 6 vertices and 6 edges to the graph: vertices $f$, $g$, $f + g$, and their negations, and edges $\lnot f \to g$, $f + g \to f$, $f + g \to g$, and their contrapositives. Let $L_0$ denote level 0 units. For each vertex $v$, we can consider the space 
\begin{equation*}
    \textrm{ImpliedBy}(v) = \linspan(L_0 \cup \{v\} \cup \{\text{descendants of $v$}\}).
\end{equation*}
Equations in the intersection $\textrm{ImpliedBy}(v) \cap \textrm{ImpliedBy}(\lnot v)$ must be true and can thus be learned by the solver. In particular, if $\lnot v \in \textrm{ImpliedBy}(v)$ and $v \in \textrm{ImpliedBy}(\lnot v)$, then $\bot$ belongs to the intersection and the formula is unsatisfiable. 

{
\allowdisplaybreaks Implication graphs concisely encode how unit propagation would unfold if a vertex were taken as a decision. For example, consider the following clauses.
\begin{align*} 
    C_1 & = \eq{x = 1} \lor \eq{y = 1} \\
    C_2 & = \eq{x = 1} \lor \eq{z = 1} \\
    C_3 & = \eq{x = 0} \lor \eq{y \oplus w = 0} \\
    C_4 & = \eq{x = 0} \lor \eq{z \oplus w = 0}
\end{align*}
After building the graph, one can compute that 
\begin{align*}
    \textrm{ImpliedBy}(\eq{x = 0}) & = \linspan(\eq{x=0}, \eq{y = 1}, \eq{z=1}) \\
    \textrm{ImpliedBy}(\eq{x = 1}) & = \linspan(\eq{x=1}, \eq{y\oplus w=0}, \eq{z \oplus w = 0})
\end{align*}
We conclude that $\eq{y \oplus z = 0}$ holds, and add it as a unit to the solver. Note that $\eq{y \oplus z = 0}$ was not a vertex in the implication graph, but we can derive it nonetheless. 
}

\solverName{} runs implication graph reasoning as preprocessing and at a fixed frequency in number of conflicts via a restart with replay, skipped if fewer than 50 new width-2 clauses have been learned from any source since the previous run. To keep runs cheap, each run searches a bounded number of pairs $(v, \lnot v)$ and only explores the graph via BFS to a bounded number of descendants. 

\subparagraph*{Bounded projection (BVE)}

Classical solvers can attempt bounded variable elimination (BVE) on all variables \cite{DBLP:conf/sat/EenB05}. Although we have faithfully translated variable elimination as projection in \cref{ssec:elim}, we need to heuristically choose a few of the exponentially many possible subspaces to project onto. Once a candidate subspace is chosen, projection onto the subspace is accepted if the total width of clauses after projection is at most the total width before projection. We call the procedure \textbf{bounded projection} to align with our previous vocabulary, but retain the abbreviation BVE. 

To find a good subspace to project onto, we observe that good subspaces are those that contain almost all of the clauses. The main intuition is that we can find such subspaces by starting with a random permutation of the clauses, then greedily taking the $\linspan^\pm$ of the (linear negations of the) clauses until the subspace formed is a hyperplane (a subspace of dimension $n$ in an $(n+1)$-dimensional space), which will contain almost all clauses with high probability. However, this naive approach is a bit slow, requiring a pass through every single clause to find one hyperplane and the list of clauses not contained in that hyperplane. Hyperplanes where elimination was rejected for growing the formula may also accidentally reappear, decreasing efficiency.

Instead, we can use the information from a single pass to find many potentially good hyperplanes. There are two levels to this. The simple thing to do is to stop at a subspace of dimension $n-k$, from which we get up to $k+1$ decent projections, at the cost of sorting through $2^{k+1}$ potential hyperplanes. To find even more hyperplanes, it is possible to have a structure where, up to a change of basis, most clauses involve only variables $x_1, \dots, x_a$, but some clauses involve variables $y_1, \dots, y_b$, and yet other clauses involve variables $z_1, \dots, z_c$, for some $b, c \le k+1$. We can detect this by tracking which equations are used to reduce which other equations during the initial pass. Then, we can declare just the $x_i$'s to be the core subspace, leaving $m$ groups of extra clauses (here $m=2$, for the $y_i$'s and $z_i$'s) using at most $k+1$ extra dimensions each, which requires sorting through just $m \cdot 2^{k+1}$ hyperplanes to get up to $m(k+1)$ projections.  

Within each group of clauses, we completely enumerate the hyperplanes (up to $2^{k+1}$ of them) and compute how many clauses are affected by each one. We attempt projection starting with the hyperplane affecting the fewest number of clauses. If no projections are accepted after a few attempts, we give up. Note that after a projection is accepted, the remaining hyperplanes in the set can be similarly projected, so that the process can repeat up to $k+1$ times. We repeat the entire process (starting with a fresh permutation of the clauses) until a pass is relatively unsuccessful, removing fewer than 1\% of the clauses that it started with. Currently, we only run BVE as preprocessing and not as inprocessing, because BVE is not compatible with a forced restart with replay, and experiments showed that for our set of instances, BVE on learned clauses would rarely have successfully eliminated anything. 

To reduce the total width after projection, we note that semantic gate detection \cite{DBLP:journals/fmsd/FleuryB22} translates easily. In particular, for a hyperplane $U$ and set of clauses $\Gamma = \{C_1, C_2, \dots\}$ affected by projection, we form the clauses $C_i'$ where $\linneg{C_i'} = \linneg{C_i} \cap \linspan^\pm(U)$. If a subset of these $C_i'$ is unsatisfiable, we call the corresponding $C_i$ a ``gate.'' Then, we only need to consider resolvents for pairs $(C_i, C_j)$ where at least one belongs to the gate, by essentially the same argument as in \cite{DBLP:conf/sat/EenB05}. We include a separate lightweight solver for finding gates, similar to Kitten \cite{DBLP:journals/fmsd/FleuryB22}.

\subparagraph*{Bit-packing, substitution, and compaction} 
We store equations with $n$ variables as $\lceil(n+1)/64 \rceil$-length row vectors of bit-packed unsigned integers. This use of bit-packing for XOR constraints is similar to \cite{DBLP:conf/cav/SoosGM20}. One downside of this representation is that memory usage is automatically quadratic in the input size, with even very sparse equations taking up a full $\lceil(n+1)/64 \rceil$-length vector. However, equations involving many variables are generally necessary in XNF reasoning, so quadratic memory usage is inevitable for most instances. For comparison, \solverName{} v0.2 also includes a sparse representation option that stores each equation as a list of the variables used in it, but this typically performs worse. 

To mitigate this slightly, we actively apply \textbf{substitution} on unit clauses learned from any source, replacing the leading variable in the equation with the rest of the equation across all clauses. Substitution can also be viewed as a form of BVE that always passes the boundedness requirement. Unlike general BVE, note that substitution is compatible with forced restarts with replay because the invariant allows new root units (but it is unnecessary when root units are learned through a level 0 backjump anyway).

The key idea is that substitution creates a column of zeros across all clauses. When 5\% of columns (comprising at least 64) are all-zeros, we rewrite every clause to remove these columns, renumbering variables and reducing the representation length of every equation by roughly 5\% or more. We call this \textbf{compaction}. While variable renumbering is an existing minor feature in solvers like CaDiCaL, it has a particularly outsized impact in \solverName{} due to our dense representation of equations.

\subparagraph*{Included heuristics that are off by default}
\solverName{} v0.2 includes translations of some heuristics that are generally included in classical CDCL solvers, but hurt \solverName{}'s performance in our testing. These heuristics are off by default for all evaluations in \cref{ssec:eval}. We provide a brief description of each.

\begin{description}
    \item[Restarts] We implement both Luby restarts \cite{DBLP:conf/istcs/LubySZ93, DBLP:conf/ijcai/Huang07} and LBD-based restarts \cite{DBLP:conf/cp/AudemardS12}, the latter of which restarts whenever the recent LBD of learned clauses becomes higher than average, with adjustable thresholds and parameters. We include options to prevent restarts when the trail is longer than average \cite{DBLP:conf/cp/AudemardS12}, as well as to prevent restarts before a certain number of conflicts have passed since the previous restart. We do not implement trail reuse \cite{DBLP:journals/jsat/TakRH11}, because the random sampling in CSIDS violates the desired property for the reused prefix to agree with the real trail after a restart. We observe that all configurations of restarts lead to decreased performance on our tested instances, and it remains open to explain this effect theoretically. Note that of course, this observation does not contradict \cref{thm:full-sim}, which says that some choice of decisions and restarts gives short proofs, not that restarts help \solverName{}'s current decision heuristics.
    \item [Stable/focused modes] \solverName{}'s default configuration is similar to the stable mode of classical CDCL solvers, involving CSIDS decisions, target phases, and no restarts \cite{DBLP:conf/sat/Oh15}. (Other solvers restart only rarely in stable mode.) Focused mode typically involves Variable Move-To-Front (VMTF) decisions, plain phase saving, and frequent restarts. Consistent with our observation that restarts hurt \solverName{}'s performance, a similar focused mode also hurts performance, but we include a basic implementation for reference. Note that as previously remarked, we can replace VMTF with a simulated CMTF using CSIDS decay $\beta \le 0.5$, and we can simulate plain phase saving by using only the recent trail.  
    \item [Inprocessing subsumption/strengthening] This heuristic extends subsumption detection to an inprocessing phase that uses the entire clause database, not just the trail as in on-the-fly subsumption, and also detects strengthening, which is when resolution with some clause results in subsumption \cite{DBLP:conf/sat/EenB05}. Classically, one can keep occurrence lists keyed by variables to find appropriate pairs of clauses quickly. We can keep similar lists based on the leading term of equations in the clause, so that \solverName{} is similarly effective compared to classical solvers at finding clauses to strengthen or be subsumed. However, we find that the resulting clauses perform badly overall, and the reason remains unexplained. Like other inprocessing heuristics, this is initiated using a restart with replay on a fixed schedule.
    \item[Activity using trail positions of learned clauses] Some implementations of VSIDS bump the activity of literals in the final learned clause $C$ \cite{DBLP:conf/dac/MoskewiczMZZM01}. We can partially mimic this idea by bumping the activity of reason clauses at positions $\pos(C)$, as well as the clause from which a decision was sampled for any position corresponding to a decision. We found that this generally hurt \solverName{}'s performance, even after re-tuning the CSIDS decay rate. 
\end{description}

\subparagraph*{\texorpdfstring{$\LRUP(\oplus)$}{LRUP(⊕)} proof logging}
Many SAT solvers output machine-checkable proofs for unsatisfiable instances in DRAT format~\cite{DBLP:conf/sat/WetzlerHH14}, but DRAT does not natively support parity reasoning. CryptoMiniSat outputs proofs in the XLRUP format~\cite{DBLP:conf/cav/TanYSMM24}. This supports parity reasoning but not extension variables, which are needed to compile our reasoning about linear clauses down to CNF-XOR. Another approach uses pseudo-Boolean solving and VeriPB~\cite{DBLP:conf/aaai/GochtN21}, but is more technically involved.

\solverName{} outputs a variation of the LRUP proof format~\cite{DBLP:conf/tacas/Cruz-FilipeMS17} that we call $\LRUP(\oplus)$, which is identical to LRUP except with parity expressions written in place of literals, with similar syntax to the XNF extension to the DIMACS input format in~\cite{DBLP:journals/mics/AndraschkoDK24}. In particular, an $\LRUP(\oplus)$ proof line consists of a conflict clause plus a list of the clauses used to derive it during conflict analysis in \cref{alg:assert}. By definition, a conflict clause $C$ can be checked to be a sound deduction by checking that unit propagation starting from units $\linneg{C}$ and the recorded reason clauses deduces contradiction. We provide a basic independently written $\LRUP(\oplus)$ proof checker that does exactly this, which has been updated alongside the release of \solverName{} v0.2 to be rewritten in Rust and be significantly faster.

All of \solverName{}'s heuristics are compatible with proof logging. In particular, vivification's unit propagation is trivially compatible. For resolution-based heuristics including substitution, BVE, and strengthening, one can check that if $C \equiv C' \lor (g + u)$ and $D \equiv D' \lor (g + v)$ resolve to give $C' \lor D' \lor (u + v)$, then units $\linneg{C' \lor D' \lor (u + v)}$ with clauses $\{C, D\}$ propagate a contradiction. For implication graph reasoning, suppose that considering $v$ and $\lnot v$ allows deduction of $f$. The set of descendants of $v$ can be viewed as unit propagation, so we can record $\lnot v \lor f$ as a temporary learned clause with the clauses in this unit propagation chain as reasons, and similarly do this for $v \lor f$. Then the unit clause $f$ follows by resolution. We note that although compaction renumbers variables within the solver, emitted proofs retain their original variable numbering. Lastly, deletion lines in $\LRUP(\oplus)$ are emitted not only by clause deletion, but also by BVE, subsumption, and many other heuristics. 

\subparagraph*{Removed options from \texorpdfstring{\solverName{}}{Xorcle} v0.1}

Two options from \solverName{} v0.1 have been removed for poor performance. Namely, we removed a VSIDS-like decision procedure that tracks activity for variables and makes only single-variable decisions, as well as a clause selector that kept watchlists by leading term but did not keep them reduced, which we had called a leading-term watchlist (in contrast to an exact watchlist). 

\subsection{Comparison with other solvers} \label{ssec:eval}

We evaluate \solverName{} v0.2's performance against leading classical CDCL solvers Kissat 4.0.4 and CaDiCaL 3.0.1, the leading CNF-XOR solver CryptoMiniSat 5.17.0 (which allows CNF and parity equations separately), and the competing XNF solver Xorricane 0.4.2.11. To enable comparison with the conference version of this article, we also include \solverName{} v0.1, CryptoMiniSat 5.13.0, and the 2-XNF solver 2-Xornado 1.0.2 in the primary cryptographic evaluation. We note that CryptoMiniSat received updates in version 5.15.0 to improve its performance significantly on XNF-derived instances. 

We used a machine with two 24-core Intel Xeon Platinum 8160M chips at 2.1~GHz and 754~GB RAM, running 32 processes at a time, with instances and solvers in randomized order to ensure approximately equal load across evaluations. Instances had a time limit of 300 seconds and a memory limit of 2~GB. No solvers reached the memory limit except for 2-Xornado, which reached the memory limit on many cryptographic instances, and \solverName{}, which only reached the memory limit on exceptionally large Tseitin formulas. For \solverName{}'s outputs, we checked that every SAT model was correct, and we used our $\LRUP(\oplus)$ proof checker to verify every UNSAT proof.

Note that the comparisons in this section generally expand on the comparisons in the conference version of this article \cite{DBLP:conf/sat/BeameS26}, except for two removed families: pebbling formulas lifted by $k$-XORs and random $k$-XNF formulas restricted by a system of linear equations. In \cite{DBLP:conf/sat/BeameS26}, it was shown that \solverName{} v0.1 performed well on the families, and that these families did not significantly distinguish among parity-aware solvers. For \solverName{} v0.2's evaluation, we consider these families to be rather unrealistic and arbitrary, so they have been dropped. 

\subparagraph*{Cryptography-related instances}

\begin{figure}
    \centering
    \includegraphics[width=\linewidth]{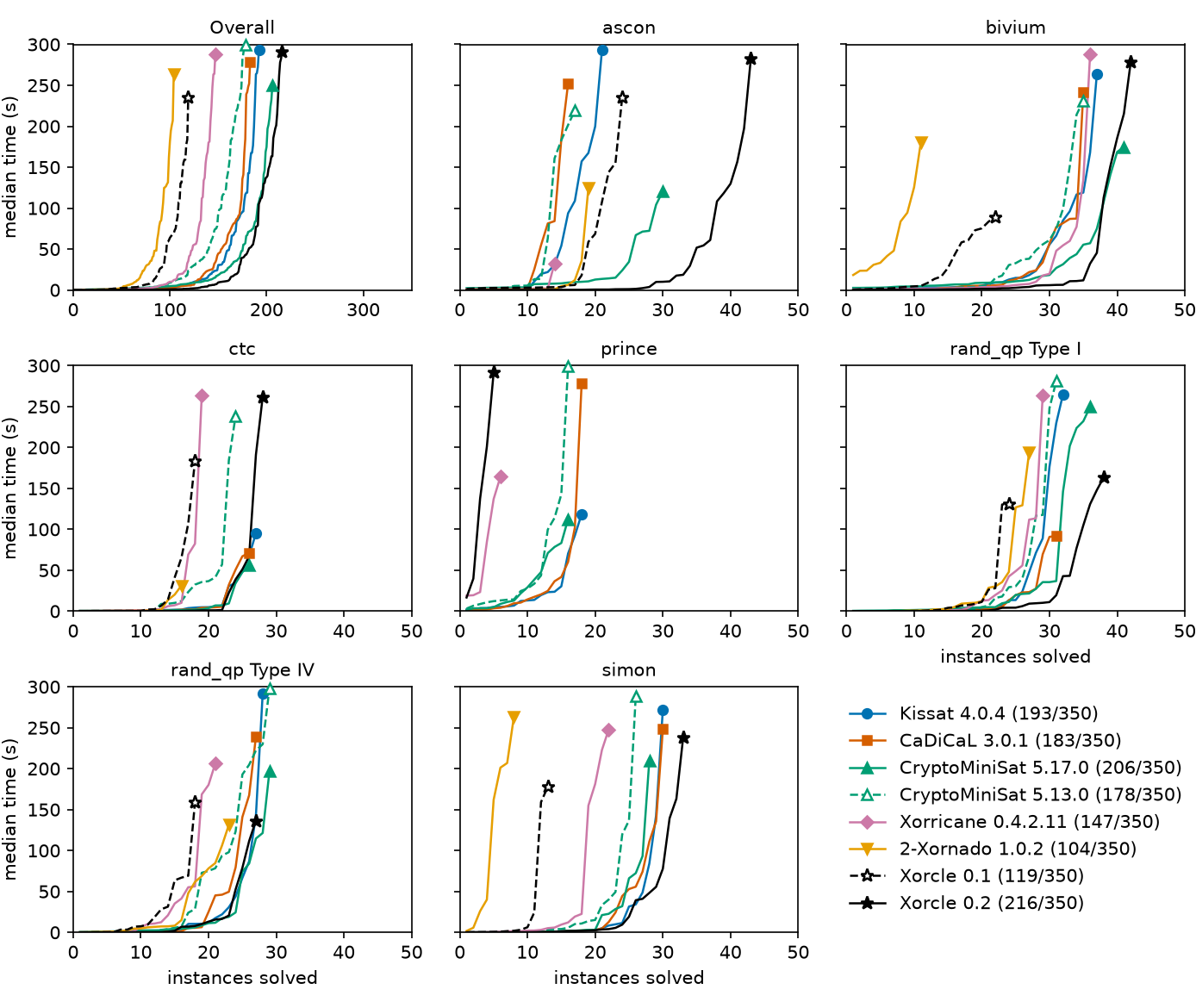}
    \caption{Evaluation on 350 cryptography-related instances released alongside Xorricane \cite{DBLP:phd/dnb/Danner25, DBLP:journals/jair/DannerK26}. \solverName{} is leading or competitive on most families, and leading overall.}
    \label{fig:main}
\end{figure}

Our primary evaluation set consists of 350 satisfiable instances released alongside Xorricane \cite{DBLP:phd/dnb/Danner25, DBLP:journals/jair/DannerK26}, based on five cryptographic ciphers and two types of random systems of quadratic polynomials used as cryptographic primitives. These instances were released in CNF, XNF, and 2-XNF, and we give each solver its native format. Following the evaluation of CryptoMiniSat on these families in \cite{manual:Soos26}, we give CryptoMiniSat these instances in CNF form, which appeared to improve its performance over a translated CNF-XOR form. We report the median time from 3 runs with different random seeds, for every solver except Xorricane and 2-Xornado, which do not support random seeds.

The results of the evaluation are shown in \cref{fig:main}. The plots clearly show that \solverName{} is competitive with other solvers on these families, with particularly strong performance on the \texttt{ascon} family and weaker performance on the \texttt{prince} family. Overall, \solverName{} v0.2 leads, solving 216 instances out of 350, compared to the runner-up solver CryptoMiniSat 5.17.0, solving 206 instances out of 350. 

\solverName{} v0.2's performance is particularly noteworthy in comparison to \solverName{} v0.1, which only solved 119 out of 350 instances. This can be attributed to the various improved heuristics in \cref{ssec:heur}, as well as overall improved engineering. We note that many of the CDCL heuristic improvements from the past decade are still missing from \solverName{}, which suggests that further improvement is possible with continued development.

An important note is that unlike industrial applications where one would like to blindly use a SAT solver to solve a given instance, a cryptographic attacker should be expected to optimize the cipher's encoding for whichever solver they choose to use, as well as to optimize each solver's flags to the encoding. Preliminary investigations show that encoding does make a significant difference in performance. For example, the encoding of the \texttt{prince} family currently allows unit propagation in only one direction through its S-boxes (which are invertible non-linear functions). Although  \solverName{} performs poorly on the \texttt{prince} family in the above comparison, simply adding clauses to allow propagation in the inverse direction seems to significantly improve \solverName{}'s performance, while actually making Kissat and CaDiCaL slower. Research into cryptographic encodings is out of scope for this article, so we benchmarked solvers against Xorricane's instances as released. Thus, we caution against reading the above plots as a conclusive comparison of solvers, but rather as an indication of which approaches are all generally competitive.

\subparagraph*{Random instances}

\begin{figure}
    \centering
    \includegraphics[width=\linewidth]{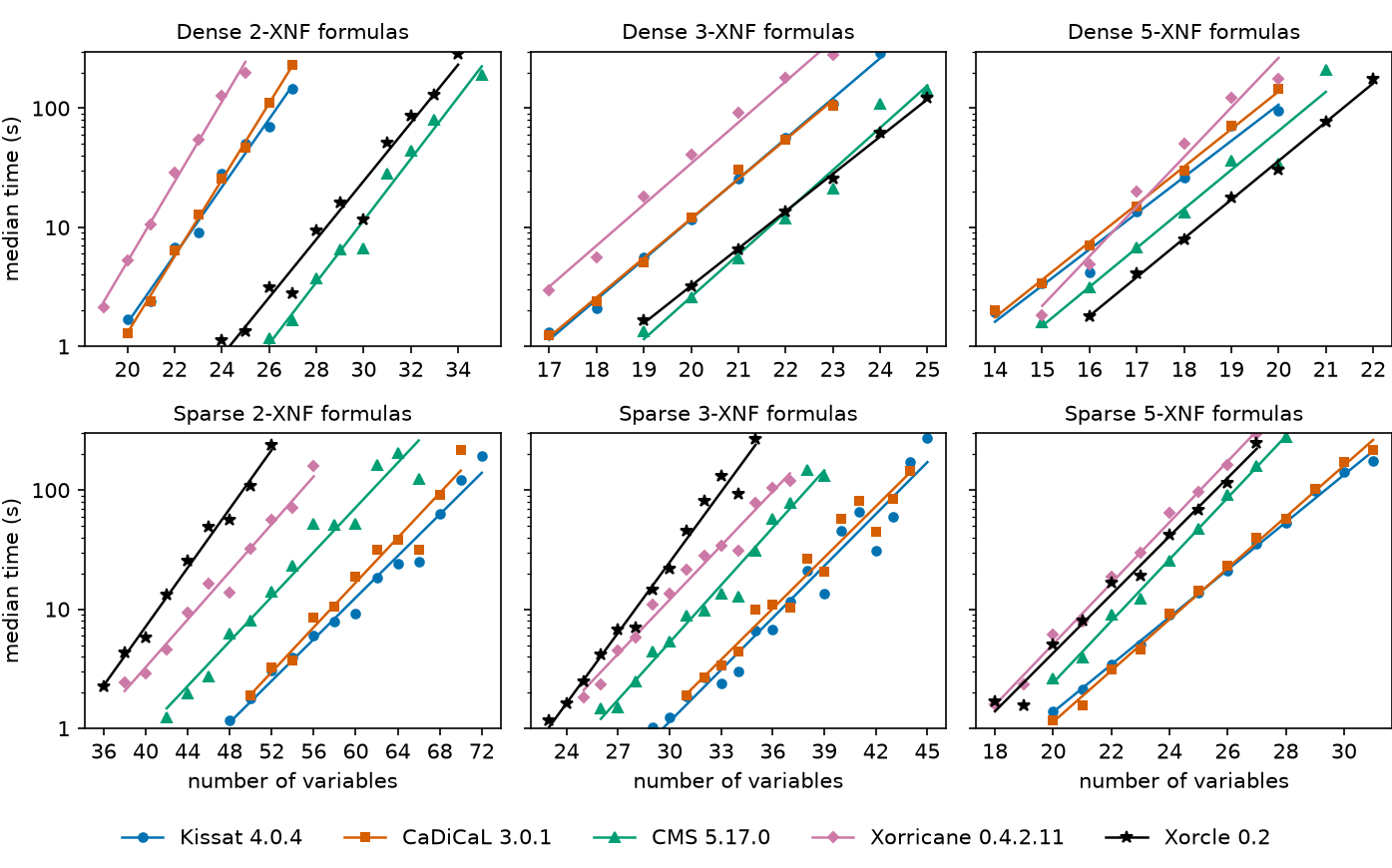}
    \caption{Evaluation on random $k$-XNF formulas, where $k$ is the width of clauses. For ``dense,'' equations in each clause are sampled from $\mathbb{F}_2^{n+1}$ uniformly at random, and for ``sparse,'' equations involve exactly 3 variables each. Each point is a median solve time over 10 random instances at the given size, about half of them satisfiable. \solverName{}'s relative performance improves on denser and wider random formulas.}
    \label{fig:rand}
\end{figure}

Our next set of benchmarks concerns random instances. With XNF, a mathematically natural way to get random $k$-XNF instances is to take a disjunction of $k$ equations sampled from $\mathbb{F}_2^{n+1}$ uniformly at random. This produces \emph{dense} equations that each involve $n/2$ variables, on average. One may also consider random $k$-XNF instances with \emph{sparse} equations with fewer variables per equation. Our sparse instances have 3 variables in each equation.

For each random family, we choose the number of clauses $m$ as a function of the number of variables $n$ so that about half of the instances are satisfiable. We determined the appropriate numbers experimentally by solving a few thousand small formulas and fitting linear functions $m = an + b$ for each combination of sparse and dense with $k \in \{2, 3, 5\}$. Formulas at such ratios are often considered to be the most difficult to solve, because there is little redundant structure or opportunity for lucky guesses.

We give the native XNF instances to \solverName{} and Xorricane. For CryptoMiniSat, we create an extension variable for each equation to convert instances into CNF-XOR, as recommended by CryptoMiniSat's prior evaluation on random XNF instances \cite{manual:Soos26}. For Kissat and CaDiCaL, we further split equations into at most 5 variables each through new extension variables, then write each $\ell$-variable equation as $2^{\ell - 1}$ clauses in CNF. This is the same method that \cite{DBLP:journals/jair/DannerK26} used to generate CNF versions of the cryptographic instances discussed earlier. For all solvers, we run 10 randomly generated instances for each variable count, increasing in variable count and stopping when at least 8 of the 10 instances time out. Figures report the median time per variable count. 

The results of the evaluation are shown in \cref{fig:rand}. \solverName{}'s performance varies between families, and is generally among the best for dense formulas and among the worst for sparse formulas, with relative performance increasing with width as well. In particular, compared with Kissat, CaDiCaL, CryptoMiniSat, and Xorricane, \solverName{} is in last place for sparse 2-XNF and 3-XNF formulas, and is in first place for dense 5-XNF formulas.

\subparagraph*{Proof-complexity-inspired instances}

\begin{table}[t]
    \centering
    \caption{Maximum value of $n$ solved within 300 seconds and 2~GB memory for each proof-complexity-inspired family. Tseitin formulas and perfect matching formulas are on $n$-vertex $k$-regular graphs, and bijective pigeonhole principle instances have $n+k$ pigeons and $n$ holes. For Tseitin and perfect matching, tested values for $n$ were $n = 1, 1.5, 2, 3, 5, 7 \times 10^N$ for various $N$, with a $+1$ adjustment if necessary for parity. For bijective PHP, we tested all integer values. Shaded cells indicate roughly polynomially scaling performance. Xorricane solved no Tseitin formulas at $k=7$.}
    \label{tab:pc}
    
    \small
    \begin{tblr}{
        hlines, vlines,
        hline{3} = {1}{-}{}, hline{3} = {2}{-}{},
        vline{-} = {3}{abovepos = 1},
        rowsep = 4pt, colsep = 4pt,
        colspec = {l *{9}{Q[c, wd=2.5em]}},
        cell{5}{2-4} = {bg=black!5},
        cell{7-8}{2-8} = {bg=black!5},
      }
        & \SetCell[c=3]{c} Tseitin$(n, k)$ & & & \SetCell[c=3]{c} PerfectMatching$(n, k)$ & & & \SetCell[c=3]{c} Bijective $\mathrm{PHP}^{n+k}_{n}$ & & \\
        $k$ & 3 & 5 & 7 & 4 & 6 & 8 & 1 & 3 & 5 \\
        Kissat 4.0.4 & 70 & 20 & 10 & 51 & 31 & 21 & 10 & 11 & 11 \\
        CaDiCaL 3.0.1 & 70 & 20 & 10 & 51 & 31 & 21 & 10 & 10 & 10 \\
        CryptoMiniSat 5.17.0 & 100K & 30K & 20K & 51 & 31 & 21 & 10 & 9 & 9 \\
        Xorricane 0.4.2.11 & 70 & 10 & -- & 51 & 21 & 21 & 8 & 8 & 8 \\
        \solverName{} 0.2 & 10K & 3,000 & 700 & 2,001 & 1,001 & 501 & 27 & 11 & 11 \\
        \solverName{} 0.2 \texttt{{-}{-}plain} & 15K & 1,500 & 500 & 5,001 & 2,001 & 1,001 & 33 & 10 & 11 \\
    \end{tblr}
    
\end{table}

\begin{figure}
    \centering
    \includegraphics[width=\linewidth]{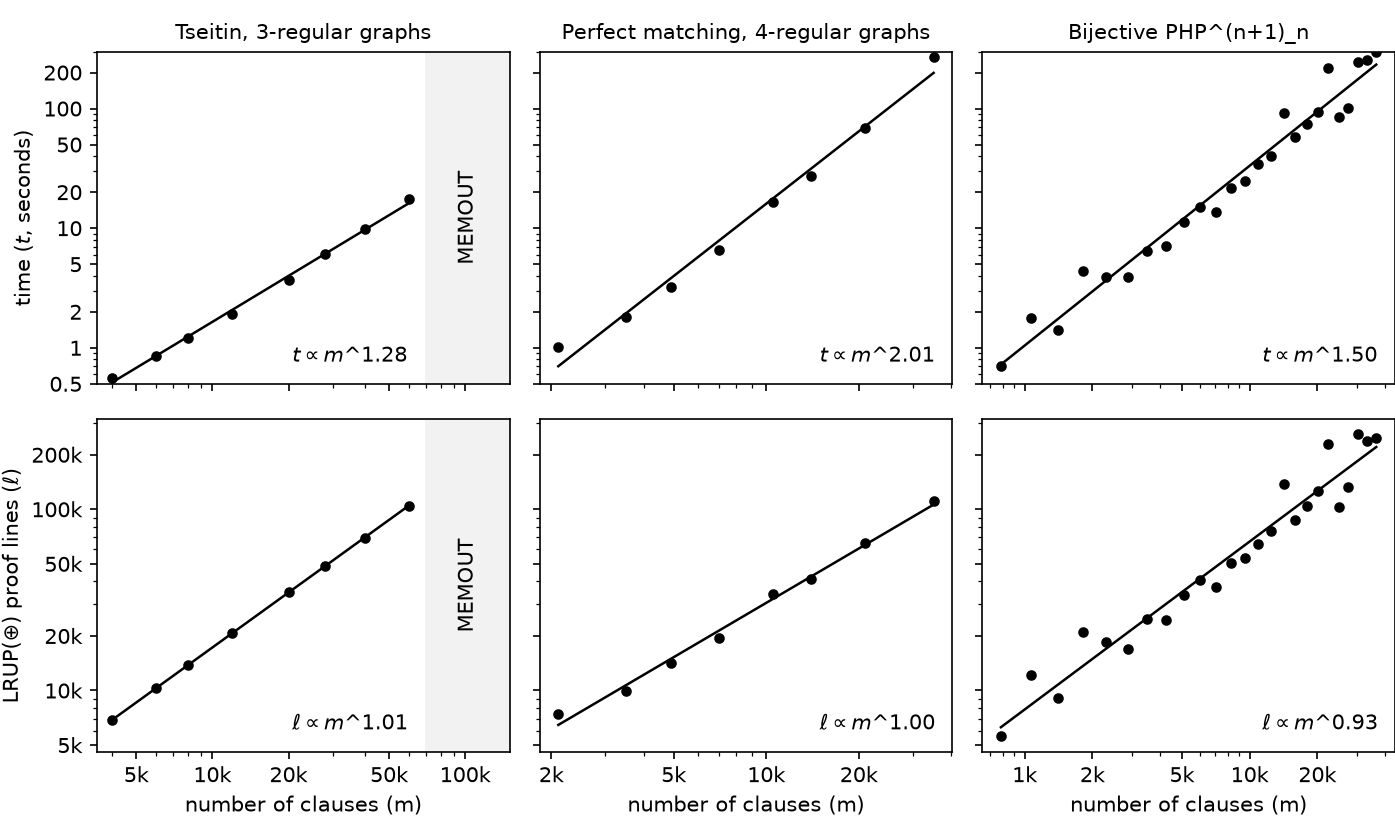}
    \caption{\solverName{} \texttt{{-}{-}plain} appears to have polynomially scaling performance and proof size on Tseitin formulas, perfect matching formulas, and bijective pigeonhole formulas with $n+1$ pigeons into $n$ holes. Note that the largest Tseitin formulas exceed the memory cap of 2 GB by reading the input alone, due to \solverName{}'s memory usage being generally quadratic in input size when using the standard dense bitvector representation of equations.}
    \label{fig:pc}
\end{figure}

Our last evaluation set consists of problems that are known to be hard for Resolution, but easy for stronger proof systems like $\Res(\oplus)$ \cite{DBLP:conf/mfcs/ItsyksonS14, DBLP:journals/apal/ItsyksonS20}. These include Tseitin formulas on random $k$-regular (expander) graphs \cite{DBLP:journals/jacm/Urquhart87}, perfect matching formulas on $k$-regular (expander) graphs with an odd number of vertices \cite{DBLP:journals/fuin/ItsyksonOSS16, DBLP:journals/theoretics/AustrinR22}, and bijective pigeonhole principle with $n+k$ pigeons into $n$ holes for odd $k$ (which are perfect matching formulas on complete bipartite graphs) \cite{DBLP:journals/tcs/Haken85, DBLP:journals/jcss/Razborov04}. Note that all of these instances are always written in CNF. Their linear structure comes from intermediate proof steps, not the input.

These instances are all generated with the \texttt{cnfgen} tool \cite{DBLP:conf/sat/LauriaENV17}. We shuffle the order of the clauses in each instance, because we noticed that the generator order gives \solverName{}'s CSIDS intersection heuristic an unusual advantage on Tseitin formulas, and thus the shuffled order is a fairer comparison. We run each solver on increasing values of $n$ until the solver fails to solve an instance within 300 seconds. The \texttt{{-}{-}plain} configuration of \solverName{} disables phase saving, vivification, implication graph reasoning, and BVE. 

The results are shown in \cref{tab:pc}. Other than \solverName{}, only CryptoMiniSat has noteworthy performance, quickly solving very large Tseitin formulas by pattern-matching the equations from their clauses and refuting the resulting system with Gaussian elimination, using no CDCL search at all. \solverName{} does not solve Tseitin formulas as quickly as CryptoMiniSat, but can solve them using only core $\CDCL(\oplus)$ search, without preprocessing or detecting any structure. Moreover, \solverName{} is the only solver that performs well on perfect matching formulas and bijective pigeonhole formulas with $n+1$ pigeons into $n$ holes, though all solvers struggle with more pigeons. 

We note that the specific structure of these instances causes \solverName{} \texttt{{-}{-}plain} to often perform better than \solverName{}'s default configuration. In \cref{fig:pc}, we show how \solverName{} \texttt{{-}{-}plain} appears to solve instances with running time at most quadratic in the input size (though the chart about perfect matching is not entirely conclusive), and $\LRUP(\oplus)$ proof length at most linear in the input size. (Note that one $\LRUP(\oplus)$ proof line can correspond to multiple $\Res(\oplus)$ proof lines. See \cref{ssec:sim-res}.) This strong performance provides empirical evidence that $\CDCL(\oplus)$, with our CSIDS decision heuristic and other basic features, could be an exponentially stronger core search algorithm compared to classical CDCL. 

% \subparagraph{A deeper dive into \texorpdfstring{\solverName{}}{Xorcle}'s performance on Tseitin formulas}

\section{Conclusion}

In this work, we introduced the $\CDCL(\oplus)$ algorithm for XNF SAT solving, and demonstrated how it is closely related to classical CDCL. A key contribution is the affine resolution rule, which replaces the classical resolution rule in all procedures, including 1-UIP clause learning. We proved that our translation of 1-UIP clause learning has the same backjump optimality property as the classical algorithm, developing a rich theory of UIPs, and connect $\CDCL(\oplus)$ with the proof complexity literature studying the $\Res(\oplus)$ proof system, showing a bidirectional simulation.

We introduced the CSIDS decision heuristic and adapted many more classical CDCL heuristics in our solver \solverName{}. These heuristics allow \solverName{} to be competitive with other solvers on most cryptographic benchmark families and leading overall. \solverName{} solves many problems that are hard for Resolution in what appears to be polynomial time, and we are interested in finding adjustments to heuristics that can close the remaining gap on bijective pigeonhole formulas with more than one extra pigeon. As we continue developing \solverName{} to translate more heuristics from recent advancements, we expect further improvements to overall performance. 

\solverName{} is not intended for general SAT solving, and would perform poorly in most domains represented in the SAT Competition. Despite this, we are excited that XNF solving may be a viable technique for applications with more linear and algebraic structure. In many such applications, we anticipate that research into improving XNF encodings may be just as important, if not more important, than overall solver improvement.

\newpage
\bibliography{bibliography}

\newpage

\appendix 
\crefalias{section}{appendix}

\section{Semantic characterizations of classical CDCL} \label{sec:classic}

In this section, we describe unit propagation, conflict clauses, asserting clauses, and 1-UIP clauses in ways that use semantic implication. These characterizations are likely folklore, but we have not found them explicitly collected in a single reference, so we show them here for the purpose of analogy.
While not necessary for any mathematical arguments in the main text, these equivalences show why our definitions of $\CDCL(\oplus)$ components are natural.

We use the following standard definitions of classical CDCL concepts~\cite{DBLP:conf/iccad/ZhangMMM01}. 
Unit propagation on a clause $C$ using a set $U$ of units learns the last literal of $C$ when all others are falsified by $U$. 
A run of unit propagation can be visualized as an implication graph, where vertices are the true units (both starting and deduced) and an edge $(x, y)$ means that $\lnot x$ is in the reason clause used to derive $y$. The following is a standard alternative characterization using semantic properties, which is mirrored in \cref{prop:semantic}.

\begin{proposition} Let $U$ be a trail, $\ell$ be a literal, and $C$ be a clause.
\begin{enumerate}
    \item (soundness) If $C$ propagates $\ell$ from $U$, then $C \land U \Rightarrow \ell$, including when $\ell = \bot$.
    \item (completeness) If $C \land U \Rightarrow \bot$, then $C$ propagates $\bot$ from $U$. Otherwise, whenever $C \land U \Rightarrow \ell$ and $U \not \Rightarrow \ell$, the clause $C$ propagates $\ell$ from $U$.
\end{enumerate}
\end{proposition}

A \emph{conflict cut} is a partition of vertices in the implication graph into one part with all decisions and one part with the conflict $\bot$. Given a cut, the associated conflict clause is $\bigvee \{\lnot x \mid (x, y) \text{ across the cut}\}$. While this definition based on implication graphs is common, there are many equivalent ways to define this in classical CDCL.

\begin{proposition} \label{prop:sem-cc}
    Let $U$ be a sequence of units ending in $\bot$ and let $\mathcal{R}$ be the corresponding reasons. Let $C$ be a clause falsified by $U$. The following are equivalent:
    \begin{enumerate}
        \item Unit propagation starting with units $\lnot C$ and clauses $\mathcal{R}$ produces a contradiction.
        \item The set $\lnot C \cup \mathcal{R}$ is unsatisfiable. 
        \item Every path from a decision to $\bot$ passes through at least one vertex in $\lnot C$.
        \item $C$ is a weakening of a conflict clause as defined via a conflict cut of the implication graph (i.e., a conflict clause with possibly extra literals). 
    \end{enumerate}
\end{proposition}

\begin{proof}
The implication $(4) \Rightarrow (1)$ is a standard fact, using the implication graph to replay unit propagation. Then, $(1) \Rightarrow (2)$ is also obvious by unit propagation being sound. For $(2) \Rightarrow (3)$, it is easy to see this by contrapositive. In particular, if there is a path from a decision to $\bot$ that does not use anything in $\lnot C$, choose any assignment that makes every literal in this path false and makes the remaining literals in $U$ true. Then $\lnot C$ is certainly satisfied, and each reason clause is satisfied by the predecessor of its corresponding unit being false, so $\lnot C \cup \mathcal{R}$ is satisfied.

To prove that $(3) \Rightarrow (4)$, it suffices to prove the case when $\lnot C$ satisfies (3) minimally by inclusion and, in this case, we claim that $C$ is actually exactly a conflict clause. We construct a conflict cut $(S, T)$ by taking $S$ to be all $\ell \in U$ for which there is a path from a decision to $\ell$ that does not contain anything in $\lnot C$, except possibly $\ell$ itself. Clearly $S$ contains the decisions and $T = U \setminus S$ contains $\bot$. We claim that $C$ is exactly the conflict clause corresponding to this cut, in other words, $\{x \mid (x, y) \text{ across } (S, T)\} = \lnot C$.

    For the $\subseteq$ direction, suppose that $\ell \in S$ has an edge to $\ell' \in T$. Because $\ell \in S$, there is a path from a decision to $\ell$ containing nothing in $\lnot C$, except possibly $\ell$ itself. But no such path exists to $\ell'$ despite the edge $(\ell, \ell')$, thus $\ell \in \lnot C$. 

    For the $\supseteq$ direction, we first note that $\lnot C \subseteq S$: if this were not true, then there would be some $\ell \in \lnot C \setminus S$ for which every path from decisions to $\ell$ contains something else from $\lnot C$, contradicting minimality (because $\ell$ could be removed from $\lnot C$). So, take any $\ell \in \lnot C$, and it remains to show that there is an outgoing edge to $T$. By minimality, there must be a path $P$ from decisions to $\bot$ whose only vertex in $\lnot C$ is $\ell$. It must go through an out-neighbor $\ell' \not \in \lnot C$ of $\ell$. Either $\ell' \in T$, so that we are done, or $\ell' \in S$. The path from a decision to $\ell'$ from the definition of $S$, combined with the subpath of $P$ between $\ell'$ and $\bot$, is thus a path from a decision to $\bot$ that does not use $\lnot C$, contradiction. 
\end{proof}

\begin{remark} When we generalize CDCL to $\CDCL(\oplus)$, not all alternative characterizations in \cref{prop:sem-cc} are suitable:
there is no implication graph, so (3) and (4) don't make sense, and, although (1) and (2) both make sense, they are actually not equivalent (\cref{ex:sem-confl}), so we choose the definition to be the generalization of (1).
\end{remark}

The level of a unit is the number of decisions made prior to and including itself. A conflict clause is \emph{asserting} if it has only one literal (whose negation is) at the last level, which we call the \emph{asserted literal}. A vertex in the implication graph is a \emph{unique implication point (UIP)} if all paths from the last decision to $\bot$ go through it. 

\begin{proposition} \label{prop:sem-uip}
Let $U$ be a trail with corresponding reasons $\mathcal{R}$. Let $L_{d-1}$ denote the units before the last decision level. The following are equivalent.
\begin{enumerate}
    \item A literal $\ell$ is a UIP as defined via the implication graph.
    \item The clause $C$ given by $\lnot C = L_{d-1} \cup \{\ell\}$ is a conflict clause, or in other words, by \cref{prop:sem-cc}, the set $L_{d-1} \cup \{\ell\} \cup \mathcal{R}$ is unsatisfiable.
    \item There exists an asserting clause $C$ that asserts $\lnot \ell$. 
\end{enumerate}
\end{proposition}

\begin{proof}
    The fact $(2) \Rightarrow (3)$ is immediate, and it is easy to see that $(3) \Rightarrow (1)$ because any path from the last decision to $\bot$ must pass through the conflict cut associated with the conflict clause $C$, and $\ell$ is the only literal in the last level whose outedges cross the cut, so the path must pass through $\ell$ and $\ell$ is a UIP.

    It remains to show that $(1) \Rightarrow (2)$. Suppose that $\ell$ is a UIP, and consider the cut $(S, T)$ where $T$ is the set of all literals on some path from $\ell$ to $\bot$, excluding $\ell$, and $S = U \setminus T$. We claim that $C$ given by $\lnot C = L_{d-1} \cup \{\ell\}$ is a weakening of the conflict clause corresponding to this cut, and it suffices to show that $\ell$ is the only literal in the last level whose outedges cross the cut.  We simply note that any literal $\ell' \in S$ in the last level with an outedge across the cut provides a way to access $\bot$ from the last decision while avoiding $\ell$, contradicting $\ell$ being a UIP, unless $\ell'$ itself is inaccessible without going through $\ell$, but in that case $\ell' \in T$, also a contradiction.  
\end{proof}

\begin{remark}
In \cref{ssec:backjump}, we generalize each of the three equivalent definitions above into different types of UIPs for $\CDCL(\oplus)$, and unlike the classical case, these end up describing three different notions. In particular, the implication-graph-based definition in (1) is most directly similar to a replay UIP, in that the implication graph gives a way to use the UIP and literals from previous levels to propagate $\bot$. The semantic condition in (2) is identical to our semantic UIPs. The asserting-clause-based definition in (3) corresponds to unit propagation UIPs, since asserting clauses are defined to propagate from previous levels. 
\end{remark}

\clearpage

\section{Effect of heuristics} \label{sec:comp-heur}

\begin{table}[H]
    \centering
    \caption{Contribution of each default-on heuristic to the overall performance of \solverName{} v0.2 on Xorricane's benchmark set \cite{DBLP:phd/dnb/Danner25, DBLP:journals/jair/DannerK26}, with a time limit of 300 seconds. For each instance, we take the median time over three random seeds for the solver. See \cref{ssec:eval} for methodology details. Statistics count timeouts as 600 seconds. Average PAR-2 refers to the mean of running times over all instances. Time ratio compares the listed configuration to the default, taking a geometric mean over instances, excluding those that either timed out in both configurations or were solved in under 10 seconds by both configurations (to reduce noise and startup effects). Note that configurations with similar PAR-2 may still have significant time ratio because PAR-2 includes all timeouts.}
    \small
    \begin{tblr}{
    hlines, vlines,
    hline{2} = {1}{-}{}, hline{2} = {2}{-}{},
    vline{-} = {2}{abovepos = 1},
    rowsep = 4pt, colsep = 4pt,
    colspec = {l *{4}{Q[c, m, wd=6.5em]}},
    cells   = {valign=m},
  }
    Configuration & Solved (of 350) & Gained / lost & Average PAR-2  & Time ratio \\
    \solverName{} default & 216 & -- & 246.4 s & -- \\
    {Replacing exact watchlist \\ with cyclic selector} & 158 & $+0$ / $-58$ & 337.8 s & 10.18$\times$ \\
    No clause deletion & 175 & $+0$ / $-41$ & 308.4 s & 4.61$\times$ \\
    No CSIDS intersection & 182 & $+0$ / $-34$ & 298.1 s & 3.95$\times$ \\
    No phase saving & 205 & $+1$ / $-12$ & 262.8 s & 1.55$\times$ \\
    No implication graph & 209 & $+4$ / $-11$ & 259.3 s & 1.84$\times$ \\
    No compaction & 210 & $+0$ / $-6$ & 255.5 s & 1.32$\times$ \\
    No clause minimization & 212 & $+4$ / $-8$ & 255.6 s & 1.49$\times$ \\
    No vivification & 213 & $+3$ / $-6$ & 252.6 s & 1.29$\times$ \\
    No OTF subsumption & 213 & $+5$ / $-8$ & 251.5 s & 1.15$\times$ \\
    No BVE & 214 & $+1$ / $-3$ & 248.5 s & 1.02$\times$ \\
  \end{tblr}
    
    \label{tab:heur}
\end{table}

\end{document}